\documentclass{article}

\usepackage[english]{babel}
\usepackage{hyperref,textcomp,amsmath,amsfonts,amssymb,mathrsfs,amsthm,enumerate,bm,bbm,verbatim}

\title{A New Method and a New Scaling For Deriving Fermionic Mean-field Dynamics}

\author{
S\"oren Petrat\footnote{Institute of Science and Technology Austria (IST Austria), Am Campus 1, 3400 Klosterneuburg, Austria. E-mail: soeren.petrat@ist.ac.at}\ \ and
Peter Pickl\footnote{Mathematisches Institut, Ludwig-Maximilians-Universit\"at, Theresienstr.\ 39, 80333 M\"unchen, Germany. E-mail: pickl@math.lmu.de}
}

\theoremstyle{plain}\newtheorem{theorem}{Theorem}[section]
\theoremstyle{plain}\newtheorem{lemma}[theorem]{Lemma}
\theoremstyle{plain}
\theoremstyle{plain}\newtheorem{assumption}[theorem]{Assumption}
\theoremstyle{plain}\newtheorem{proposition}[theorem]{Proposition}
\theoremstyle{definition}\newtheorem{definition}[theorem]{Definition}

\allowdisplaybreaks[1]

\newcommand{\scp}[2]{\left\langle #1 , #2 \right\rangle}
\newcommand{\SCP}[2]{\langle #1 , #2 \rangle}
\newcommand{\bigSCP}[2]{\Big\langle #1 , #2 \Big\rangle}
\newcommand{\bra}[1]{\langle #1 |}

\newcommand{\ket}[1]{| #1 \rangle}

\newcommand{\ketbra}[2]{| #1 \rangle \langle #2 |}

\newcommand{\ketbr}[1]{| #1 \rangle \langle #1 |}

\newcommand{\norm}[2][]{\left|\left| #2 \right|\right|_{#1}}
\newcommand{\Hilbert}{\mathscr{H}}
\newcommand{\Span}{\mathrm{span}}

\renewcommand{\Im}{\mathrm{Im}}
\newcommand{\AAA}{\mathcal{A}}
\newcommand{\QQQ}{\mathbb{Q}}
\newcommand{\RRR}{\mathbb{R}}
\newcommand{\CCC}{\mathbb{C}}
\newcommand{\NNN}{\mathbb{N}}
\newcommand{\ZZZ}{\mathbb{Z}}

\newcommand{\id}{\mathbbm{1}}
\newcommand{\floor}[1]{\left\lfloor #1 \right\rfloor}
\newcommand{\dir}{\mathrm{dir}}
\newcommand{\exch}{\mathrm{exch}}
\newcommand{\tr}{\mathrm{tr}}
\newcommand{\HS}{\mathrm{HS}}
\newcommand{\op}{\mathrm{op}}
\newcommand{\sym}{\mathrm{sym}}
\newcommand{\as}{\mathrm{as}}
\newcommand{\kin}{\mathrm{kin}}

\newcommand{\Var}{\mathrm{Var}}
\newcommand{\mf}{\mathrm{mf}}
\newcommand{\Term}{\mathrm{Term}}

\newcommand{\bigO}{\mathcal{O}}

\newcommand{\eqexp}[1]{\text{\scriptsize [#1]~~~}}
\newcommand{\eqexpl}[1]{\text{\scriptsize [#1]\!\!\!\!\!\!\!\!\!\!}}

\newcounter{remarks}
\begin{document}

\maketitle

\begin{abstract}
We introduce a new method for deriving the time-dependent Hartree or Hartree-Fock equations as an effective mean-field dynamics from the microscopic Schr\"odinger equation for fermionic many-particle systems in quantum mechanics. The method is an adaption of the method used in \cite{pickl:2011method} for bosonic systems to fermionic systems. It is based on a Gronwall type estimate for a suitable measure of distance between the microscopic solution and an antisymmetrized product state. We use this method to treat a new mean-field limit for fermions with long-range interactions in a large volume. Some of our results hold for singular attractive or repulsive interactions. We can also treat Coulomb interaction assuming either a mild singularity cutoff or certain regularity conditions on the solutions to the Hartree(-Fock) equations. In the considered limit, the kinetic and interaction energy are of the same order, while the average force is subleading. For some interactions, we prove that the Hartree(-Fock) dynamics is a more accurate approximation than a simpler dynamics that one would expect from the subleading force. With our method we also treat the mean-field limit coupled to a semiclassical limit, which was discussed in the literature before, and we recover some of the previous results. All results hold for initial data close (but not necessarily equal) to antisymmetrized product states and we always provide explicit rates of convergence.
\end{abstract}

\textbf{MSC class:} 35Q40, 35Q55, 81Q05, 82C10


\tableofcontents

\section{Introduction}
The behavior of an interacting many-body system in classical or quantum mechanics can be very complicated and the microscopic equations governing its behavior are usually practically impossible to solve for more than three or four particles. For large interacting systems it is therefore essential to use a statistical description in order to make statements about the typical behavior of a large number of particles. One type of such a statistical description is to approximate the microscopic dynamics by an effective one-body dynamics, i.e., replace the microscopic evolution equation with very many degrees of freedom by a simpler, usually non-linear equation with very few degrees of freedom. Famous examples are the Boltzmann, Navier-Stokes and Vlasov equations in classical mechanics, and the Hartree, Hartree-Fock and Gross-Pitaevskii equations in quantum mechanics; different regimes can lead to different effective evolution equations, e.g., kinetic equations or mean-field equations. Apart from the study of these equations and their interesting properties and consequences, it is an ongoing project of mathematical physics to \emph{derive} them from the microscopic dynamics, i.e., to prove rigorously that the solutions to the effective equations approximate the solutions to the microscopic equations well in certain situations. Only some cases of such a rigorous derivation are known. In classical mechanics, this could for example be shown for the Vlasov equation \cite{hepp:1977}, however, for the Boltzmann equation it has been shown only for very short times \cite{lanford:1975}, and for the Navier-Stokes equation it is still an open problem (see \cite{spohn:1991} for an excellent overview). In quantum mechanics, starting in the 70s and by a series of recent works, the derivation of effective dynamics for bosons near a condensate is well understood, see \cite{hepp:1974,spohn:1980,erdoes:2001,froehlich:2009,rodnianski:2009,pickl:2011method,pickl:2010hartree} for the case of the Hartree equation and \cite{erdoes:2006,erdoes:2007,erdoes:2007_2,erdoes:2009,erdoes:2010,pickl:2010gp_pos,pickl:2010gp_ext,benedikter:2012} for the Gross-Pitaevskii equation. Derivations of mean-field dynamics for fermions have been discussed in \cite{narnhofer:1981,spohn:1981,bardos:2003,erdoes:2004,froehlich:2011,benedikter:2013,benedikter:2014}; in particular \cite{benedikter:2013} treats an interesting scaling limit very comprehensively. In general, the Fermi statistics makes the situation more complicated, see the discussion in Sections~\ref{sec:density_O_N} and \ref{sec:density_O_1}.

For fermionic many-particle systems in quantum mechanics, the microscopic evolution is given by the Schr\"odinger equation (we set $\hbar=1=2m$ throughout this work)
\begin{equation}\label{Schr_intro}
i \partial_t \psi^t = H \psi^t
\end{equation}
for antisymmetric complex-valued $N$-particle wave functions $\psi^t \in L^2(\RRR^{3N})$ (for simplicity, we neglect spin and always work in three dimensions). Antisymmetry means that $\psi^t(\ldots,x_j,\ldots,x_k,\ldots) = - \psi^t(\ldots,x_k,\ldots,x_j,\ldots)$ $\forall j \neq k$. We consider Hamiltonians
\begin{equation}\label{Schr_intro_H}
H = \sum_{j=1}^N H_j^0 + \sum_{i<j} v^{(N)}(x_i-x_j),
\end{equation}
where $H_j^0$ acts only on $x_j$ and $v^{(N)}(x)=v^{(N)}(-x)$ is a real-valued pair-interaction potential (the superscript $(N)$ denotes a possible scaling with the particle number $N$ which is explained below). According to \eqref{Schr_intro}, the unitary time evolution of an initial wave function $\psi^0$ is given by $\psi^t = e^{-iHt}\psi^0$ if $H$ is self-adjoint which we henceforth assume. Note that for antisymmetric initial conditions $\psi^0$, the wave function $\psi^t$ remains antisymmetric under the Schr\"odinger evolution \eqref{Schr_intro} with Hamiltonian \eqref{Schr_intro_H} for all times. For the desired mean-field description, consider $N$ orthonormal one-particle wave functions (also called orbitals) $\varphi_1^t, \ldots, \varphi_N^t \in L^2(\RRR^3)$ which are solutions to the \emph{fermionic Hartree equations} (sometimes called \emph{reduced Hartree-Fock equations}). These are the coupled system of non-linear differential equations
\begin{equation}\label{hartree_intro}
i \partial_t \varphi_j^t = H^0 \varphi_j^t + \Big(v^{(N)} * \rho_N^t \Big) \, \varphi_j^t,
\end{equation}
for $j=1,\ldots,N$, where $*$ denotes convolution, and $\rho_N^t = \sum_{i=1}^N |\varphi_i^t|^2$ is the spatial density. Note that for orthonormal initial conditions $\varphi_1^0, \ldots, \varphi_N^0$, \eqref{hartree_intro} preserves the orthonormality for all times. The term
\begin{equation}\label{intro_direct}
\Big( v^{(N)} * \rho_N^t \Big)(x) = \int_{\RRR^3 } v^{(N)}(x-y) \rho_N^t(y) \, d^3y
\end{equation}
is called the \emph{mean-field}. It can be viewed as the average value of the interaction potential at point $x$, created by particles distributed according to the density $\rho_N^t$. Note that closely related effective equations for fermions are the Hartree-Fock equations, where an additional exchange term
\begin{equation}\label{exch_intro}
- \sum_{k=1}^N \Big(v^{(N)} * ({\varphi_k^t}^*\varphi_j^t) \Big)(x) \, \varphi_k^t(x)
\end{equation}
is present on the right-hand side of \eqref{hartree_intro}. In general, the Hartree-Fock equations are expected to be a better approximation than the fermionic Hartree equations. However, the exchange term is always smaller than the direct term \eqref{intro_direct}, and, as we see later, is negligibly small for the scaling limits that we consider in the sense that including the exchange term does not improve our error estimates. It would be interesting to study other scaling limits where taking into account the exchange term gives a more accurate approximation.

Now suppose that some initial $\varphi_1^0, \ldots, \varphi_N^0$ are given. Let the initial $N$-particle wave function be $\psi^0 \approx \bigwedge_{j=1}^N \varphi_j^0$, where 
\begin{equation}
\Bigg(\bigwedge_{j=1}^N \varphi_j\Bigg)(x_1,\ldots,x_N) = (N!)^{-1/2} \sum_{\sigma \in S_N} (-1)^{\sigma} \prod_{j=1}^N \varphi_{\sigma(j)}(x_j),
\end{equation}
is the antisymmetrized product of $\varphi_1, \ldots, \varphi_N$, with $S_N$ the symmetric group and $(-1)^\sigma$ the sign of the permutation $\sigma$. Then, under the Schr\"odinger evolution \eqref{Schr_intro}, this initial wave function evolves to $\psi^t = e^{-iHt} \psi^0$. We want to compare this $\psi^t$ to the wave function $\bigwedge_{j=1}^N \varphi_j^t$, where the $\varphi_j^t$ are the solutions to the fermionic Hartree equations \eqref{hartree_intro}. In other words, if still $\psi^t \approx \bigwedge_{j=1}^N \varphi_j^t$ at some time $t$, then the Schr\"odinger dynamics is approximated well by the Hartree dynamics. To show such a statement is the goal of this article.

Note that in the presence of an interaction potential $v^{(N)}$ it is in general never true that $e^{-iHt} \bigwedge_{j=1}^N \varphi_j^0 = \bigwedge_{j=1}^N \varphi_j^t$, since the interaction leads to correlations between the particles, i.e., the wave function evolves into a superposition of many antisymmetrized product states. These correlations are caused by deviations from the mean-field behavior, i.e., by fluctuations around the mean-field. For a rigorous derivation we thus need to consider a regime where the fluctuations become small for large $N$, on the relevant time scales. Such an appropriate limit $N\to\infty$ is then called a mean-field limit.

\subsection{Density $\propto N$ Regime}\label{sec:density_O_N}
One such regime describes a system of fermions with very high density proportional to $N$. The microscopic wave function $\psi^t(x_1,\ldots,x_N)$ is a solution to the Schr\"odinger equation (here for non-relativistic particles)
\begin{equation}\label{Schr_scaled_sc}
i N^{-1/3} \partial_t \psi^t = \left( \sum_{j=1}^N \left( -N^{-2/3} \Delta_{x_j} + w^{(N)}(x_j) \right) + N^{-1} \sum_{i<j} v(x_i-x_j) \right) \psi^t,
\end{equation}
and $\varphi_1^t,\ldots,\varphi_N^t$ are solutions to the corresponding fermionic Hartree equations
\begin{equation}\label{hartree_scaled_sc}
i N^{-1/3} \partial_t \varphi_j^t = \left( -N^{-2/3} \Delta + w^{(N)} + N^{-1} \left(v * \rho_N^t \right) \right) \varphi_j^t,
\end{equation}
for $j=1,\ldots,N$ (recall $\rho_N^t=\sum_{i=1}^N |\varphi_i^t|^2$), where $w^{(N)}$ is a real-valued external field. The evolution \eqref{Schr_scaled_sc} is considered for initial data in a volume of order $1$\footnote{Here, the phrase that a quantity is ``of order x'' means that it is bounded from below and above by some constant times $x$. It is not meant in the sense of the Landau symbol $\bigO(x)$, where a function $f(x)$ is $\bigO(g(x))$, if there is a constant $C$ such that $|f(x)| \leq C |g(x)|$ for $x$ large enough.} which can be realized if the particles are confined by a nice external trapping potential $w^{(N)}$. Note, that due to the antisymmetry of $\psi^t$, the total kinetic energy for particles in a volume of order $1$ is always bigger or equal $C N^{5/3}$ (which can, e.g., be read off from the kinetic energy inequality in Section~\ref{sec:energy inequalities}). For states close to the ground state and nice external fields, the kinetic energy is of order $N^{5/3}$ (e.g., for free particles in a box of constant side length, this can easily be checked). Then, the kinetic term on the right-hand side of \eqref{Schr_scaled_sc} is of order $N^{-2/3} N^{5/3} = N$. The interaction term is of the same order, since $N^2$ terms are scaled down with $N^{-1}$.

Let us describe heuristically two ways of how to arrive at the scaling limit of \eqref{Schr_scaled_sc} (and let us for simplicity disregard external fields here). One way is to start with initial conditions in a volume of order $1$ and kinetic energy of order $N^{5/3}$ and then add a scaling for the interaction, such that the scaled interaction and the kinetic energy are of the same order. The Schr\"odinger equation is then
\begin{equation}\label{Schr_scaled_sc_deriv}
i \partial_{\tilde{t}} \psi^{\tilde{t}} = \left( -\sum_{j=1}^N \Delta_{x_j} + N^{-1/3} \sum_{i<j} v(x_i-x_j) \right) \psi^{\tilde{t}}.
\end{equation}
Since the average momentum per particle is of order $N^{1/3}$, but the average scaled interaction or force per particle is of order $N^{-1/3}N = N^{2/3}$, one can, for large $N$, expect a nice limiting equation only for short times $\tilde{t}$ of order $N^{-1/3}$ (change in momentum $=$ force $\times$ time). By introducing $t = N^{1/3} \tilde{t}$ (and then dividing both sides of the equation by $N^{2/3}$), we get Equation~\eqref{Schr_scaled_sc}. Let us describe another, more physical way to arrive at Equation~\eqref{Schr_scaled_sc}, which, however, only works for Coulomb interaction. Consider initial data in a small volume of order $N^{-1}$ and kinetic energy of order $N^{7/3}$. This is realized, e.g., for electrons in the ``core'' region of a large-$Z$ atom, see \cite{lieb:1981}. The unscaled Schr\"odinger equation for Coulomb interacting particles is
\begin{equation}\label{Schr_scaled_sc_deriv2}
i \partial_{\tilde{t}} \psi^{\tilde{t}} = \left( -\sum_{j=1}^N \Delta_{\tilde{x}_j} + \sum_{i<j} |\tilde{x}_i-\tilde{x}_j|^{-1} \right) \psi^{\tilde{t}}.
\end{equation}
A calculation for the free ground state in a box shows that the average force per particle is of order $N^{5/3}$. Since the average momentum is of order $N^{2/3}$, from ``change in momentum $=$ force $\times$ time'' we can expect nice behavior for large $N$ for times of order $N^{-1}$. If we now introduce rescaled variables $t = N \tilde{t}$ and $x = N^{1/3} \tilde{x}$ (such that in the variables $t,x$ we consider the same kind of initial conditions as in \eqref{Schr_scaled_sc}), we arrive at Equation~\eqref{Schr_scaled_sc}.

Note that Equation~\eqref{Schr_scaled_sc} has a semiclassical structure due to the $N^{-1/3}$ factors, which play the role of a very small parameter, like the $\hbar$ in the Schr\"odinger equation with units. As a consequence, the solutions to the Schr\"odinger equation \eqref{Schr_scaled_sc} are close to solutions to the classical Vlasov equation, e.g., in the sense that the Wigner transform of a solution to \eqref{Schr_scaled_sc} is close to a classical phase space density $\rho^t(x,p)$ that solves the Vlasov equation. Such a derivation of the Vlasov equation from the microscopic Schr\"odinger equation has been considered in \cite{narnhofer:1981} and improved in \cite{spohn:1981}. Several other works \cite{lions:1993,markowich:1993,gasser:1998,pezzotti:2009,athanassoulis:2011,amour:2013,amour:2013_2,benedikter:2015} also study the derivation of the Vlasov equation starting from the Hartree equation \eqref{hartree_scaled_sc} or the Hartree-Fock equation. However, the mean-field dynamics \eqref{hartree_scaled_sc} is a better approximation to the dynamics \eqref{Schr_scaled_sc} than the Vlasov dynamics. A derivation of the mean-field dynamics \eqref{hartree_scaled_sc} from the microscopic dynamics \eqref{Schr_scaled_sc} has first been given in \cite{erdoes:2004} for bounded analytic interactions and short times. In \cite{benedikter:2013,benedikter:2014}, the derivation was crucially improved in the sense that it was shown to hold for all times, for initial data with a certain semiclassical structure. Furthermore, the result holds for pseudo-relativistic free Hamiltonians and with fewer regularity assumptions on the interaction. The interaction is, however, still assumed to be bounded. Since many applications concern electrons and in light of the discussion around Equation~\eqref{Schr_scaled_sc_deriv2}, a derivation for Coulomb interaction would be desirable. Let us also mention \cite{benedikter:2015_2} where the analysis of \cite{benedikter:2013} is extended to fermionic mixed states.

\subsection{Density $\propto 1$ Regime}\label{sec:density_O_1}
Let us now introduce a new scaling limit for fermions. The inspiration for this new limit is twofold. On the one hand, we want to study a scaling limit, where the solution to the Schr\"odinger equation is not close to a solution to the Vlasov equation. In such a limit, one would see more quantum effects like interference of wave packets. On the other hand, we are inspired by the application of the Hartree-Fock equations to large molecules, where the size of the system is of order $N$. This is the case for the scaling limit we now introduce.

We consider initial conditions in a volume of order $N$, i.e., with average density of order $1$. The idea is that for such initial data a mean-field approximation is valid for \emph{long-range interactions} like Coulomb interaction. The microscopic evolution for the wave function $\psi^t(x_1,\ldots,x_N)$ is given by the Schr\"odinger equation (for non-relativistic particles)
\begin{equation}\label{scaling_Schr}
i \partial_t \psi^t = \left( \sum_{j=1}^N \left( -\Delta_{x_j} + w^{(N)}(x_j) \right) + N^{-\beta} \sum_{i<j} v(x_i-x_j) \right) \psi^t,
\end{equation}
and the corresponding fermionic Hartree equations are
\begin{equation}\label{scaling_hartree}
i \partial_t \varphi_j^t = \left( -\Delta + w^{(N)} + N^{-\beta} \left(v * \rho_N^t \right) \right) \varphi_j^t,
\end{equation}
for $j=1,\ldots,N$, $\rho_N^t=\sum_{i=1}^N |\varphi_i^t|^2$, \emph{scaling exponent} $\beta \in \RRR$ and $v(x)=v(-x)$ a real-valued long-range interaction. The Schr\"odinger equation \eqref{scaling_Schr} is interesting for initial data with total kinetic energy of order $N$. (Note that by the kinetic energy inequality (see Section~\ref{sec:energy inequalities}), this implies that the system volume is at least of order $N$.) Now the crucial observation is that the total unscaled interaction energy for long-range $v$ is not of order $N^2$, as one might expect from the order $N^2$ terms in the double sum $\sum_{i<j}$ in \eqref{scaling_Schr}. It is in fact smaller, namely of order $N^{1+\tilde{\beta}}$, with some $0\leq \tilde{\beta} \leq 1$ depending on the long-range behavior of $v$. This can be seen heuristically by considering the mean-field for constant density and interactions $|x|^{-s}$, since ($V_N$ denotes a volume of order $N$)
\begin{equation}\label{heuristics_mf}
\big(|\cdot|^{-s} * \rho_N^t\big) \approx C \int_{V_N} |x|^{-s} \,d^3x \approx C \int_0^{N^{1/3}} r^{-s}\, r^2 dr \propto N^{1-s/3} =: N^{\tilde{\beta}},
\end{equation}
for appropriate $s>0$ (see Lemma~\ref{lem:scaling_x-s}). Thus, if we set $\beta=\tilde{\beta}$ in Equation~\eqref{scaling_Schr}, the kinetic term and the interaction term are both of order $N$.

Let us now consider the average force per particle. For simplicity, let us discuss it heuristically in the mean-field approximation (a similar consideration can be done for the microscopic wave function $\psi^t$). There, the average force, including the scaling with $\tilde{\beta}=1-s/3$, is given by the gradient of the mean-field, i.e., it is of the order
\begin{equation}\label{heuristics_mf_force}
N^{-\tilde{\beta}} \nabla \big(|\cdot|^{-s} * \rho_N^t\big) \approx N^{-\tilde{\beta}} \big(|\cdot|^{-s-1} * \rho_N^t\big) \approx C N^{-\tilde{\beta}}  \int_0^{N^{1/3}} r^{-s-1}\, r^2 dr \propto N^{-\tilde{\beta}} N^{2/3-s/3} = N^{-1/3}.
\end{equation}
Thus, the mean-field is almost constant on scales of order $1$ and only varies over the whole system size. One would therefore heuristically expect free evolution to leading order. This classically inspired heuristics is actually only almost right, since orbitals which are spread over the whole system should feel an effect coming from the mean-field. We expect, however, that this could be easily taken care of by adding a time and space dependent phase to the freely evolving orbitals. Those orbitals can be expected to approximate the Schr\"odinger equation \eqref{scaling_Schr}. However, similar to before, one would expect that the fermionic Hartree equations \eqref{scaling_hartree} provide a more accurate approximation to the dynamics. In fact, we prove this in Section~\ref{sec:main_theorem_mf_x-s} for $0<s<3/5$, where we also discuss relevant observables.

Note that one can repeat the heuristic calculation \eqref{heuristics_mf_force} and ask for the variation of the force. One finds, e.g., that
\begin{equation}\label{heuristics_mf_force_variation}
N^{-\tilde{\beta}} \Delta \big(|\cdot|^{-s} * \rho_N^t\big) \approx N^{-\tilde{\beta}} \big(|\cdot|^{-s-2} * \rho_N^t\big) \approx C N^{-\tilde{\beta}} \int_0^{N^{1/3}} r^{-s-2}\, r^2 dr \propto N^{-\tilde{\beta}} N^{1/3-s/3} = N^{-2/3},
\end{equation}
for $s<1$, and
\begin{equation}\label{heuristics_mf_force_variation_s_1}
N^{-\tilde{\beta}} \Delta_x \big(|\cdot|^{-1} * \rho_N^t\big)(x) = N^{-\tilde{\beta}} \big(\Delta |\cdot|^{-1} * \rho_N^t\big)(x) = N^{-\tilde{\beta}} \rho_N^t(x) \approx N^{-\tilde{\beta}} \approx N^{-2/3}.
\end{equation}
for Coulomb interaction (where one would actually expect an $\ln N$ correction to the variation of the force). Thus, while the leading order in this scaling limit is free dynamics, the next to leading order is a constant force (at least when taking appropriate phase factors into account, see the discussion after Equation~\eqref{heuristics_mf_force}).

For interactions $|x|^{-s}$, we can compare the new scaling limit to the one from Section~\ref{sec:density_O_N}. Let us choose $s=1$ here. The new scaling limit for initial conditions $\tilde{\psi}^0$ in a volume of order $N$ is then
\begin{equation}\label{scaling_Schr_again}
i \partial_{\tilde{t}} \tilde{\psi}^{\tilde{t}} = \left( -\sum_{j=1}^N \Delta_{\tilde{x}_j} + N^{-2/3} \sum_{i<j} |\tilde{x}_i-\tilde{x}_j|^{-1} \right) \tilde{\psi}^{\tilde{t}}.
\end{equation}
If we rescale the spatial variables such that the initial conditions are in a volume of order $1$, i.e., we introduce $x=N^{-1/3}\tilde{x}$, then the rescaled wave function $\psi^{\tilde{t}}$ solves
\begin{equation}\label{scaling_Schr_again_resc}
i \partial_{\tilde{t}} \psi^{\tilde{t}} = \left( -N^{-2/3} \sum_{j=1}^N \Delta_{\tilde{x}_j} + N^{-1} \sum_{i<j} |\tilde{x}_i-\tilde{x}_j|^{-1} \right) \psi^{\tilde{t}}.
\end{equation}
This is similar to Equation~\eqref{Schr_scaled_sc}, but on a different time scale. The time scales are related by $\tilde{t} = N^{1/3} t$. In other words, if we formulate the new scaling limit for initial conditions in a volume of order $1$, then the time scales are much shorter than those considered in Section~\ref{sec:density_O_N}. The important difference between the two scaling limits lies in the kind of initial conditions one would naturally consider. In \eqref{Schr_scaled_sc}, the initial conditions would vary over spatial scales of order $1$, i.e., the whole system size, and not over scales of order, say $N^{-1/3}$. This is also assumed in \cite{erdoes:2004,benedikter:2013}. On the other hand, the initial conditions for \eqref{scaling_Schr} would naturally vary over scales of order $1$. If we would rescale those initial conditions to a volume of size $1$, then they would have a small scale structure on spatial scales of order $N^{-1/3}$. In fact, in our main results about the new scaling limit, we have no restrictions for the initial conditions, except that the kinetic energy is appropriately bounded by $CN$.

The described regime, to our knowledge, has not been considered in the literature before for a derivation of mean-field dynamics. Note, however, that \cite{bardos:2003,bardos:2004,bardos:2007,froehlich:2011} consider the case $\beta=1$ which, for Coulomb interaction, leads to an interaction term of order $N^{-1/3}$. Compared to that, the results in this article are an improvement by a factor $N^{1/3}$ in the interaction strength. However, in \cite{froehlich:2011} Coulomb interaction without any cutoff is considered, while we have to introduce a very mild cutoff on scales much shorter than the average particle distance, in order to treat the Coulomb singularity. Note that in \cite{bach:2015}, by using the method introduced in this article and additional techniques like the Fefferman-de la Llave decomposition of the Coulomb potential, also Coulomb interaction without cutoff can be dealt with.

\subsection{Outline}
Next, we present the main results of this article. The results for the new scaling \eqref{scaling_Schr}, which are the focus of this work, are presented and discussed in Section~\ref{sec:main_theorem_mf_x-s}. In Section~\ref{sec:estimates_semiclassical}, we derive the mean-field equations \eqref{hartree_scaled_sc} for a class of bounded interactions and initial states with a semiclassical structure. For this case, we reproduce some of the results from \cite{benedikter:2013}, with minor improvements on the initial conditions. We give this derivation in order to demonstrate the generality of our approach and since the proof is very short. In Section~\ref{sec:counting_functional}, we introduce our method for deriving fermionic mean-field dynamics. In Section~\ref{sec:dens_mat_summary}, we explain the connection between the functional we use in the proof and reduced density matrices. In Section~\ref{sec:main_theorem_mf_general_v}, we state a slightly more general version of our main result in terms of the newly defined functional and in Section~\ref{sec:sketch_of_proof} we sketch the proof of these theorems. All proofs are given in Sections \ref{sec:notation} to \ref{sec:proof_sc_scaling}. Finally, let us remark that the length of this article is due to the fact that we introduce our method in great detail and generality, in particular in Sections \ref{sec:notation} to \ref{sec:estimates_projectors}. The core parts of the proof are given in Sections \ref{sec:alpha_m_dot_rigorous} and \ref{sec:proof_sc_scaling}.

\section{Main Results}\label{sec:main_results}
In order to state the main results of this article, we need a measure for the closeness of a many-body wave function $\psi \in L^2(\RRR^{3N})$ to the antisymmetrized product state $\bigwedge_{j=1}^N\varphi_j$, with $\varphi_1,\ldots,\varphi_N \in L^2(\RRR^3)$. A natural choice is a trace norm estimate on the reduced density matrices. For any normalized antisymmetric $\psi \in L^2(\RRR^{3N})$, the reduced one-particle density matrix is defined by its integral kernel
\begin{equation}\label{definition_dens_mat_one_part}
\gamma_1^{\psi}(x;y) = \int \psi(x,x_2,\ldots,x_N) \psi^*(y,x_2,\ldots,x_N) \,d^3x_2 \ldots d^3x_N.
\end{equation}
If we want to have control over the statistics of one-body observables $A:L^2(\RRR^3)\to L^2(\RRR^3)$, we need to control the trace norm difference of the reduced one-particle density matrices of $\psi$ and $\bigwedge \varphi_j$, since, e.g., for bounded $A$,
\begin{equation}
\tr A \gamma^{\psi}_1 - \tr A \gamma^{\bigwedge \varphi_j}_1 \leq \norm[\op]{A} \Big|\Big| \gamma^{\bigwedge \varphi_j}_1 - \gamma^{\psi}_1 \Big|\Big|_{\tr},
\end{equation}
where $\norm[\op]{\cdot}$ denotes the operator norm and $\norm[\tr]{\cdot}$ the trace norm (see also Section~\ref{sec:density_matrices}). Thus, we express our main results in terms of the trace norm difference above. Note that we actually prove slightly stronger convergence statements, see Section~\ref{sec:counting_functional}.

\subsection{Density $\propto 1$ Regime with Interactions $|x|^{-s}$}\label{sec:main_theorem_mf_x-s}
In this section we explicitly consider the non-relativistic Schr\"odinger equation \eqref{scaling_Schr} and the corresponding fermionic Hartree equations \eqref{scaling_hartree}, as discussed in Section~\ref{sec:density_O_1}. The results in this subsection are concerned with interaction potentials $|x|^{-s}$ with $0<s<6/5$, sometimes with singularity weakened or cutoff, and the corresponding $\beta=1-s/3$ (see also Lemma~\ref{lem:scaling_x-s}). For the following results we assume the existence of solutions to the equations \eqref{scaling_Schr} and \eqref{scaling_hartree} and that the total kinetic energy of the solutions to \eqref{scaling_hartree} (but not necessarily to \eqref{scaling_Schr}) is bounded by $AN$, i.e., $\sum_{i=1}^N ||\nabla \varphi_i^t||^2 \leq AN$. There are several works about solution theory to the Hartree(-Fock) equations \cite{bove:1974,chadam:1975,chadam:1976,bove:1976} which establish existence and uniqueness of solutions even for singular (attractive or repulsive) interactions and external fields like $|x|^{-1}$. A blowup of solutions is only expected for strong attractive interactions (e.g., for gravitating fermions) with semirelativistic free Hamiltonian, see, e.g., \cite{froehlich:2007,hainzl:2009,hainzl:2010}. (Indeed, as is shown in \cite{froehlich:2007}, even for bounded interactions, $v*\rho_N^t$ can become infinite in the limit of $t\to T$, for some $T < \infty$.) Therefore, for non-relativistic Hamiltonians, due to the conservation of the total Hartree energy, $\sum_{i=1}^N ||\nabla \varphi_i^t||^2 \leq AN$ always holds if it holds for the initial states $\varphi_1^0, \ldots, \varphi_N^0$ and if the external field is nice enough (e.g., scaled external Coulomb fields generated by nuclei with some $N$-independent distances to each other are ok).

\begin{theorem}\label{thm:E_kin_only}
Let $t\in[0,T)$ for some $0<T\in \RRR\cup\infty$. Let $\psi^t \in L^2(\RRR^{3N})$ be a solution to the Schr\"odinger equation \eqref{scaling_Schr} with antisymmetric initial condition $\psi^0 \in L^2(\RRR^{3N})$. Let $\varphi_1^t,\ldots,\varphi_N^t \in L^2(\RRR^3)$ be solutions to the fermionic Hartree equations \eqref{scaling_hartree} with orthonormal initial conditions $\varphi_1^0,\ldots,\varphi_N^0 \in L^2(\RRR^3)$, and with
\begin{equation}\label{E_kin_AN}
\sum_{i=1}^N \norm{\nabla \varphi_i^t}^2 \leq AN
\end{equation}
for some $A>0$ and all $t<T$. Then there are positive constants $C$, such that
\begin{enumerate}[(a)]
\item\label{case_a} for interactions
\begin{equation}\label{int_first_case}
v(x) = \pm |x|^{-s},
\end{equation}
with $0<s<3/5$ and $\beta=1-s/3$ we have
\begin{equation}\label{main_alpha_ineq_n_applied1}
\Big|\Big| \gamma^{\bigwedge \varphi_j^t}_1 - \gamma^{\psi^t}_1 \Big|\Big|_{\tr} \leq 2e^{Ct} \, \Big|\Big| \gamma^{\bigwedge \varphi_j^0}_1 - \gamma^{\psi^0}_1 \Big|\Big|_{\tr}^{1/2} + \left( 8\left(e^{Ct}-1\right) \right)^{1/2} N^{-1/2},
\end{equation}
with $C \propto A^{s/2}$;
\item\label{case_b} for interactions $v = \pm v_{s,\delta}$ with
\begin{align}
0 \leq v_{s,\delta}(x) \left\{\begin{array}{cl} \leq DN^{\delta s} &, \, \text{for } |x|\leq N^{-\delta} \\ =|x|^{-s} &, \, \text{for } |x|>N^{-\delta} , \end{array}\right.
\end{align}
with $D>0$, $0<s<6/5$, $\beta = 1 - s/3$ and $\delta < (3-2s)/(6s)$ we have
\begin{equation}\label{main_alpha_ineq_m_applied1}
\Big|\Big| \gamma^{\bigwedge \varphi_j^t}_1 - \gamma^{\psi^t}_1 \Big|\Big|_{\tr} \leq 2e^{Ct} N^{1/2-\gamma/2}\, \Big|\Big| \gamma^{\bigwedge \varphi_j^0}_1 - \gamma^{\psi^0}_1 \Big|\Big|_{\tr}^{1/2} + \left( 8\left(e^{Ct}-1\right) \right)^{1/2} N^{-\gamma/2},
\end{equation}
for all $0 < \gamma \leq 1 - 2\delta s - 2s/3$;
\item\label{case_c} for interactions $v(x)=\pm |x|^{-1}$ and $\beta=2/3$, under the condition that for all $t$ and some $\varepsilon>0$ there are $C_1(t),C_2(t)$ (independent of $N$) such that
\begin{equation}\label{cond_for_coulomb}
\sum_{j=1}^N \norm{\nabla^{3/2+\varepsilon} \varphi_j^t}^2 \leq C_1(t) \, N ~~~~ \text{or} ~~~~ \norm[\infty]{\rho_N^t} \leq C_2(t),
\end{equation}
we have 
\begin{equation}\label{main_alpha_ineq_n_under_cond}
\Big|\Big| \gamma^{\bigwedge \varphi_j^t}_1 - \gamma^{\psi^t}_1 \Big|\Big|_{\tr} \leq 2e^{\int_0^t C(s)ds} \, \Big|\Big| \gamma^{\bigwedge \varphi_j^0}_1 - \gamma^{\psi^0}_1 \Big|\Big|_{\tr}^{1/2} + \left( 8\left(e^{\int_0^t C(s)ds}-1\right) \right)^{1/2} N^{-1/2},
\end{equation}
for some $C(t)$ independent of $N$.
\end{enumerate}
\end{theorem}

\begin{proof}
See Section~\ref{sec:proofs_main_theorems}. \phantom\qedhere
\end{proof}

\vspace{3mm}

\noindent\textbf{Remarks.}
\begin{enumerate}
\setcounter{enumi}{\theremarks}

\item Let us put the three different cases of Theorem~\ref{thm:E_kin_only} into perspective. In \eqref{case_a} we treat singular interactions, but with weaker singularity than Coulomb. Case \eqref{case_b} includes Coulomb interaction with the singularity cut off on a ball with radius much smaller than the average particle distance. However, in case \eqref{case_b} the convergence rate is only $N^{-\gamma/2}$, with $\gamma<1/3$ in the Coulomb case, depending on the cutoff distance $N^{-\delta}$. In case \eqref{case_c} we treat Coulomb interaction without cutoff. We get the convergence rate $N^{-1/2}$. For this we need additional regularity assumptions on the solutions to the fermionic Hartree equations. We expect these conditions to hold for all times under suitable assumptions on the initial data, but we refrain from proving that in this paper.

\item Using Lemmas~\ref{lem:density_conv} and \ref{lem:density_conv_alpha_f} we can also deduce bounds in Hilbert-Schmidt norm. For example, in case \eqref{case_a} we have
\begin{equation}
\sqrt{N}\Big|\Big| \gamma^{\bigwedge \varphi_j^t}_1 - \gamma^{\psi^t}_1 \Big|\Big|_{\HS} \leq \sqrt{2} e^{Ct} \, \left( \sqrt{N} \Big|\Big| \gamma^{\bigwedge \varphi_j^0}_1 - \gamma^{\psi^0}_1 \Big|\Big|_{\HS}\right)^{1/2} + \left( \sqrt{2}\left(e^{Ct}-1\right) \right)^{1/2} N^{-1/2}.
\end{equation}
(Note that our choice of normalization is $\big|\big|\gamma^{\bigwedge \varphi_j^t}_1\big|\big|_{\tr} = 1$ and $\sqrt{N}\big|\big|\gamma^{\bigwedge \varphi_j^t}_1\big|\big|_{\HS} = 1$.)

\item If we set the external field $w^{(N)}=0$, then, for the interactions in this theorem, the bound \eqref{E_kin_AN} holds for all times, if it holds for the initial conditions. Furthermore, we have existence and uniqueness for the solutions to \eqref{scaling_Schr} and \eqref{scaling_hartree} in this case. Thus, for $w^{(N)}=0$, all the results in this theorem hold under the sole assumption that \eqref{E_kin_AN} holds at $t=0$.

\item\label{rem:initial_data_in_cells} Note that in the setup of Theorem~\ref{thm:E_kin_only}, the kinetic energy per particle is of order $1$, i.e., a particle can on average travel a distance of order $1$, while the diameter of the system is of order $N^{1/3}$. Heuristically speaking this means that the system looks static on very large scales. The results above are therefore relevant for observables that are sensitive for properties on small scales. Let us give a simple example of such an observable. We divide the volume of the system into $2N$ cells, where each cell has a volume of order $1$. As initial state we choose smooth wave packets with disjoint supports, each with kinetic energy of order $1$. Each packet occupies one of the cells, and we leave a neighboring cell unoccupied. Then a suitable observable $A$ would be the number of particles in the unoccupied cells. After a short time of order $1$, due to the diffusion of wave packets, there will be order $N$ particles in the previously unoccupied cells. Furthermore, the wave packets will show interference effects.

\item Let us follow up on the example of the previous remark and discuss what result we would expect if we compare the Schr\"odinger evolution with the free evolution. As discussed in Section~\ref{sec:density_O_1}, the average force per particle is of order $N^{-1/3}$. This would be the expected relative error we make in estimating the expectation value of the observable $A$ from above: after some time $t$ the unoccupied cells contain order $N \pm N^{2/3}$ particles. On the other hand, the estimates \eqref{main_alpha_ineq_n_applied1} and \eqref{main_alpha_ineq_n_under_cond} would give a prediction of order $N \pm N^{1/2}$. Thus, in the cases \eqref{case_a} and \eqref{case_c}, the prediction from the fermionic Hartree equations is better than what we could expect from the free evolution (or the free evolution with a time and space dependent phase).

\end{enumerate}
\setcounter{remarks}{\theenumi}

Let us conclude this subsection with a remark about Coulomb interaction with $N^{-1}$ scaling, and with $N^{-2/3}$ scaling for times smaller than order $1$. In the first case, we have free evolution for times of order $N^{1/3-\varepsilon}$, for any $\varepsilon>0$. In the second case, we even have closeness to any state, without evolution, if, of course, the initial conditions are close. We summarize this in the following proposition which for simplicity we only state for Coulomb interaction without external fields.

\begin{proposition}\label{pro:coulomb_free_N_epsilon}
Let $\psi^t \in L^2(\RRR^{3N})$ be a solution to the Schr\"odinger equation 
\begin{equation}
i \partial_t \psi^t = H \psi^t = \left( -\sum_{j=1}^N \Delta_{x_j} + N^{-\delta} \sum_{i<j} |x_i-x_j|^{-1} \right) \psi^t,
\end{equation}
with antisymmetric initial condition $\psi^0 \in L^2(\RRR^{3N})$ with total energy $\SCP{\psi^0}{H\psi^0} = E$. Then,
\begin{enumerate}[(a)]
\item if $\varphi_1^t,\ldots,\varphi_N^t \in L^2(\RRR^3)$ are solutions to the free equations $i \partial_t \varphi_j^t(x) = -\Delta \varphi_j^t(x)$ with orthonormal initial conditions $\varphi_1^0,\ldots,\varphi_N^0 \in L^2(\RRR^3)$, there is a $C>0$ such that for all $t>0$,
\begin{equation}\label{proposition_result}
\Big|\Big| \gamma^{\bigwedge \varphi_j^t}_1 - \gamma^{\psi^t}_1 \Big|\Big|_{\tr} \leq 2 \, \Big|\Big| \gamma^{\bigwedge \varphi_j^0}_1 - \gamma^{\psi^0}_1 \Big|\Big|_{\tr}^{1/2} + C \left(N^{1/6-\delta} E^{1/2}\, t\right)^{1/2},
\end{equation}
\item given time-independent $\varphi_1,\ldots,\varphi_N \in L^2(\RRR^3)$ with $\sum_{j=1}^N ||\nabla \varphi_j||^2 = E^{\mf}_{\kin}$, there is a $C>0$ such that for all $t>0$,
\begin{equation}\label{proposition_result2}
\Big|\Big| \gamma^{\bigwedge \varphi_j}_1 - \gamma^{\psi^t}_1 \Big|\Big|_{\tr} \leq 2 \, \Big|\Big| \gamma^{\bigwedge \varphi_j}_1 - \gamma^{\psi^0}_1 \Big|\Big|_{\tr}^{1/2} + C \left( \left( EN^{-1} + E^{\mf}_{\kin}N^{-1} + N^{1/6-\delta} E^{1/2} \right) \, t \right)^{1/2}.
\end{equation}
\end{enumerate}
\end{proposition}

\begin{proof}
See Section~\ref{sec:proofs_main_theorems}. \phantom\qedhere
\end{proof}

\subsection{Density $\propto N$ Regime}\label{sec:estimates_semiclassical}
Let us now state the result for the derivation of the mean-field dynamics for the regime discussed in Section~\ref{sec:density_O_N}. Such a derivation has recently been given in \cite{benedikter:2013} and here, we reproduce some of these results. We actually use estimates about the propagation of properties of the initial data from \cite{benedikter:2013} (see Lemma~\ref{lem:sc_prop_sc}). A slight improvement of our result is that our conditions on the closeness of the initial data are more transparent and general, see Remark~\ref{rem:comparison}. In contrast to \cite{benedikter:2013}, we express our result solely in terms of the one-particle reduced density matrices. We consider the Schr\"odinger equation \eqref{Schr_scaled_sc} (for simplicity, without external fields) and the corresponding fermionic Hartree equations \eqref{hartree_scaled_sc}. Note that for the interactions we consider the exchange term can easily be shown to be subleading (compared to the error term in the bound \eqref{main_alpha_ineq_sc}), so the theorem can be proven directly for both the fermionic Hartree and for the Hartree-Fock equations
\begin{equation}\label{outline_hartree_fock_scaled_sc_app}
i N^{-1/3} \partial_t \varphi_j^t = \left( -N^{-2/3} \Delta + N^{-1} \left(v * \rho_N^t \right) \right) \varphi_j^t - N^{-1} \sum_{k=1}^N \left(v * (\varphi_k^{t*}\varphi_j^t)\right) \varphi_k^t,
\end{equation}
for $j=1,\ldots,N$ (recall $\rho_N^t = \sum_{i=1}^N |\varphi_i^t|^2$). Note that here we do not have to use the long-range behavior of the interaction, since we are interested in solutions in a volume of order $1$. For technical reasons, we consider a class of interactions that is in particular bounded. Our theorem is analogous to \cite[Thm.\ 2.1]{benedikter:2013}. We denote the trace norm by $\norm[\tr]{\cdot}$ (see also Section~\ref{sec:density_matrices}) and the commutator by $[A,B]=AB-BA$.

\begin{theorem}\label{thm:sc_main_thm}
Set $w^{(N)} = 0$. Let $\psi^t \in L^2(\RRR^{3N})$ be a solution to the Schr\"odinger equation \eqref{Schr_scaled_sc} with antisymmetric initial condition $\psi^0 \in L^2(\RRR^{3N})$. Let $\varphi_1^t,\ldots,\varphi_N^t \in L^2(\RRR^3)$ be solutions to the fermionic Hartree equations \eqref{hartree_scaled_sc} or to the Hartree-Fock equations \eqref{outline_hartree_fock_scaled_sc_app}, with orthonormal initial conditions $\varphi_1^0,\ldots,\varphi_N^0 \in L^2(\RRR^3)$. We assume that $v \in L^1(\RRR^3)$ and
\begin{equation}\label{sc_thm_cond0}
\int d^3k \, (1+|k|^2)\, |\hat{v}(k)| < \infty,
\end{equation}
where $\hat{v}$ is the Fourier transform of $v$, and that the initial conditions $\varphi_1^0,\ldots,\varphi_N^0$ are such, that
\begin{equation}\label{sc_thm_cond1}
\sup_{k\in\RRR^3} (1+|k|)^{-1} \, \Big\lvert\Big\lvert \big[ p^0, e^{ik\cdot x} \big] \Big\rvert\Big\rvert_{\tr} \leq cN^{2/3},
\end{equation}
\begin{equation}\label{sc_thm_cond2}
\Big\lvert\Big\lvert\big[ p^0, \nabla \big] \Big\rvert\Big\rvert_{\tr} \leq c N,
\end{equation}
for some constant $c>0$, where $p^0 = \sum_{j=1}^N \ketbr{\varphi_j^0}$ is the projector on the span of $\varphi_1^0,\ldots,\varphi_N^0$. Then, there are positive $C_1,C_2$, such that for all $t>0$,
\begin{equation}\label{main_alpha_ineq_sc}
\Big|\Big| \gamma^{\bigwedge \varphi_j^t}_1 - \gamma^{\psi^t}_1 \Big|\Big|_{\tr} \leq 2e^{C_1\left(e^{C_2t}-1\right)} \, \Big|\Big| \gamma^{\bigwedge \varphi_j^0}_1 - \gamma^{\psi^0}_1 \Big|\Big|_{\tr}^{1/2} + \left( 8\left(e^{C_1\left(e^{C_2t}-1\right)}-1\right) \right)^{1/2} N^{-1/2}.
\end{equation}
\end{theorem}

\begin{proof}
See Section~\ref{sec:proof_sc_scaling}. \phantom\qedhere
\end{proof}

\noindent\textbf{Remarks.}
\begin{enumerate}
\setcounter{enumi}{\theremarks}

\item The condition \eqref{sc_thm_cond1} captures the semiclassical structure of the initial data, see \cite{benedikter:2013} for a more detailed discussion. The condition \eqref{sc_thm_cond2} is used to propagate \eqref{sc_thm_cond1} in time. Note that for the initial state from Remark~\ref{rem:initial_data_in_cells} the condition \eqref{sc_thm_cond2} does not hold.

\item The theorem also holds with external fields that are such that they preserve the bounds \eqref{sc_thm_cond1} and \eqref{sc_thm_cond2} for all $t$.
\item The double exponential in the bound \eqref{main_alpha_ineq_sc} comes from using the Gronwall Lemma twice: once for propagating the conditions \eqref{sc_thm_cond1} and \eqref{sc_thm_cond2} in time, and once for the time derivative of a quantity similar to the left-hand side of \eqref{main_alpha_ineq_sc} (see Section~\ref{sec:counting_functional}).

\item Using Lemma~\ref{lem:density_conv} we can also deduce a bound in Hilbert-Schmidt norm, i.e.,
\begin{equation}\label{HS_sc}
\sqrt{N}\Big|\Big| \gamma^{\bigwedge \varphi_j^t}_1 - \gamma^{\psi^t}_1 \Big|\Big|_{\HS} \leq \sqrt{2} e^{C_1\left(e^{C_2t}-1\right)} \, \left( \sqrt{N} \Big|\Big| \gamma^{\bigwedge \varphi_j^0}_1 - \gamma^{\psi^0}_1 \Big|\Big|_{\HS}\right)^{1/2} + \left( \sqrt{2}\left(e^{C_1\left(e^{C_2t}-1\right)}-1\right) \right)^{1/2} N^{-1/2}.
\end{equation}

\item\label{rem:comparison} Let us compare Theorem~\ref{thm:sc_main_thm} to \cite[Thm.\ 2.1]{benedikter:2013}. Our bounds \eqref{main_alpha_ineq_sc} and \eqref{HS_sc} are similar to the bounds \cite[Eq.\ (2.20)]{benedikter:2013} and \cite[Eq.\ (2.19)]{benedikter:2013}. In particular, the convergence rates are the same. The bounds \eqref{main_alpha_ineq_sc} and \eqref{HS_sc} are formulated for more general initial data, in the sense that convergence of the initial trace respectively ($\sqrt{N}$ times) Hilbert-Schmidt norm is sufficient to control the left-hand sides of \eqref{main_alpha_ineq_sc} and \eqref{HS_sc}. In \cite[Thm.\ 2.1]{benedikter:2013}, there are two additional bounds that improve the rate of convergence: in \cite[Eq.\ (2.21)]{benedikter:2013} this is achieved by additional assumptions on the closeness of the initial data to antisymmetric product states and in \cite[Eq.\ (2.22)]{benedikter:2013} by evaluating the trace norm difference only for certain observables. Furthermore, in \cite[Thm.\ 2.2]{benedikter:2013} bounds for the reduced $k$-particle density matrices are proven.

\end{enumerate}
\setcounter{remarks}{\theenumi}

\section{Exposition of The Method}\label{sec:counting_functional}
We now present our method for deriving fermionic mean-field dynamics form the microscopic Schr\"odinger dynamics. This method does not rely on BBGKY hierarchies but on a Gronwall type estimate for a suitably defined functional, which leads to many technical advantages. The key of the method is to define the functional $\alpha_f(\psi,\varphi_1,\ldots,\varphi_N)$ that measures the closeness of a many-body wave function $\psi \in L^2(\RRR^{3N})$ to the antisymmetrized product state $\bigwedge_{j=1}^N\varphi_j$, with $\varphi_1,\ldots,\varphi_N \in L^2(\RRR^3)$. For the definition of $\alpha_f$ we need to introduce several projectors that play a crucial role in the proofs later.

\begin{definition}\label{def:projectors}
Let $\varphi_1, \ldots, \varphi_N \in L^2(\RRR^3)$ be orthonormal and $\psi \in L^2(\RRR^{3N})$ be normalized. For all $j,m = 1,\ldots,N$ we define the projector
\begin{equation}
p_m^{\varphi_j} := \ketbr{\varphi_j}_m = \ketbr{\varphi_j(x_m)} = \underbrace{\id \otimes \ldots \otimes \id}_{m-1 ~ \text{times}} \otimes \,\ketbra{\varphi_j}{\varphi_j}\, \otimes \underbrace{\id \otimes \ldots \otimes \id}_{N-m ~ \text{times}},
\end{equation}
i.e., its action on any $\psi \in L^2(\RRR^{3N})$ is given by
\begin{equation}
\left(p_m^{\varphi_j}\psi\right)(x_1,\ldots,x_N) = \varphi_j(x_m)\int\varphi_j^*(x_m)\psi(x_1,\ldots,x_N) \, d^3x_m.
\end{equation}
We define $p_m := \sum_{j=1}^N p_m^{\varphi_j}$ and $q_m := 1-p_m$. For any $0 \leq k \leq N$ we define
\begin{equation}
P^{(N,k)} := \left( \prod_{m=1}^k q_m \prod_{m=k+1}^N p_m \right)_{\sym} = \sum_{\vec{a} \in \AAA_k} \prod_{m=1}^N (p_m)^{1-a_m}(q_m)^{a_m},
\end{equation}
with the set $\AAA_k := \big\{ \vec{a}=(a_1,\ldots,a_N) \in \{0,1\}^N: \sum_{m=1}^N a_m=k \big\}$, i.e., $P^{(N,k)}$ is the symmetrized tensor product of $q_1,\ldots,q_k,p_{k+1},\ldots,p_N$. We define $P^{(N,k)}=0$ for all $k<0$ and $k>N$. We call any $f:\{0,\ldots,N\} \to [0,1]$ with $f(0)=0$, $f(N)=1$ a \emph{weight function}. For any weight function $f$ we define the operator
\begin{equation}
\widehat{f} := \sum_{k=0}^N f(k) P^{(N,k)},
\end{equation}
and the functional
\begin{equation}\label{definition_alpha}
\alpha_f(\psi,\varphi_1,\ldots,\varphi_N) := \SCP{\psi}{\widehat{f} \, \psi},
\end{equation}
where $\SCP{\cdot}{\cdot}$ denotes the scalar product on $L^2(\RRR^{3N})$.
\end{definition}

The functional $\alpha_f$ has first been introduced by one of the authors (P.P.) for bosons \cite{pickl:2011method}, that is, with $p_m = \ketbr{\varphi}_m$. The functional was used in \cite{pickl:2011method,pickl:2010hartree} for the derivation of the bosonic Hartree equation, and in \cite{pickl:2010gp_ext,pickl:2010gp_pos} for the derivation of the Gross-Pitaevskii equation. A variant of it has been used in \cite{deckert:2012,deckert:2014} to derive the effective dynamics of a tracer particle in a Bose gas and in \cite{anapolitanos:2011} to derive the Hartree-von Neumann limit.

Let us explain Definition~\ref{def:projectors} a little further. First, note that $p_m$, $q_m$ and $P^{(N,k)}$ are indeed projectors, since $\varphi_1,\ldots,\varphi_N$ are assumed to be orthonormal. As operators on $L^2(\RRR^3)$, $p_1$ projects on the subspace spanned by $\varphi_1,\ldots,\varphi_N$, and $q_1$ on its complement, i.e., in particular, $p_1q_1=0$. The $P^{(N,k)}$ have the property that
\begin{equation}\label{P_Nk_properties}
\sum_{k=0}^N P^{(N,k)} = \prod_{m=1}^N (p_m+q_m) = 1 \quad\quad \text{and} \quad\quad P^{(N,k)}P^{(N,\ell)} = \delta_{k\ell}P^{(N,k)},
\end{equation}
where the last equation is true since $P^{(N,k)}$ contains exactly $k$ $q$-projectors in each summand and $q_mp_m=0$. We can thus define the decomposition $\psi = \sum_{k=0}^N P^{(N,k)}\psi$ of the microscopic wave function, where now each contribution $P^{(N,k)}\psi$ has $(N-k)$ particles in one of the orbitals $\varphi_1,\ldots,\varphi_N$ (``good particles'') and $k$ particles \emph{not} in the orbitals $\varphi_1,\ldots,\varphi_N$ (``bad particles''). With the functional $\alpha_f$, each contribution $P^{(N,k)}\psi$ is given the weight $f(k)$, i.e., $f$ specifies how much weight is given to the number of particles outside the antisymmetrized product state $\bigwedge_{j=1}^N\varphi_j$. Since $0 \leq \alpha_f \leq 1$, $\alpha_f \approx 0$ means (for suitable weight functions $f$) that the approximation of $\psi$ by $\bigwedge_{j=1}^N\varphi_j$ is very good, while $\alpha_f \approx 1$ means that the approximation is not valid at all. By choosing an appropriate $f$ we can fine tune what exactly is meant by the closeness of $\psi$ to $\bigwedge_{j=1}^N\varphi_j$. One obvious and very simple weight is the relative number $k/N$. We always denote this weight function by
\begin{equation}\label{weight_n}
n(k) = k/N
\end{equation}
and the corresponding functional by $\alpha_n$. For this weight function, due to the antisymmetry of $\psi$ and \eqref{P_Nk_properties}, the functional has the simple form (recall that each summand in $P^{(N,k)}$ contains exactly $k$ projectors $q$)
\begin{equation}
\alpha_n = \sum_{k=0}^N \frac{k}{N} \, \bigSCP{\psi}{P^{(N,k)} \psi} = \sum_{k=0}^N N^{-1} \bigSCP{\psi}{\sum_{m=1}^N q_m P^{(N,k)} \psi} = \bigSCP{\psi}{q_1 \psi}.
\end{equation}
Let us note here that $\alpha_n$ has been used before in \cite{bach:1993,graf_solovej:1994} to measure deviation from the antisymmetrized product structure in the static setting; see also the remarks following \eqref{alpha_n_and_tr}. Furthermore, $\alpha_n$ coincides with the measure introduced in \cite{benedikter:2013} for states in Fock space with fixed particle number $N$. Another important weight is
\begin{equation}\label{weight_m_gamma}
m^{(\gamma)}(k) = \left\{\begin{array}{cl} k N^{-\gamma} &, \text{for } k \leq N^{\gamma}\\ 1 & , \text{otherwise,} \end{array}\right.
\end{equation}
with some $0 < \gamma \leq 1$. The function $m^{(\gamma)}(k)$ gives a much larger weight to already very few particles outside the antisymmetrized product structure. On the other hand, for $k>N^{\gamma}$, i.e., very many particles outside the antisymmetrized product structure, $m^{(\gamma)}(k)$ gives the same weight $1$ for all $k>N^{\gamma}$. These properties enable us to derive mean-field approximations for a much wider range of physical situations, e.g., singular or weakly scaled interaction potentials. Let us stress that the freedom in the choice of the weight function is a key feature and advantage of the described method.

Note that we could also define projectors $p^{\varphi} = \sum_{m=1}^N p_m^{\varphi}$ and $q^{\varphi} = 1-p^{\varphi}$. For antisymmetric $\psi$, the $\alpha_f$ functional defined with those projectors coincides with $\alpha_f$ from Definition~\ref{def:projectors}, since $\big( \prod_{j=1}^k q^{\varphi_j} \prod_{j=k+1}^N p^{\varphi_j} \big)_{\sym}\psi = P^{(N,k)} \psi$, as can be seen from multiplying out the left-hand side. Note that the projectors $p^{\varphi}$ are related to fermionic creation and annihilation operators $a(\varphi)$ and $a^{\dagger}(\varphi)$ (see, e.g., \cite{lieb:2010}) by $p^{\varphi} = a^{\dagger}(\varphi)a(\varphi)$, i.e., $p^{\varphi}$ is the number operator for the state $\varphi$.

The goal of this work is to prove bounds on $\alpha_f(t) = \alpha_f\big(\psi^t,\varphi_1^t,\ldots,\varphi_N^t\big)$, where $\psi^t$ is a solution to the Schr\"odinger equation and $\varphi^t_1,\ldots,\varphi^t_N$ are solutions to the fermionic Hartree equations. In more detail, we first look for a bound on the time derivative of $\alpha_f(t)$ of the type
\begin{equation}\label{general_form_bound_alpha_dot}
\partial_t \alpha_f(t) \leq C(t) \left( \alpha_f(t) + N^{-\delta} \right),
\end{equation}
which then, by Gronwall's Lemma, implies the bound
\begin{equation}\label{general_form_bound}
\alpha_f(t) \leq \, e^{\int_0^t C(s) ds} \, \alpha_f(0) + \left( e^{\int_0^t C(s) ds} -1 \right) N^{-\delta},
\end{equation}
where the function $C(t)$ is independent of $N$, and $\delta>0$ is called the \emph{convergence rate}. In the main theorems of Section~\ref{sec:main_results}, the weight function is either $n$ from \eqref{weight_n} or $m^{(\gamma)}$ from \eqref{weight_m_gamma}. A bound of the form \eqref{general_form_bound} implies that if initially (at time $t=0$) $\alpha_f$ is small, then it stays small for times $t>0$ and $N$ large enough. In the limit of $N\to\infty$, we arrive at the statement that $\lim_{N\to\infty}\alpha_f(t=0)=0$ implies $\lim_{N\to\infty}\alpha_f(t)=0$ for all $t>0$. Note that for all $f\geq n$, we have $\alpha_f \geq \alpha_n$, i.e., if the bound \eqref{general_form_bound_alpha_dot} holds for such an $f$, we find
\begin{equation}
\alpha_n(t) \leq \alpha_f(t) \leq \, e^{\int_0^t C(s) ds} \, \left( \alpha_f(0) + N^{-\delta} \right).
\end{equation}
Thus, by using a weight function $f \geq n$ (which can be advantageous in the estimates, see below) we can still control $\alpha_n(t)$, but now under stronger conditions on the initial state, namely that $\alpha_f(0)$ has to be small in $N$.

\subsection{Connection to Density Matrices}\label{sec:dens_mat_summary}
The functional $\alpha_f$ is closely related to the trace and Hilbert-Schmidt norms of the difference between reduced one-particle density matrices as defined in \eqref{definition_dens_mat_one_part}. 
Note that for an antisymmetrized product state $\bigwedge_{j=1}^N \varphi_j$ we find
\begin{equation}
\gamma^{\bigwedge \varphi_j}_1 = N^{-1} p_1.
\end{equation}
Let us now consider $\alpha_n$, i.e., the $\alpha$-functional with the weight $n(k)=k/N$. Note that 
\begin{equation}\label{alpha_n_and_tr}
\alpha_n = \SCP{\psi}{q_1 \psi} = \tr(\gamma_1^{\psi}q_1) = \tr(\gamma_1^{\psi}(1-p_1)) = N \tr(\gamma_1^{\psi}(N^{-1}-\gamma^{\bigwedge \varphi_j}_1)),
\end{equation}
where $\tr(\cdot)$ denotes the trace. It is the expression $\tr(\gamma_1^{\psi}q_1)$ that has been used before to measure deviation from the antisymmetrized product structure in the static setting, see \cite{bach:1993,graf_solovej:1994}. The relations between $\alpha_n$ and the differences between Schr\"odinger and Hartree reduced one-particle density matrices in trace norm $\norm[\tr]{\cdot}$ and Hilbert-Schmidt norm $\norm[\HS]{\cdot}$ are summarized in the following lemma (see Section~\ref{sec:density_matrices} for the definition of the norms).

\begin{lemma}\label{lem:density_conv}
Let $\psi \in L^2(\RRR^{3N})$ be antisymmetric and normalized, and let $\varphi_1,\ldots,\varphi_N \in L^2(\RRR^3)$ be orthonormal. Then
\begin{equation}\label{bound_tr_alpha_HS}
\norm[\tr]{\gamma^{\bigwedge \varphi_j}_1 - \gamma^{\psi}_1}^2 \leq 8 \, \alpha_n \leq 8 \sqrt{N} \norm[\HS]{\gamma^{\bigwedge \varphi_j}_1 - \gamma^{\psi}_1},
\end{equation}
\begin{equation}\label{bound_HS_alpha_tr}
N \norm[\HS]{\gamma^{\bigwedge \varphi_j}_1 - \gamma^{\psi}_1}^2 \leq 2 \, \alpha_n \leq \norm[\tr]{\gamma^{\bigwedge \varphi_j}_1 - \gamma^{\psi}_1}.
\end{equation}
\end{lemma}

\begin{proof}
See Section~\ref{sec:density_matrices}. \phantom\qedhere
\end{proof}

Note that the extra $N$ factors are due to the choice of normalization; indeed
\begin{equation}
\norm[\tr]{\gamma^{\bigwedge \varphi_j}_1} = 1 ~~~~\text{and}~~~~ \norm[\HS]{\gamma^{\bigwedge \varphi_j}_1} = N^{-1/2}.
\end{equation}
Lemma~\ref{lem:density_conv} is the main result of this section. It implies in particular that
\begin{equation}\label{equivalence_alpha_tr_HS}
\lim_{N \to \infty} \alpha_n = 0 ~~
\Longleftrightarrow~~  \lim_{N \to \infty} \norm[\tr]{\gamma^{\bigwedge \varphi_j}_1 - \gamma^{\psi}_1} = 0 ~~
\Longleftrightarrow~~  \lim_{N \to \infty} \sqrt{N}\norm[\HS]{\gamma^{\bigwedge \varphi_j}_1 - \gamma^{\psi}_1} = 0,
\end{equation}
i.e., convergence of $\gamma_1^\psi$ to $\gamma^{\bigwedge \varphi_j}_1$ in certain norms is equivalent to convergence of $\alpha_n$ to zero. However, there is a difference in the convergence rates, e.g., if $\alpha_n$ converges with rate $N^{-1}$, then the density matrices converge in trace norm only with rate $N^{-1/2}$, as can be seen from \eqref{bound_tr_alpha_HS}. 

Let us now consider more general weight functions $f$ with $f(k) \geq n(k)$ for all $k$ and in particular the weight $m^{(\gamma)}(k)$ from \eqref{weight_m_gamma} which we use later. Note that for those weights $\alpha_f \geq \alpha_n$.

\begin{lemma}\label{lem:density_conv_alpha_f}
Let $\psi \in L^2(\RRR^{3N})$ be antisymmetric and normalized, and let $\varphi_1,\ldots,\varphi_N \in L^2(\RRR^3)$ be orthonormal. Then, for all $f$ with $f(k)\geq k/N ~\forall k=1,\ldots,N$,
\begin{equation}\label{bound_tr_HS_alpha_f}
\norm[\tr]{\gamma^{\bigwedge \varphi_j}_1 - \gamma^{\psi}_1}^2 \leq 8 \, \alpha_f, \qquad N \norm[\HS]{\gamma^{\bigwedge \varphi_j}_1 - \gamma^{\psi}_1}^2 \leq 2 \, \alpha_f,
\end{equation}
\begin{equation}\label{bound_tr_alpha_m}
\alpha_{m^{(\gamma)}} \leq  N^{1-\gamma} \alpha_n.
\end{equation}
\end{lemma}

Thus, convergence of $\alpha_f$ to zero still implies convergence of $\gamma_1^\psi$ to $\gamma^{\bigwedge \varphi_j}_1$, but in general not the other way around. We use \eqref{bound_tr_alpha_m} to express some of our main result only in terms of trace norms.

Let us make a brief remark about convergence in operator norm. Note that $||\gamma^{\psi}_1||_{\op} \leq N^{-1}$ for antisymmetric $\psi$, so a possible indicator of convergence would be the operator norm times $N$. This, however, is not a good type of convergence for our purpose, since the operator norm is given by the largest eigenvalue (which at most can be $N^{-1}$ for fermionic density matrices; recall our choice of normalization). Thus, while convergence of $N$ times the operator norm does imply convergence of $\alpha_n$, the opposite is not true. One orbital not in the antisymmetrized product of the $\varphi_1,\ldots,\varphi_N$ is enough to let the operator norm of $N$ times the difference between the density matrices be equal to one, while $\alpha_n$ converges to zero.

\subsection{Main Theorems in Terms of $\alpha$}\label{sec:main_theorem_mf_general_v}
Let us state and discuss two Theorems with bounds on $\alpha_n$ and $\alpha_{m^{(\gamma)}}$ which hold for very general Hamiltonians. From these theorems we can then deduce Theorem~\ref{thm:E_kin_only}.

Theorems \ref{thm:estimates_terms_alpha_dot_beta_n} and \ref{thm:estimates_terms_alpha_dot_beta_general} are of the form: Given certain properties of the solutions to the fermionic Hartree equations, the mean-field approximation for the dynamics is good, i.e., $\alpha_f(t)\leq C(t)\big(\alpha_f(0)+N^{-\delta}\big)$ for some $\delta>0$. We consider wave functions $\psi^t \in L^2(\RRR^{3N})$ that are solutions to the Schr\"odinger equation \eqref{Schr_intro} with self-adjoint Hamiltonian \eqref{Schr_intro_H} with real and possibly scaled interaction $v^{(N)}(x)=v^{(N)}(-x)$. The most important example for the free part $H_j^0$ is the non-relativistic free Hamiltonian with external field, $H_j^0=-\Delta_j+w^{(N)}(x_j)$, but we could also replace the Laplacian by relativistic operators like $\sqrt{-\Delta +m^2}-m$ ($m>0$) or $|\nabla|$. The fermionic mean-field equations for the one-particle wave functions $\varphi^t_1,\ldots,\varphi^t_N \in L^2(\RRR^3)$ are given by \eqref{hartree_intro}. The first theorem gives a bound on $\alpha_n(t) = \SCP{\psi^t}{q_1^t\psi^t}$.

\begin{theorem}\label{thm:estimates_terms_alpha_dot_beta_n}
Let $t\in[0,T)$ for some $0<T\in \RRR\cup\infty$. Let $\psi^t \in L^2(\RRR^{3N})$ be a solution to the Schr\"odinger equation \eqref{Schr_intro} with Hamiltonian \eqref{Schr_intro_H} with antisymmetric initial condition $\psi^0 \in L^2(\RRR^{3N})$. Let $\varphi_1^t,\ldots,\varphi_N^t \in L^2(\RRR^3)$ be solutions to the fermionic Hartree equations \eqref{hartree_intro} with orthonormal initial conditions $\varphi_1^0,\ldots,\varphi_N^0 \in L^2(\RRR^3)$.

We assume that $v^{(N)}$ and $\rho_N^t = \sum_{i=1}^N |\varphi_i^t|^2$ for all $t\in[0,T)$ are such that there is a positive $D(t)$ (independent of $N$), such that
\begin{equation}\label{alpha_dot_n_ass_2}
\sup_{y\in\RRR^3} \Big(\big(v^{(N)}\big)^2*\rho_N^t\Big)(y) \leq D(t) \, N^{-1}.
\end{equation}
Then there is a positive $C(t) = 9\sqrt{D(t)}$, such that
\begin{equation}\label{main_alpha_ineq_n}
\alpha_n(t) \leq e^{\int_0^t C(s) ds} \, \alpha_n(0) + \left( e^{\int_0^t C(s) ds} - 1 \right) N^{-1}.
\end{equation}
\end{theorem}

\begin{proof}
See Section~\ref{sec:proofs_main_theorems_gen}. \phantom\qedhere
\end{proof}

\noindent\textbf{Remarks.}
\begin{enumerate}
\setcounter{enumi}{\theremarks}
\item The corresponding bounds for the reduced one-particle density matrices follow from Lemma~\ref{lem:density_conv}.

\item The condition \eqref{alpha_dot_n_ass_2} is only a condition on the solutions to the fermionic Hartree equations \eqref{hartree_intro}, and not on the solutions to the Schr\"odinger equation \eqref{Schr_intro}. Note that for solutions $\varphi_1^t,\ldots,\varphi_N^t \in H^1(\RRR^3)$ (the first Sobolev space) and interactions with at most a singularity $|x|^{-1}$, the left-hand side of \eqref{alpha_dot_n_ass_2} is always finite due to Hardy's inequality. Therefore, the challenge lies in the $N$-dependence on the right-hand side of \eqref{alpha_dot_n_ass_2}. 

\item Note that condition \eqref{alpha_dot_n_ass_2} implies that (by Cauchy-Schwarz and $\int \rho_N^t=N$)
\begin{equation}
\sup_{y\in\RRR^3} \Big(\big\lvert v^{(N)} \big\rvert*\rho_N^t\Big)(y) \leq \sqrt{D(t)},
\end{equation}
i.e., the scaled mean-field interaction is everywhere bounded.

\item\label{itm:fluctuations} There is an interesting connection between condition \eqref{alpha_dot_n_ass_2} and the fluctuations around the mean-field. At each point $y \in \RRR^3$, let us denote the ``fluctuations'' in the interaction at $y$ in the state $\phi = \bigwedge_{j=1}^N \varphi_j^t$ by $\Var[\phi]\big(\sum_{k=1}^N v^{(N)}_{ky}\big)$, with the definition $\Var[\phi](X) = \SCP{\phi}{X^2\phi} - \SCP{\phi}{X\phi}^2$, and where $v_{ky}=v(x_k-y)$ and $x_1,\ldots,x_N$ denote the integration variables in the scalar product. Then, by a simple calculation, we find $\Var[\phi]\big(\sum_{k=1}^N v^{(N)}_{ky}\big) \leq \big((v^{(N)})^2*\rho_N^t\big)(y)$. Therefore, if condition \eqref{alpha_dot_n_ass_2} holds, then the fluctuations around the mean-field vanish at each point for large $N$, with rate $N^{-1}$. Note that $N^{-1}$ is the typical size of fluctuations in the (weak) law of large numbers, for independently identically distributed random variables. It is therefore not surprising that under this condition the derivation of the mean-field dynamics succeeds (although condition \eqref{alpha_dot_n_ass_2} could well be too restrictive; see also Theorem~\ref{thm:estimates_terms_alpha_dot_beta_general}).

\item\label{itm:exch} Theorem~\ref{thm:estimates_terms_alpha_dot_beta_n} shows indirectly that under the condition \eqref{alpha_dot_n_ass_2} the scaled exchange term is at most of $\bigO(N^{-1/2})$. (For the simple example of plane waves in a box of size $N$ and Coulomb interaction with scaling exponent $\beta = 2/3$, it is actually of $\bigO(N^{-2/3})$.) This is so because an exchange term of $\bigO(N^{-\delta})$ gives an error term of $\bigO(N^{-2\delta})$ in the $\alpha_n$ estimate. We show this in Remark \ref{itm:exch_term_order}, following the proof in Section~\ref{sec:proofs_main_theorems_gen}.
\end{enumerate}
\setcounter{remarks}{\theenumi}

The condition \eqref{alpha_dot_n_ass_2} can be relaxed and replaced by other conditions if we use the weight function $m^{(\gamma)}(k)$ from \eqref{weight_m_gamma}. This allows to treat more singular interactions and smaller scaling exponents. Let us first summarize the precise assumptions that we need on the scaled interaction $v^{(N)}$ and the density $\rho_N^t$ and afterwards state the theorem.

\begin{assumption}\label{ass:for_main_thm}
For all $t\in[0,T)$, $\rho_N^t = \sum_{i=1}^N |\varphi_i^t|^2$ and $v^{(N)}$ are such that there are a (possibly $N$-dependent) volume $\Omega_N\subset\RRR^3$, positive $D_i(t)$ (independent of $N$) and $0<\gamma\leq 1$, such that
\begin{equation}\label{alpha_dot_m_ass_2}
\sup_{y\in\RRR^3} \Big(\big(v^{(N)}\big)^2*\rho_N^t\Big)(y) \leq D_1(t) \, N^{-\gamma},
\end{equation}
\begin{equation}\label{alpha_dot_m_ass_3}
\int \bigg(\big(v^{(N)}\big)^2*\rho^t_N\bigg)(y)\,\rho^t_N(y)\,d^3y \leq D_2(t),
\end{equation}
\begin{equation}\label{alpha_dot_m_ass_4}
\sup_{y\in\RRR^3} \int_{\Omega_N+y} \Big(v^{(N)}(y-x)\Big)^2\rho^t_N(x)\,d^3x \leq D_3(t) \, N^{-1},
\end{equation}
\begin{equation}\label{alpha_dot_m_ass_5}
\sup_{y \in \RRR^3 \setminus \Omega_N} \big|v^{(N)}(y)\big| \leq \sqrt{D_4(t)} \, N^{-1/2-\gamma/2}.
\end{equation}
\end{assumption}

Under this assumption we can conclude convergence of $\alpha_{m^{(\gamma)}}(t)$.

\begin{theorem}\label{thm:estimates_terms_alpha_dot_beta_general}
Let $t\in[0,T)$ for some $0<T\in \RRR\cup\infty$. Let $\psi^t \in L^2(\RRR^{3N})$ be a solution to the Schr\"odinger equation \eqref{Schr_intro} with Hamiltonian \eqref{Schr_intro_H} with antisymmetric initial condition $\psi^0 \in L^2(\RRR^{3N})$. Let $\varphi_1^t,\ldots,\varphi_N^t \in L^2(\RRR^3)$ be solutions to the fermionic Hartree equations \eqref{hartree_intro} with orthonormal initial conditions $\varphi_1^0,\ldots,\varphi_N^0 \in L^2(\RRR^3)$.

We assume that $v^{(N)}$ and $\rho_N^t := \sum_{i=1}^N |\varphi_i^t|^2$ for all $t\in[0,T)$ are such that Assumption~\ref{ass:for_main_thm} holds.
Then there is a positive $C(t)= 16\sqrt{3}\, \max\big\{ D_1(t), D_2(t), D_3(t), D_4(t) \big\}^{1/2}$, such that
\begin{equation}\label{main_alpha_ineq_m}
\alpha_{m^{(\gamma)}}(t) \leq e^{\int_0^t C(s) ds} \, \alpha_{m^{(\gamma)}}(0) + \left( e^{\int_0^t C(s) ds} - 1 \right) N^{-\gamma}.
\end{equation}
\end{theorem}

\begin{proof}
See Section~\ref{sec:proofs_main_theorems_gen}. \phantom\qedhere
\end{proof}

\noindent\textbf{Remarks.}
\begin{enumerate}
\setcounter{enumi}{\theremarks}

\item The corresponding bounds for the reduced one-particle density matrices follow from Lemmas~\ref{lem:density_conv} and \ref{lem:density_conv_alpha_f}.

\end{enumerate}
\setcounter{remarks}{\theenumi}

\subsection{Sketch of Proof}\label{sec:sketch_of_proof}
Let us give a brief overview of how the bounds \eqref{main_alpha_ineq_n} and \eqref{main_alpha_ineq_m} can be obtained. We also aim to illustrate that the idea to ``count the number of particles'' not in the antisymmetrized product is very natural and has a clear physical interpretation which is reflected in the proof. We first consider the weight function $n(k)=k/N$. Note that $q_1^t$ is a solution to the Heisenberg equation of motion $i\partial_t q_1^t = [H_1^{\mf},q_1^t]$, where $H_1^{\mf}=H_1^0 + V_1^{(N)}$ is the ``mean-field Hamiltonian'' from the right-hand side of  \eqref{hartree_intro}, with $V_1^{(N)} = \big( v^{(N)}*\rho_N^t \big)$. Using the antisymmetry of $\psi^t$, we then find
\begin{equation}\label{summary_dt_alpha}
\partial_t \alpha_n(t) = \partial_t \bigSCP{\psi^t}{q_1^t\psi^t} = i \, \bigSCP{\psi^t}{\left[ (N-1) v^{(N)}_{12} - V_1^{(N)},q_1^t\right]\psi^t} = (I)_n + (II)_n + (III)_n,
\end{equation}
where, using $p_1^t+q_1^t=1=p_2^t+q_2^t$,
\begin{align}
\label{alpha_derivative_n_gen_I}(I)_n &= 2 \, \Im\, \bigSCP{\psi^t}{q_1^t\bigg( (N-1)p_2^tv^{(N)}_{12}p_2^t - V^{(N)}_1 \bigg) p_1^t \psi^t}, \\
\label{alpha_derivative_n_gen_II}(II)_n &= 2 \, \Im\, \bigSCP{\psi^t}{q_1^tq_2^t \, (N-1)v^{(N)}_{12} \, p_1^tp_2^t \psi^t}, \\
\label{alpha_derivative_n_gen_III}(III)_n &= 2 \, \Im\, \bigSCP{\psi^t}{q_1^tq_2^t (N-1)v^{(N)}_{12} p_1^tq_2^t \psi^t}.
\end{align}
Let us emphasize here that the kinetic and external field terms coming from the Schr\"odinger and the fermionic mean-field equations cancel, which is why the theorems in Section~\ref{sec:main_theorem_mf_general_v} hold for any $H_j^0$. Also, terms $(II)_n$ and $(III)_n$ do not depend on the mean-field $V^{(N)}_1$; it is only term $(I)_n$ where a difference between Schr\"odinger interaction and the mean-field enters. Note that these three terms have an intuitive explanation. The change $\partial_t \alpha_n(t)$ in the ``number of bad particles'' is given by three different kind of transitions due to the microscopic interaction: transitions between one good, one bad particle and two good particles (term $(I)_n$); between two bad and two good particles (term $(II)_n$); and transitions between two bad and one good, one bad particle (term $(III)_n$). Note, however, that for $\alpha_n(0)=0$, i.e., $\psi^0 = \bigwedge_{j=1}^N \varphi_j^0$, we find $\partial_t\alpha_n(t) \big|_{t=0} = 0$. Thus, what causes the deviations from the antisymmetrized product structure in the first place, is a higher order effect, viz., the fluctuations around the mean-field. The exponential growth in time in \eqref{general_form_bound} comes from the fact that the larger the number of bad particles, the more possibilities do good particles have to become correlated, i.e., as in \eqref{general_form_bound_alpha_dot}, the change in the number of bad particles at some time $t$ is expected to be proportional to the number of bad particles at this time, plus fluctuations.

Let us briefly discuss how the three terms in \eqref{summary_dt_alpha} can be bounded rigorously (Lemma~\ref{lem:estimates_terms_alpha_dot_beta}). Let us first show why the terms can be bounded by $C_N \alpha_n(t)$ plus error terms of $\bigO(N^{-1})$, where $C_N$ denotes some possibly $N$-dependent constant, and afterwards discuss why $C_N$ in fact does \emph{not} depend on $N$ for the scalings discussed in Sections \ref{sec:density_O_N} and \ref{sec:density_O_1}. Recall that $\norm{q_1^t\psi^t} = \sqrt{\SCP{\psi^t}{q_1^t\psi^t}} = \sqrt{\alpha_n(t)}$. Thus, by using Cauchy-Schwarz, term $(III)_n$ is bounded by $C_N \alpha_n(t)$, assuming $\norm[\op]{v_{12}p_1^t}$ is bounded. In term $(II)_n$ we also have two $q$'s available, but both $q$'s are left from $v_{12}^{(N)}$, so we cannot directly apply Cauchy-Schwarz. However, there is a trick that uses the antisymmetry of $\psi^t$ to shift the $q_2^t$ to the right side of the scalar product, on the expense of a boundary term of $\bigO(N^{-1})$. Then, by Cauchy-Schwarz, we can again conclude $(II)_n \leq C_N (\alpha_n(t) + N^{-1})$. Term $(I)_n$ is the crucial term in the sense that there is only one $q$-projector available (which is not enough for the desired bound) and one thus has to use a cancellation between the microscopic interaction and the mean-field. From the proof of Lemma~\ref{lem:estimates_terms_alpha_dot_beta} we see that one can indeed extract an extra $q$-projector from the difference $(N-1)p_2^tv_{12}^{(N)}p_2^t - V_1^{(N)}$, again at the expense of an error term of $\bigO(N^{-1})$, which yields the desired bound.\footnote{Note that here lies the crucial difference compared to the bosonic Hartree equation, as treated in \cite{pickl:2011method,pickl:2010hartree}. There, since there is just one orbital $\varphi$, i.e., $p_m = \ketbra{\varphi}{\varphi}_m$, and since $v^{(N)} = N^{-1} v$, one can directly calculate that
\begin{align}
V^{(N)}_1 - (N-1)p_2v^{(N)}_{12} p_2 = \left(v*|\varphi|^2\right)(x_1) \, (q_2 + N^{-1}p_2),
\end{align}
i.e., the corresponding term $(I)_n$ can directly be seen to be bounded by $C_N \big(\alpha_n(t) + N^{-1} \big)$. (Note that in \cite{pickl:2011method,pickl:2010hartree}, the term $\left(v*|\varphi|^2\right)(x_1) \, q_2$ is usually regarded as being part of term $(III)_n$.)}

For Coulomb interaction and the corresponding scaling $v^{(N)} = N^{-2/3}v$, there is a prefactor $N^{1/3}$ in front of the scalar products in the three terms from \eqref{summary_dt_alpha}. For the desired bound one thus needs to gain an extra $N^{-1/3}$ from the scalar products, so that the constants $C_N$ from above are of $\bigO(1)$, i.e., $N$-independent. For the case discussed in Section~\ref{sec:density_O_1}, it is the long-range decay of the Coulomb interaction that gives us this additional factor. In the estimates, the $p$ projectors are used to smooth out the singularity, which leads to terms like $v*\rho_N^t$ or $v^2*\rho_N^t$. From these terms we gain the additional $N^{-1/3}$ by using kinetic energy inequalities (Lieb-Thirring inequalities, see Section~\ref{sec:mean-field_scalings_general}) which take into account the orthonormality of $\varphi_1^t,\ldots,\varphi_N^t$ (fermi pressure), i.e., the fact that the system volume is of order $N$ when the total kinetic energy is bounded by $AN$. This leads to the results in Section~\ref{sec:main_theorem_mf_x-s}. In the semiclassical case, we have to gain the additional $N^{-1/3}$ by other means, namely from condition \eqref{sc_thm_cond1}. In the estimates, this is reflected by using $N^{-1}\big\lvert\big\lvert p_1^te^{ik\cdot x}q_1^t \big\rvert\big\rvert_{\tr} \leq c(t) N^{-1/3}$, see Section~\ref{sec:proof_sc_scaling}.

Finally, let us discuss where and how the use of the weight function \eqref{weight_m_gamma} helps in the estimates. For strong singularities or weak scalings, term $(III)_n$ can be problematic, since there is only one projector $p$ available to smooth out the singularity and gain the additional $N^{-1/3}$ by the kinetic energy inequalities. Here we would like to improve the estimate. For general weight functions, $\partial_t \alpha_f(t)$ is calculated in Section~\ref{sec:alpha_dot}. Each of the transitions between good and bad particles is now weighted with the corresponding change in the weight function, i.e., a derivative of the weight function appears. More exactly, $\partial_t \alpha_f(t)$ is again given by three terms similar to \eqref{summary_dt_alpha}, but with $\psi^t$ replaced by
\begin{equation}\label{der_f_psi}
\sqrt{N}\widehat{f'^{(d)}}^{1/2}\psi^t,
\end{equation}
where $f'^{(d)}$ is a discrete derivative of $f$. (This is consistent with \eqref{summary_dt_alpha}, since the derivative of $n(k)=k/N$ equals $N^{-1}$.) The key observation is now that the derivative of the weight function $m^{(\gamma)}$ is zero for all $k>N^{\gamma}$, i.e., all contributions coming from $P^{(N,k)}\psi^t$ with $k>N^{\gamma}$ vanish. This helps us to make term $(III)_{m^{(\gamma)}}$ small, which only gives contributions from large $k$ due to the three projectors $q$.

\section{Notation and Preliminaries}\label{sec:notation}
We summarize here some notation, list often used inequalities and establish some basic properties of the projectors $p_m^{\varphi}$ that we use in the proofs throughout the following sections. The Hilbert space of complex square integrable functions on $\RRR^d$ is denoted by $L^2(\RRR^d) = L^2(\RRR^d, \CCC)$ and $H^1(\RRR^d) = \left\{ f \in L^2(\RRR^d): \norm{\nabla f} < \infty \right\}$ denotes the first Sobolev space. For $f\in L^2(\RRR^d)$ we sometimes write
\begin{equation}
\norm{f}^2 = \scp{f}{f} = \int_{\RRR^d} |f(x)|^2 \, d^dx = \int |f|^2.
\end{equation}
We denote by $\SCP{\cdot}{\cdot}_{a+1,\ldots,N}$ the scalar product only in the variables $x_{a+1},\ldots,x_N$, i.e., it is a ``partial trace'' or ``partial scalar product'', formally defined for any $\chi,\psi\in L^2(\RRR^{3N})$ by
\begin{equation}\label{partial_scp}
\SCP{\psi}{\chi}_{a+1,\ldots,N}(x_1,\ldots,x_a) := \int d^3x_{a+1} \ldots \int d^3x_N \, \psi^*(x_1,\ldots,x_N) \chi(x_1,\ldots,x_N),
\end{equation}
which should be regarded as a vector in $L^1(\RRR^{3a})$ (for $\chi=\psi$, it is the diagonal of the reduced $a$-particle density matrix, a generalization of \eqref{definition_dens_mat_one_part}). In the same style we denote by $\ket{\cdot}_m$ a vector in $L^2(\RRR^3)$ acting only on the $m$-th variable of $L^2\big((\RRR^3)^N\big)$, by $\bra{\cdot}_m$ its dual, and by $\scp{\cdot}{\cdot}_m$ the scalar product only in the $m$-th variable. Given a function $h:\RRR^d\to\RRR$ we introduce $h_{12}:\RRR^d \times \RRR^d \to \RRR, h_{12}(x_1,x_2)=h(x_1-x_2)$. In general, subscripts $i_1,\ldots,i_a$ usually indicate that an operator acts only on $x_{i_1},\ldots,x_{i_a}$. We always denote by $B_R(x) = \big\{ y \in \RRR^d: |x-y| < R \big\}$ the open ball with radius $R$ around $x$. For any set $\Omega \subset \RRR^d$ we write $\Omega^c = \RRR^d \setminus \Omega$. 

Note that the operator $p_m^{\varphi}$ from Definition~\ref{def:projectors} is indeed a projector on $L^2(\RRR^{3N})$ and that for $\varphi_i \perp \varphi_j$ and all $m,n=1,\ldots,N$ we have
\begin{equation}
p_m^{\varphi_i} p_m^{\varphi_j} = 0 \quad\text{and}\quad \left[ p_m^{\varphi_i}, p_n^{\varphi_j} \right] = 0.
\end{equation}
From that we conclude that $p_m$ and $q_m=1-p_m$ are projectors with
\begin{equation}
p_m q_m = 0 \quad\text{and}\quad [p_m,p_n] = [p_m,q_n] = [q_m,q_n] = 0,
\end{equation}
for all $m,n=1,\ldots,N$. An important property of $p_m^{\varphi}$ is that for antisymmetric $\psi_{\as} \in L^2(\RRR^{3N})$,
\begin{equation}\label{p_mp_n_as_0}
p_m^{\varphi}p_n^{\varphi} \psi_{\as} = 0
\end{equation}
for all $m \neq n$, since
\begin{align}
\left(p_m^{\varphi}p_n^{\varphi} \psi_{\as}\right)(x_1,\ldots,x_N) &= \varphi(x_m)\varphi(x_n) \int dx_m dx_n \varphi(x_m)^*\varphi(x_n)^* \psi_{\as}(\ldots,x_m,\ldots,x_n,\ldots) \nonumber \\
&= - \varphi(x_m)\varphi(x_n) \int dx_m dx_n \varphi(x_m)^*\varphi(x_n)^* \psi_{\as}(\ldots,x_n,\ldots,x_m,\ldots) \nonumber \\
&= - \varphi(x_m)\varphi(x_n) \int dx_n dx_m \varphi(x_n)^*\varphi(x_m)^* \psi_{\as}(\ldots,x_m,\ldots,x_n,\ldots).
\end{align}

For the operator norm of the projectors $p_m^{\varphi}$ it makes an important difference if it is calculated on all $L^2$ functions or only on antisymmetric functions in $L^2$, as the following lemma shows.
\begin{lemma}\label{lem:projector_norms}
For any $a=0,\ldots,N$, let $\psi_{\as}^{1,\ldots,a}\in L^2(\RRR^{3N})$ be normalized and antisymmetric in all variables except $x_1,\ldots,x_a$, and let $\varphi \in L^2(\RRR^3)$. Then, for all $m = a+1,\ldots,N$,
\begin{equation}\label{scp_a-N}
\SCP{\psi_{\as}^{1,\ldots,a}}{p_m^{\varphi}\psi_{\as}^{1,\ldots,a}}_{a+1,\ldots,N}(x_1,\ldots,x_a) \leq (N-a)^{-1} \, \SCP{\psi_{\as}^{1,\ldots,a}}{\psi_{\as}^{1,\ldots,a}}_{a+1,\ldots,N}(x_1,\ldots,x_a),
\end{equation}
for almost all $x_1,\ldots,x_a$ (with the definition of $\SCP{\cdot}{\cdot}_{a+1,\ldots,N}$ from \eqref{partial_scp}). 
\end{lemma}

It follows in particular that for all $m = a+1,\ldots,N$
\begin{equation}
\SCP{\psi_{\as}^{1,\ldots,a}}{p_m^{\varphi}\psi_{\as}^{1,\ldots,a}} \leq (N-a)^{-1}
\end{equation}
and, for all antisymmetric normalized $\psi_{\as}\in L^2(\RRR^{3N})$ and all $m=1,\ldots,N$,
\begin{equation}
\SCP{\psi_{\as}}{p_m^{\varphi}\psi_{\as}} \leq N^{-1}.
\end{equation}

\begin{proof}[Proof of Lemma \ref{lem:projector_norms}]
For this proof it is useful to think of $\psi_{\as}^{1,\ldots,a}(x_1,\ldots,x_N)$ as a function in the variables $x_{a+1},\ldots,x_N$, with fixed parameters $x_1,\ldots,x_a$. We can then use $\SCP{\psi_{\as}^{1,\ldots,a}}{\tilde{\psi}_{\as}^{1,\ldots,a}}_{a+1,\ldots,N}$ as scalar product in $x_{a+1},\ldots,x_N$ and define correspondingly $\big|\big|\psi_{\as}^{1,\ldots,a}\big|\big|_{a+1,\ldots,N}^2 = \SCP{\psi_{\as}^{1,\ldots,a}}{\psi_{\as}^{1,\ldots,a}}_{a+1,\ldots,N}$. For ease of notation we do not explicitly write out the dependence on $x_1,\ldots,x_a$ for this proof. Using the antisymmetry of $\psi_{\as}^{1,\ldots,a}$ in $x_{a+1},\ldots,x_N$, Cauchy-Schwarz and \eqref{p_mp_n_as_0}, we find
\begin{align}\label{antisymm1a_estimate}
& \SCP{\psi_{\as}^{1,\ldots,a}}{p_m^{\varphi}\psi_{\as}^{1,\ldots,a}}_{a+1,\ldots,N} \nonumber \\
&\qquad= (N-a)^{-1} 
\bigSCP{\psi_{\as}^{1,\ldots,a}}{\sum_{n=a+1}^N p_n^{\varphi}\psi_{\as}^{1,\ldots,a}}_{a+1,\ldots,N} \nonumber \\
&\qquad\leq (N-a)^{-1} \norm[a+1,\ldots,N]{\psi_{\as}^{1,\ldots,a}} \norm[a+1,\ldots,N]{\sum_{n=a+1}^N p_n^{\varphi}\psi_{\as}^{1,\ldots,a}} \nonumber \\
&\qquad= (N-a)^{-1} \sqrt{\SCP{\psi_{\as}^{1,\ldots,a}}{\psi_{\as}^{1,\ldots,a}}_{a+1,\ldots,N}} \sqrt{\bigSCP{\psi_{\as}^{1,\ldots,a}}{\sum_{\ell=a+1}^N p_{\ell}^{\varphi} \sum_{n=a+1}^N p_n^{\varphi}\psi_{\as}^{1,\ldots,a}}_{a+1,\ldots,N}} \nonumber \\
&\qquad= (N-a)^{-1/2} \sqrt{\SCP{\psi_{\as}^{1,\ldots,a}}{\psi_{\as}^{1,\ldots,a}}_{a+1,\ldots,N}} \sqrt{\SCP{\psi_{\as}^{1,\ldots,a}}{p_m^{\varphi}\psi_{\as}^{1,\ldots,a}}_{a+1,\ldots,N}} ~,
\end{align}
which yields \eqref{scp_a-N}.
\end{proof}

\section{Proof of Results About Density Matrices}\label{sec:density_matrices}
In this section, we prove the results from Section~\ref{sec:dens_mat_summary} about the relation of $\alpha_f$ to the reduced one-particle density matrices of $\psi$ and of the antisymmetrized product state $\bigwedge_{j=1}^N \varphi_j$. For a bounded operator $A$ on a Hilbert space $\Hilbert$, we define the operator norm, trace norm and Hilbert-Schmidt norm by
\begin{equation}
\norm[\op]{A} := \sup_{\psi \in \Hilbert, \norm{\psi}=1} \norm{A\psi}, ~\quad \norm[\tr]{A} := \tr|A| = \sum_i \scp{\phi_i}{|A| \phi_i}, ~\quad \norm[\HS]{A} := \sqrt{\tr(A^*A)},
\end{equation}
where $\{ \phi_i \}_{i\in\NNN}$ is some orthornormal basis and $A^*$ denotes the adjoint of $A$. Note that (for proofs, see, e.g., \cite[chapter~VI]{reedsimon1:1980})
\begin{equation}
\norm[\op]{A} \leq \norm[\HS]{A} \leq \norm[\tr]{A}
\end{equation}
and
\begin{equation}\label{tr_HS_op_ineqs}
\norm[\tr]{AB} \leq \norm[\op]{A} \norm[\tr]{B},~~~ \norm[\HS]{AB} \leq \norm[\op]{A} \norm[\HS]{B},~~~ \norm[\tr]{AB} \leq \norm[\HS]{A} \norm[\HS]{B},
\end{equation}
which we frequently use in the following proof.

\begin{proof}[Proof of Lemma \ref{lem:density_conv}]
First, note that reduced one-particle density matrices have the well-known properties that they are non-negative, i.e., $\big\langle f ,\gamma^{\psi}_1 f \big\rangle \geq 0$ $\forall f \in L^2(\RRR^3)$, that $\big|\big| \gamma^{\psi}_1 \big|\big|_{\tr} = 1$, and that for antisymmetric $\psi \in L^2(\RRR^{3N})$, $\big|\big| \gamma^{\psi}_1 \big|\big|_{\op} \leq N^{-1}$. Recall that we are here concerned with $\alpha_n = \SCP{\psi}{q_1\psi}$, i.e., the $\alpha$-functional with the weight $n(k)=k/N$. Also, recall that $\gamma^{\bigwedge \varphi_j}_1 = N^{-1} p_1$. We first show that $N^{-1}p_1 - p_1 \gamma^{\psi}_1 p_1$ is a non-negative operator with trace norm $\alpha_n$. Note that the operator $p_1\gamma^{\psi}_1p_1$ maps the $N$-dimensional subspace $\Span(\varphi_1,\ldots,\varphi_N)$ to itself, and is non-negative and self-adjoint. We can therefore diagonalize it, i.e., there is an orthonormal basis $\{ \chi_1,\ldots,\chi_N \}$, such that
\begin{equation}
p_1\gamma^{\psi}_1p_1 = \sum_{i=1}^N \lambda_i \, \ketbra{\chi_i}{\chi_i}_1 = \sum_{i=1}^N \lambda_i \, p_1^{\chi_i},
\end{equation}
with $0 \leq \lambda_i \leq N^{-1}$ $\forall i=1,\ldots,N$, since (see also Lemma~\ref{lem:projector_norms})
\begin{equation}
\lambda_i = \scp{\chi_i}{p_1\gamma^{\psi}_1p_1 \chi_i} = \SCP{\psi}{\ketbra{\chi_i}{\chi_i}_1\psi} \leq N^{-1} \SCP{\psi}{\psi} \leq N^{-1}.
\end{equation}
Therefore, $N^{-1} p_1 - p_1\gamma^{\psi}_1p_1 = \sum_{i=1}^N \left(N^{-1} - \lambda_i\right) p_1^{\chi_i}$ is non-negative and 
\begin{equation}\label{trace_mu_phi-mu_psi}
\norm[\tr]{N^{-1} p_1 - p_1\gamma^{\psi}_1p_1} = \tr\left(N^{-1} p_1 - p_1\gamma^{\psi}_1p_1\right) = 1 - \SCP{\psi}{p_1 \psi} = \alpha_n.
\end{equation}

We now show $\norm[\tr]{\gamma^{\bigwedge \varphi_j}_1 - \gamma^{\psi}_1} \leq \sqrt{8 \alpha_n}$. Note that the operators $\gamma_1^{\psi}$, $p_1\gamma_1^{\psi}p_1$ and $q_1\gamma_1^{\psi}q_1$ are non-negative, and that
\begin{equation}\label{trace_pmup_qmuq}
\norm[\tr]{p_1\gamma^{\psi}_1p_1} = \SCP{\psi}{p_1 \psi} = 1-\alpha_n \quad\text{and}\quad \norm[\tr]{q_1\gamma^{\psi}_1q_1} = \SCP{\psi}{q_1 \psi} = \alpha_n.
\end{equation}
By inserting two identities $1=p_1+q_1$ we find, using the triangle inequality, \eqref{trace_mu_phi-mu_psi}, \eqref{trace_pmup_qmuq} and \eqref{tr_HS_op_ineqs},
\begin{align}\label{estimate_trace_norm_pq}
\norm[\tr]{\gamma^{\bigwedge \varphi_j}_1 - \gamma^{\psi}_1} &\leq \norm[\tr]{N^{-1} p_1 - p_1\gamma^{\psi}_1p_1} + \norm[\tr]{p_1\gamma^{\psi}_1q_1} + \norm[\tr]{q_1\gamma^{\psi}_1p_1} + \norm[\tr]{q_1\gamma^{\psi}_1q_1} \nonumber \\
&= \alpha_n + \norm[\tr]{p_1\sqrt{\gamma^{\psi}_1}\sqrt{\gamma^{\psi}_1}q_1} + \norm[\tr]{q_1\sqrt{\gamma^{\psi}_1}\sqrt{\gamma^{\psi}_1}p_1} + \alpha_n \nonumber \\
&\leq 2 \alpha_n + 2\norm[\HS]{\sqrt{\gamma^{\psi}_1}p_1} \norm[\HS]{\sqrt{\gamma^{\psi}_1}q_1} \nonumber \\
&= 2 \alpha_n + 2 \sqrt{\norm[\tr]{p_1\gamma^{\psi}_1p_1} \norm[\tr]{q_1\gamma^{\psi}_1q_1}} \nonumber \\
&= 2 \alpha_n + 2 \sqrt{\alpha_n (1-\alpha_n)}.
\end{align}
Since $0 \leq \alpha_n \leq 1$, it is indeed true that
\begin{equation}\label{alpha_sqrt_ineq}
2 \alpha_n + 2 \sqrt{\alpha_n (1-\alpha_n)} \leq \sqrt{8\alpha_n},
\end{equation}
since the continuous function $f(\alpha)= \sqrt{8\alpha} - 2 \alpha - 2 \sqrt{\alpha (1-\alpha)}$ has its only minimum at $\alpha=1/2$ with $f(1/2)=0$, and also $f(0)=f(1) \geq 0$, thus $f(\alpha)\geq 0$ for all $\alpha\in [0,1]$, i.e., \eqref{alpha_sqrt_ineq} holds.

We now show $2 \alpha_n \leq \norm[\tr]{\gamma^{\bigwedge \varphi_j}_1 - \gamma^{\psi}_1}$. We find, using \eqref{trace_mu_phi-mu_psi}, \eqref{trace_pmup_qmuq}, \eqref{tr_HS_op_ineqs}, $\tr(q_1\gamma^{\psi}_1p_1) = \tr(p_1\gamma^{\psi}_1q_1) = 0$ and $|\tr(A)|\leq \norm[\tr]{A}$, that
\begin{align}\label{alpha_leq_tr}
2 \alpha_n &= \tr\left( \gamma^{\bigwedge \varphi_j}_1 - p_1\gamma^{\psi}_1p_1 \right) + \tr\left( q_1\gamma^{\psi}_1q_1 \right) \nonumber \\
&= \tr\left( \left(\gamma^{\bigwedge \varphi_j}_1 - \gamma^{\psi}_1\right) (p_1-q_1) \right) \nonumber \\
&\leq \norm[\tr]{ \left(\gamma^{\bigwedge \varphi_j}_1 - \gamma^{\psi}_1\right) (p_1-q_1) } \nonumber \\
&\leq \norm[\tr]{ \gamma^{\bigwedge \varphi_j}_1 - \gamma^{\psi}_1 } \Big|\Big|p_1-q_1\Big|\Big|_{\op} \nonumber \\
&= \norm[\tr]{ \gamma^{\bigwedge \varphi_j}_1 - \gamma^{\psi}_1 }.
\end{align}
Note that indeed $\norm[\op]{p_1-q_1}=1$, since for all $f \in L^2(\RRR^3)$,
\begin{equation}
\norm{(p_1-q_1)f}^2 = \scp{f}{(p_1-q_1)^2f}= \scp{f}{(p_1+q_1)f}=\norm{f}^2.
\end{equation}

We now show $\norm[\HS]{\gamma^{\bigwedge \varphi_j}_1 - \gamma^{\psi}_1}^2 \leq 2N^{-1} \alpha_n$. Recall that $\big|\big| \gamma^{\psi}_1 \big|\big|_{\op} \leq N^{-1}$ and $\big|\big| \gamma^{\psi}_1 \big|\big|_{\tr} = 1$. We find, using \eqref{tr_HS_op_ineqs},
\begin{align}\label{HS_leq_alpha}
\norm[\HS]{\gamma^{\bigwedge \varphi_j}_1 - \gamma^{\psi}_1}^2 &= \tr\left( \left(\gamma^{\bigwedge \varphi_j}_1 - \gamma^{\psi}_1\right)^2 \right) \nonumber \\
&= N^{-1} - N^{-1}\SCP{\psi}{p_1 \psi} - N^{-1}\SCP{\psi}{p_1 \psi} + \tr\left( \left(\gamma^{\psi}_1\right)^2 \right) \nonumber \\
&\leq N^{-1} \left( 1 - \SCP{\psi}{p_1 \psi} \right) - N^{-1}\SCP{\psi}{p_1 \psi} + \norm[\op]{\gamma^{\psi}_1}\norm[\tr]{\gamma^{\psi}_1} \nonumber \\
&\leq N^{-1} \alpha_n - N^{-1}\SCP{\psi}{p_1 \psi} + N^{-1} \nonumber \\
&= 2N^{-1} \alpha_n.
\end{align}

We now show $\alpha_n \leq \sqrt{N} \norm[\HS]{\gamma^{\bigwedge \varphi_j}_1 - \gamma^{\psi}_1}$. Using first \eqref{trace_mu_phi-mu_psi} and then \eqref{tr_HS_op_ineqs}, $\norm[\op]{p_1}=1$ and $\norm[\HS]{p_1}=\sqrt{N}$, we find
\begin{align}\label{alpha_leq_HS}
\alpha_n &= \norm[\tr]{p_1\left(\gamma^{\bigwedge \varphi_j}_1 - \gamma^{\psi}_1\right)p_1} \nonumber \\
&\leq \norm[\tr]{p_1\left(\gamma^{\bigwedge \varphi_j}_1 - \gamma^{\psi}_1\right)} \Big|\Big| p_1 \Big|\Big|_{\op} \nonumber \\
&\leq \Big|\Big| p_1 \Big|\Big|_{\HS} \norm[\HS]{\left(\gamma^{\bigwedge \varphi_j}_1 - \gamma^{\psi}_1\right)} \nonumber \\
&\leq \sqrt{N} \norm[\HS]{\left(\gamma^{\bigwedge \varphi_j}_1 - \gamma^{\psi}_1\right)}.
\end{align}
\end{proof}

\begin{proof}[Proof of Lemma \ref{lem:density_conv}]
The inequalities \eqref{bound_tr_HS_alpha_f} follow directly from $\alpha_n\leq \alpha_f$ and Lemma~\ref{lem:density_conv}. Let us denote the floor function by $\floor{\cdot}$, i.e., for any $x \in \RRR$, $\floor{x} = \max\{m \in \ZZZ: m \leq x \}$. Then
\begin{equation}
\alpha_{m^{(\gamma)}} = N^{1-\gamma} \sum_{k=0}^{\floor{N^{\gamma}}} \frac{k}{N} \SCP{\psi}{P^{(N,k)}\psi} + \sum_{k=\floor{N^{\gamma}}+1}^N \SCP{\psi}{P^{(N,k)}\psi} \leq N^{1-\gamma} \alpha_n.
\end{equation}
\end{proof}

\section{The Time Derivative of $\alpha_f(t)$}\label{sec:alpha_dot}
The expression for the time derivative of
\begin{equation}
\alpha_f(t) = \sum_{k=0}^N f(k) \SCP{\psi^t}{P^{(N,k)} \psi^t}
\end{equation}
for arbitrary weight functions $f(k)$ follows from direct calculation (recall that the projectors $P^{(N,k)}$ depend on $t$ through their dependence on $\varphi_1^t,\ldots,\varphi_N^t$). We use that the time derivative of $\psi^t \in L^2(\RRR^{3N})$ is given by the Schr\"odinger equation \eqref{Schr_intro} with self-adjoint Hamiltonian $H$ given by \eqref{Schr_intro_H} with real interaction $v^{(N)}(x)=v^{(N)}(-x)$. As effective equations for the one-particle wave functions $\varphi^t_1,\ldots,\varphi^t_N \in L^2(\RRR^3)$ we consider the general equations
\begin{equation}\label{mean-field_eq}
i \partial_t \varphi_j^t(x) = H^\mf \varphi_j^t(x) = H^0\varphi_j^t(x) + \left(V^{(N)}\varphi_j^t\right)(x),
\end{equation}
where the operator $V^{(N)}$ can possibly depend on $N$, $j$ and $\varphi_1^t,\ldots,\varphi_N^t$. The two interesting cases are when there is only \emph{direct interaction},
\begin{equation}\label{mean-field_dir_int}
V^{(N)}(x) = V^{\dir,(N)}(x) = (v^{(N)} * \rho_N^t)(x),
\end{equation}
where $\rho_N^t = \sum_{i=1}^N |\varphi^t_i|^2$, and when there is direct and \emph{exchange interaction},
\begin{align}\label{mean-field_direxch_int}
V^{(N)}\varphi_j^t(x) &= \left( V^{\dir,(N)} + V^{\exch,(N)} \right)\varphi_j^t(x) \nonumber \\
&= (v^{(N)} * \rho_N^t)(x)\varphi_j^t(x) - \sum_{\ell=1}^N \left(v^{(N)} * (\varphi_{\ell}^{t*}\varphi_j^t)\right)(x) \, \varphi_{\ell}^t(x).
\end{align}

We calculate the time derivative of $\alpha_f(t)$ in two steps. First, in Lemma~\ref{lem:alpha_derivative_pre} we directly calculate a straightforward expression for the time derivative. We then prove an auxiliary lemma which we use to bring the time derivative into a form (Lemma~\ref{lem:alpha_derivative}) that can nicely be estimated (later in Lemma~\ref{lem:estimates_terms_alpha_dot_beta}).

\begin{lemma}\label{lem:alpha_derivative_pre}
Let $\psi^t\in L^2(\RRR^{3N})$ be an antisymmetric solution to the Schr\"odinger equation \eqref{Schr_intro} and let $\varphi_1^t,\ldots,\varphi_N^t\in L^2(\RRR^3)$ be orthonormal solutions to the effective equations \eqref{mean-field_eq}. Then
\begin{equation}\label{alpha_dot_commutator}
\partial_t \alpha_f(t) = \frac{i}{2} \bigSCP{\psi^t}{\left[ W_{12}, \widehat{f} \,\right] \psi^t},
\end{equation}
where $W_{12} := N(N-1)v^{(N)}_{12} - NV^{(N)}_1 - NV^{(N)}_2$.
\end{lemma}

\begin{proof}
Note that the operators $p_m,q_m,P^{(N,k)}$ all depend on $t$ through the orbitals $\varphi_1^t,\ldots,\varphi_N^t$. For ease of notation, we do not explicitly write out this $t$-dependence. Let us first prove that $\widehat{f}$ fulfills the Heisenberg equation of motion
\begin{equation}\label{f_dot}
i \partial_t \widehat{f} = \left[ \, \sum_{m=1}^N H_m^{\mf} , \widehat{f} \, \right].
\end{equation}
Note that $i \partial_t p_m = \left[ H_m^{\mf}, p_m \right]$ and, using $p_m+q_m=1$, $i \partial_t q_m = \left[ H_m^{\mf}, q_m \right]$. Then, by applying the product rule and using $\left[ H_m^{\mf}, p_n \right] = 0 ~\forall m \neq n$, it follows that
\begin{equation}
i \partial_t P^{(N,k)} = \left[ \sum_{m=1}^N H_m^{\mf}, P^{(N,k)} \right],
\end{equation}
which implies \eqref{f_dot}. Using this and the antisymmetry of $\psi^t$, we find
\begin{align}
\partial_t \alpha_f(t) &= \partial_t \bigSCP{\psi^t}{\widehat{f}\,\psi^t} \nonumber \\
&= \bigSCP{\left(\partial_t\psi^t\right)}{\widehat{f}\,\psi^t} + \bigSCP{\psi^t}{\widehat{f}\left(\partial_t\psi^t\right)} + \bigSCP{\psi^t}{\left(\partial_t\widehat{f}\right)\psi^t} \nonumber \\
&= i\, \bigSCP{H\psi^t}{\widehat{f}\,\psi^t} - i \, \bigSCP{\psi^t}{\widehat{f} \, H\psi^t} - i \, \bigSCP{\psi^t}{\left[\sum_{m=1}^N H_m^{\mf},\widehat{f}\right]\psi^t} \nonumber \\
&= i \, \bigSCP{\psi^t}{\left[H-\sum_{m=1}^N H_m^{\mf},\widehat{f}\right]\psi^t} \nonumber \\
&= i \, \bigSCP{\psi^t}{\left[\sum_{j=1}^N H_j^0 + \sum_{1\leq i < j \leq N} v^{(N)}(x_i-x_j) - \sum_{m=1}^N \left( H_m^0 + V_m^{(N)} \right),\widehat{f}\right]\psi^t} \nonumber \\
&= i \, \bigSCP{\psi^t}{\left[ \frac{N(N-1)}{2} v^{(N)}(x_1-x_2) - \frac{N}{2} V_1^{(N)} - \frac{N}{2} V_2^{(N)},\widehat{f}\right]\psi^t}.
\end{align}
\end{proof}

Let us now bring the expression \eqref{alpha_dot_commutator} into a form we can use for the desired Gronwall estimate. For that, we need to define shifted operators and discrete derivatives of weight functions $f$.

\begin{definition}\label{def:f_derivative}
Let $d \in \ZZZ$ and $f$ be a weight function as in Definition~\ref{def:projectors}.
\begin{enumerate}[(a)]
\item We define the shifted operators
\begin{equation}
\widehat{f}_d := \sum_{k=-d}^{N-d} f(k+d) P^{(N,k)} = \sum_{k=0}^N f(k) P^{(N,k-d)} = \sum_{k=0}^N f(k+d) P^{(N,k)},
\end{equation}
where for the last step we defined $f(k)=0$ for all $k<0$ and $k>N$.
\item For $d>0$, we define discrete first derivatives of $f$ by
\begin{equation}
f'^{(d)}(k) := \Big( f(k)-f(k-d) \Big) \id_{\{ d,\ldots,N \}}(k),
\end{equation}
\begin{equation}
f'^{(-d)}(k) := \Big( f(k+d)-f(k) \Big) \id_{\{ 0,\ldots,N-d \}}(k),
\end{equation}
where $\id_{A}(k) = 1$ for $k \in A$ and $0$ otherwise.
\end{enumerate}
\end{definition}

For Lemma~\ref{lem:shift_fhat} and the proof of Lemma~\ref{lem:alpha_derivative} we need some more notation.

\begin{definition}\label{def:projectors2}
Let $\varphi_1, \ldots, \varphi_N \in L^2(\RRR^3)$ be orthonormal. For all $a,n \in \NNN$ with $a \leq n < N$, $\{ i_1,\ldots,i_n \} \subset \{ 1,\ldots, N\}$ we define
\begin{equation}
P_{i_1 \ldots i_n}^{(a)} := \left(\prod_{m=1}^a q_{i_m} \prod_{m=a+1}^n p_{i_m}\right)_{\sym},
\end{equation}
where the subscript $\sym$ means the symmetrized tensor product.
\end{definition}

In accordance with Definition~\ref{def:projectors} we have $P^{(k)}_{1 \ldots N} = P^{(N,k)}$. For the proof of Lemma~\ref{lem:alpha_derivative} we need an auxiliary Lemma that shows how to shift an $\widehat{f}$ to the other side of a two-particle operator.

\begin{lemma}\label{lem:shift_fhat}
As in Definition~\ref{def:projectors2}, we abbreviate $P_{12}^{(0)} = p_1p_2$, $P_{12}^{(1)} = p_1q_2 + q_1p_2$ and $P_{12}^{(2)} = q_1q_2$. Let $h_{12}$ be an operator that acts only on the first and second particle index. Then, for all $a,b = 0,1,2$,
\begin{equation}\label{shift_f_right_to_left}
\left( P_{12}^{(a)}h_{12}P_{12}^{(b)} \right) \widehat{f} = \widehat{f}_{b-a} \left( P_{12}^{(a)}h_{12}P_{12}^{(b)} \right),
\end{equation}
\begin{equation}\label{shift_f_left_to_right}
\widehat{f} \left( P_{12}^{(a)}h_{12}P_{12}^{(b)} \right) = \left( P_{12}^{(a)}h_{12}P_{12}^{(b)} \right) \widehat{f}_{a-b}.
\end{equation}
\end{lemma}

\begin{proof}
We only prove \eqref{shift_f_right_to_left} since \eqref{shift_f_left_to_right} can be proven in just the same way. In the following calculation we use the splitting 
\begin{equation}\label{def_split_P}
P^{(N,k)} = \sum_{d=0}^2 P_{12}^{(d)} P_{3 \ldots N}^{(k-d)},
\end{equation}
where $P_{12}^{(d)}$ contains exactly $d$ $q$-projectors, and $P_{3 \ldots N}^{(k-d)}$ contains $k-d$ $q$-projectors and acts only on the variables $x_3,\ldots x_N$. Then we find
\begin{align}
\left( P_{12}^{(a)}h_{12}P_{12}^{(b)} \right) \widehat{f}
&= \sum_{k=0}^N f(k)~ P_{12}^{(a)}h_{12}P_{12}^{(b)} \sum_{d=0}^2 P_{12}^{(d)} P_{3 \ldots N}^{(k-d)} \nonumber \\
\eqexp{by $P_{12}^{(b)}P_{12}^{(d)}=\delta_{bd}P_{12}^{(b)}$}&= \sum_{k=0}^N f(k)~ P_{12}^{(a)}h_{12}P_{12}^{(b)} P_{3 \ldots N}^{(k-b)} \nonumber \\
&= \sum_{k=0}^N f(k)~ P_{3 \ldots N}^{(k-b)} P_{12}^{(a)}h_{12}P_{12}^{(b)} \nonumber \\
\eqexp{by $P_{12}^{(d)}P_{12}^{(a)}=\delta_{ad}P_{12}^{(a)}$}&= \sum_{k=0}^N f(k)~ \sum_{d=0}^2 P_{12}^{(d)} P_{3 \ldots N}^{(k+a-b-d)} P_{12}^{(a)}h_{12}P_{12}^{(b)} \nonumber \\
&= \sum_{k=0}^N f(k)~ P^{(N,k+a-b)} P_{12}^{(a)}h_{12}P_{12}^{(b)} \nonumber \\
\eqexp{by Def.~\ref{def:f_derivative}}&= \widehat{f}_{b-a} \left( P_{12}^{(a)}h_{12}P_{12}^{(b)} \right).
\end{align}
\end{proof}

Note that Lemma~\ref{lem:shift_fhat} holds more generally for $m$-particle operators $h_{1 \ldots m}$, i.e.,
\begin{equation}\label{shift_f_1m}
\left( P_{1 \ldots m}^{(a)}h_{1 \ldots m}P_{1 \ldots m}^{(b)} \right) \widehat{f} = \widehat{f}_{b-a} \left( P_{1 \ldots m}^{(a)}h_{1 \ldots m}P_{1 \ldots m}^{(b)} \right)
\end{equation}
for all $a,b \in \{ 0,\ldots,m \}$. As last preparatory remark, note that for all $f\geq0$,
\begin{equation}\label{sqrt_f}
\widehat{f}^{1/2} = \widehat{f^{1/2}},
\end{equation}
since $P^{(N,k)}P^{(N,\ell)} = \delta_{k \ell}P^{(N,k)}$.\footnote{Note that for all $0 < s \in \QQQ$ also the more general relations $\left( \widehat{f} \right)^s = \widehat{f^s}$ and, if $f(0) = 0$ and $f(k)>0 ~\forall k>0$, $\widehat{f}^{-s}\left( \id - P^{(N,0)} \right) = \widehat{f^{-s}}$ hold.}

With Lemma~\ref{lem:shift_fhat} we can simplify the expression \eqref{alpha_dot_commutator} for the time derivative of $\alpha_f(t)$ by splitting it into three parts, each of which will be estimated separately later in Lemma~\ref{lem:estimates_terms_alpha_dot_beta}. The advantage of having the square root of the ``derivatives'' $\widehat{f'^{(d)}}$ next to $\psi^t$ is made clear in Lemma~\ref{lem:q_root_f} and the estimates in Lemma~\ref{lem:estimates_terms_alpha_dot_beta}.

\begin{lemma}\label{lem:alpha_derivative}
Let $\psi^t\in L^2(\RRR^{3N})$ be an antisymmetric solution to the Schr\"odinger equation \eqref{Schr_intro} and let $\varphi_1^t,\ldots,\varphi_N^t\in L^2(\RRR^3)$ be orthonormal solutions to the effective equations \eqref{mean-field_eq}. Then, for all monotone increasing $f(k)$,
\begin{equation}\label{alpha_derivative}
\partial_t \alpha_f(t) = (I)_f + (II)_f + (III)_f
\end{equation}
with 
\begin{align}
\label{alpha_dot_1}(I)_f &:= 2 N \, \Im\, \bigSCP{\psi^t}{\widehat{f'^{(1)}}^{1/2} \, \bigg[ q_1^t \left( (N-1)p_2^tv^{(N)}_{12}p_2^t - V^{(N)}_1 \right) p_1^t \bigg] \, \widehat{f'^{(-1)}}^{1/2}\psi^t} \\
\label{alpha_dot_2}(II)_f &:= N \, \Im\, \bigSCP{\psi^t}{\widehat{f'^{(2)}}^{1/2} \, \bigg[ q_1^tq_2^t (N-1)v^{(N)}_{12} p_1^tp_2^t \bigg] \, \widehat{f'^{(-2)}}^{1/2}\psi^t} \\
\label{alpha_dot_3}(III)_f &:= 2 N \, \Im\, \bigSCP{\psi^t}{\widehat{f'^{(1)}}^{1/2} \, \bigg[ q_1^tq_2^t (N-1)v^{(N)}_{12} p_1^tq_2^t \bigg] \, \widehat{f'^{(-1)}}^{1/2}\psi^t}.
\end{align}
\end{lemma}

\noindent\textbf{Remarks.}
\begin{enumerate}
\setcounter{enumi}{\theremarks}
\item Note that the time derivative is formally the same as for bosons, where $p_1:=\ketbr{\varphi}_1$, see \cite{pickl:2010gp_pos,pickl:2011method}. Note that in \cite{pickl:2010gp_pos,pickl:2011method} the splitting into three summands is done slightly differently: compared to \eqref{alpha_derivative}, an additional identity $1=p_2+q_2$ is added in front of $V^{(N)}_1$ and the operator $\widehat{f'^{(-d)}}^{1/2}$ is pulled over to the left side of the scalar product.
\item For the case $f(k) = n(k) = k/N$ we find a simple expression for the time derivative of $\alpha_n(t) = \SCP{\psi^t}{q_1^t \psi^t}$. It can easily be found by direct calculation as in \eqref{summary_dt_alpha} or from \eqref{alpha_derivative}. Note that, in view of Definition~\ref{def:projectors} and the identity $\sum_{k=0}^N P^{(N,k)} = 1$,
\begin{equation}
q_1 \widehat{f'^{(1)}}^{1/2} = q_1 \sum_{k=1}^N \left( \frac{k}{N} - \frac{(k-1)}{N} \right)^{1/2} P^{(N,k)} = N^{-1/2} ~ q_1 \sum_{k=1}^N P^{(N,k)} = q_1\, N^{-1/2},
\end{equation}
and similarly 
\begin{equation}
p_1 \widehat{f'^{(-1)}}^{1/2} = p_1\, N^{-1/2}, ~~~~ q_1q_2 \widehat{f'^{(2)}}^{1/2} = q_1q_2 \,\sqrt{2}\, N^{-1/2}, ~~~~ p_1p_2 \widehat{f'^{(-2)}}^{1/2} = p_1p_2 \,\sqrt{2}\, N^{-1/2}.
\end{equation}
Then \eqref{alpha_derivative} simplifies to
\begin{align}\label{alpha_derivative_n_remark}
\partial_t \alpha_n(t) &= 2 \, \Im\, \bigSCP{\psi^t}{q_1^t \left( (N-1)p_2^tv^{(N)}_{12}p_2^t - V^{(N)}_1 \right) p_1^t \psi^t} \nonumber \\*
&\quad + 2 \, \Im\, \bigSCP{\psi^t}{q_1^tq_2^t (N-1)v^{(N)}_{12} p_1^tp_2^t \psi^t} \nonumber \\*
&\quad + 2 \, \Im\, \bigSCP{\psi^t}{q_1^tq_2^t (N-1)v^{(N)}_{12} p_1^tq_2^t \psi^t}.
\end{align}
\item Note that the proof of Lemma~\ref{lem:alpha_derivative} can easily be generalized to $m$-particle interactions $W_{1 \ldots m}$. With the proper definition of $W_{1 \ldots m}$ coming from Lemma~\ref{lem:alpha_derivative_pre} and using \eqref{shift_f_1m}, we find
\begin{equation}
\partial_t \alpha_f(t) = \frac{i}{2} \bigSCP{\psi^t}{\left[W_{1 \ldots m},\widehat{f}\right]\psi^t} = \Im \sum_{a > b} \bigSCP{\psi^t}{ \widehat{f'^{(a-b)}}^{1/2}  P_{1 \ldots m}^{(a)} W_{1 \ldots m} P_{1 \ldots m}^{(b)} \widehat{f'^{(b-a)}}^{1/2} \psi^t}.
\end{equation}
\end{enumerate}
\setcounter{remarks}{\theenumi}

\begin{proof}[Proof of Lemma \ref{lem:alpha_derivative}]
We calculate the time derivative of $\alpha_f(t)$ starting from the expression \eqref{alpha_dot_commutator} of Lemma~\ref{lem:alpha_derivative_pre}. The idea of the proof is to insert two identities $1=p_1+q_1$ and $1=p_2+q_2$ in front of each $\psi^t$ (which leads to $16$ summands) and then to use Lemma~\ref{lem:shift_fhat} in order to shift $\widehat{f}$. It turns out that a lot of terms drop out due to the commutator structure. Let us first note two auxiliary calculations. For all $a,b = 0,1,2$ with $a>b$, we have
\begin{equation}\label{aux_f_1}
\left( \widehat{f} - \widehat{f}_{b-a} \right) P_{12}^{(a)} = \sum_{k=a-b}^N \Big( f(k) - f(k-(a-b)) \Big) P^{(N,k)} P_{12}^{(a)} = \widehat{f'^{(a-b)}} \geq 0,
\end{equation}
where the positivity is true since $f$ was assumed to be monotone increasing, and
\begin{align}\label{aux_f_2}
\widehat{\sqrt{f'^{(a-b)}}}_{a-b} &= \sum_{k=0}^N \Big( f(k+a-b) - f(k) \Big)^{1/2} \id_{\{ a-b,\ldots,N \}}(k+a-b) P^{(N,k)} \nonumber \\
&= \left( \sum_{k=0}^{N-(a-b)} \Big( f(k+a-b) - f(k) \Big) P^{(N,k)} \right)^{1/2} \nonumber \\
&= \widehat{f'^{(b-a)}}^{1/2}.
\end{align}
Then, by inserting two identities $1 = \sum_{a=0}^2 P_{12}^{(a)}$ (in the following sums, the indices $a,b$ always run from $0$ to $2$), we find
\begin{align}
\partial_t \alpha_f(t) &= \frac{i}{2} \bigSCP{\psi^t}{\left[W_{12},\widehat{f}\right]\psi^t} \nonumber \\
&= \frac{i}{2} \bigSCP{\psi^t}{ \sum_a P_{12}^{(a)} \Big(W_{12}\widehat{f} - \widehat{f}W_{12}\Big) \sum_b P_{12}^{(b)}\psi^t} \nonumber \\
\eqexp{by Lem.~\ref{lem:shift_fhat}}&= -\frac{i}{2} \sum_{a,b} \bigSCP{\psi^t}{ \Big( \widehat{f} - \widehat{f}_{b-a} \Big) P_{12}^{(a)} W_{12} P_{12}^{(b)}\psi^t} \nonumber \\
\eqexp{by \eqref{aux_f_1} and $a\leftrightarrow b$}&= \Im \sum_{a > b} \bigSCP{\psi^t}{ \widehat{f'^{(a-b)}} P_{12}^{(a)} W_{12} P_{12}^{(b)}\psi^t} \nonumber \\
\eqexp{by \eqref{sqrt_f}}&= \Im \sum_{a > b} \bigSCP{\psi^t}{ \widehat{f'^{(a-b)}}^{1/2} \widehat{\sqrt{f'^{(a-b)}}} P_{12}^{(a)} W_{12} P_{12}^{(b)}\psi^t} \nonumber \\
\eqexp{by Lem.~\ref{lem:shift_fhat}}&= \Im \sum_{a > b} \bigSCP{\psi^t}{ \widehat{f'^{(a-b)}}^{1/2}  P_{12}^{(a)} W_{12} P_{12}^{(b)} \widehat{\sqrt{f'^{(a-b)}}}_{a-b} \psi^t} \nonumber \\
\eqexp{by \eqref{aux_f_2}}&= \Im \sum_{a > b} \bigSCP{\psi^t}{ \widehat{f'^{(a-b)}}^{1/2}  P_{12}^{(a)} W_{12} P_{12}^{(b)} \widehat{f'^{(b-a)}}^{1/2} \psi^t} \nonumber \\
\eqexp{by $W_{12}=W_{21}$}&= 2 \, \Im\, \bigSCP{\psi^t}{\widehat{f'^{(1)}}^{1/2} \, q_1^t p_2^t W_{12} p_1^tp_2^t \, \widehat{f'^{(-1)}}^{1/2}\psi^t} \nonumber \\*
&\quad + \Im\, \bigSCP{\psi^t}{\widehat{f'^{(2)}}^{1/2} \, q_1^tq_2^t W_{12} p_1^tp_2^t \, \widehat{f'^{(-2)}}^{1/2}\psi^t} \nonumber \\*
&\quad + 2 \, \Im\, \bigSCP{\psi^t}{\widehat{f'^{(1)}}^{1/2} \, q_1^tq_2^t W_{12} p_1^tq_2^t \, \widehat{f'^{(-1)}}^{1/2}\psi^t}.
\end{align}
Equation~\eqref{alpha_derivative} follows by inserting the definition $W_{12} =  N(N-1)v^{(N)}_{12} - NV^{(N)}_1 - NV^{(N)}_2$, using $p_mq_m=0$ and using $p_1q_2=p_1-p_1p_2$ in the third summand.
\end{proof}

\section{Proof of Results for General Hamiltonians}\label{sec:proofs_general_thm}

\subsection{Using the New Weight Function}
In Section~\ref{sec:alpha_dot} we went through some trouble to handle general weight functions in the time derivative of $\alpha_f(t)$. The next lemma shows what we gain by choosing the weight function $m^{(\gamma)}(k)$ from \eqref{weight_m_gamma}. Morally, the lemma says that by choosing $\gamma < 1$ compared to $\gamma=1$, we make the convergence rate worse when there is no $q$ available (see \eqref{q_root_f_0q}), we do not loose anything when there is one $q$ available (see \eqref{q_root_f_1q}) and we gain powers in $N$ when there are two $q$'s available (see \eqref{q_root_f_2q}).

\begin{lemma}\label{lem:q_root_f}
Let $\widehat{m} = \widehat{m^{(\gamma)}} = \sum_{k=0}^N m^{(\gamma)}(k) P^{(N,k)}$ with $m^{(\gamma)}(k)$ as in \eqref{weight_m_gamma} for some $0<\gamma\leq 1$ and let $\psi \in L^2(\RRR^{3N})$ be antisymmetric and normalized. Then, for all $d =1,2$,
\begin{equation}\label{q_root_f_0q}
\norm{ \widehat{m'^{(\pm d)}}^{1/2}\psi}^2 \leq d N^{-\gamma},
\end{equation}
\begin{equation}\label{q_root_f_1q}
\norm{ q_1 \, \widehat{m'^{(-d)}}^{1/2}\psi}^2 \leq d N^{-1} \alpha_{m^{(\gamma)}} ~~~~\text{and}~~~~ \norm{ q_1 \, \widehat{m'^{(d)}}^{1/2}\psi}^2 \leq (d+1) N^{-1} \alpha_{m^{(\gamma)}},
\end{equation}
and
\begin{equation}\label{q_root_f_2q}
\norm{ q_1q_2 \, \widehat{m'^{(1)}}^{1/2}\psi}^2 \leq 3 N^{\gamma-2} \alpha_{m^{(\gamma)}}.
\end{equation}
\end{lemma}

\noindent\textbf{Remarks.}
\begin{enumerate}
\setcounter{enumi}{\theremarks}
\item Note that more generally the estimate
\begin{equation}
\norm{ q_1\ldots q_n \, \widehat{m'^{(\pm d)}} ^{1/2}\psi}^2 \leq C(d,n) \, N^{n(\gamma-1)-\gamma} \, \alpha_{m^{(\gamma)}}
\end{equation}
holds for all $1<d<N$ and $1<n<N$, for some constant $C$ that depends only on $d$ and $n$.
\end{enumerate}
\setcounter{remarks}{\theenumi}

\begin{proof}[Proof of Lemma \ref{lem:q_root_f}]
First, recall that for all antisymmetric $\phi \in L^2(\RRR^{3N})$,
\begin{equation}\label{q_alpha}
\bigSCP{\phi}{q_1 \phi} = N^{-1} \bigSCP{\phi}{\sum_{m=1}^N q_m \phi} = \sum_{k=0}^N N^{-1} \bigSCP{\phi}{\sum_{m=1}^N q_m P^{(N,k)} \phi} = \sum_{k=0}^N \frac{k}{N} \, \bigSCP{\phi}{P^{(N,k)} \phi},
\end{equation}
and in a similar way
\begin{equation}\label{qq_alpha}
\bigSCP{\phi}{q_1q_2 \phi} = \frac{1}{N(N-1)} \bigSCP{\phi}{\sum_{\substack{m,n = 1 \\ m \neq n}}^N q_mq_n \phi} = \sum_{k=0}^N \frac{k(k-1)}{N(N-1)} \, \bigSCP{\phi}{P^{(N,k)} \phi}.
\end{equation}
Recall that we denote the floor function by $\floor{\cdot}$, i.e., for any $x \in \RRR$, $\floor{x} = \max\{m \in \ZZZ: m \leq x \}$. Then we find
\begin{align}
\norm{ \widehat{m'^{(\pm d)}}^{1/2}\psi}^2 &= \bigSCP{\psi}{\widehat{m'^{(\pm d)}}\psi} \nonumber \\
&= \left\{\begin{array}{cl} \sum_{k=d}^N \big( m(k)-m(k-d) \big) \, \SCP{\psi}{P^{(N,k)} \psi} &, \text{for } + \\ \sum_{k=0}^{N-d} \big( m(k+d)-m(k) \big) \, \SCP{\psi}{P^{(N,k)} \psi} & , \text{for } - \end{array}\right. \nonumber \\
&\leq d N^{-\gamma} \, \bigSCP{\psi}{\sum_{k=0}^N P^{(N,k)} \psi} \nonumber \\
&\leq d N^{-\gamma},
\end{align}
and, by using \eqref{q_alpha},
\begin{align}
\norm{ q_1 \widehat{m'^{(\pm d)}}^{1/2}\psi}^2 &= \bigSCP{\psi}{q_1 \widehat{m'^{(\pm d)}}\psi} \nonumber \\
&= \left\{\begin{array}{cl} \sum_{k=d}^N \big( m(k)-m(k-d) \big) \frac{k}{N} \, \SCP{\psi}{P^{(N,k)} \psi} &, \text{for } + \\ \sum_{k=0}^{N-d} \big( m(k+d)-m(k) \big) \frac{k}{N} \, \SCP{\psi}{P^{(N,k)} \psi} & , \text{for } - \end{array}\right. \nonumber \\
&\leq \frac{d}{N} \, \left\{\begin{array}{cl} \sum_{k=d}^{\floor{N^{\gamma}}+d} \frac{k}{N^{\gamma}} \, \SCP{\psi}{P^{(N,k)} \psi} &, \text{for } + \\ \sum_{k=0}^{\floor{N^{\gamma}}} \frac{k}{N^{\gamma}} \, \SCP{\psi}{P^{(N,k)} \psi} & , \text{for } - \end{array}\right. \nonumber \\
&\leq C(d) \, N^{-1} \alpha_{m^{(\gamma)}},
\end{align}
and, by using \eqref{qq_alpha},
\begin{align}
\norm{ q_1q_2 \widehat{m'^{(\pm d)}}^{1/2}\psi}^2 &= \bigSCP{\psi}{q_1q_2 \widehat{m'^{(\pm d)}}\psi} \nonumber \\
&= \left\{\begin{array}{cl} \sum_{k=d}^N \big( m(k)-m(k-d) \big) \frac{k(k-1)}{N(N-1)} \, \SCP{\psi}{P^{(N,k)} \psi} &, \text{for } + \\ \sum_{k=0}^{N-d} \big( m(k+d)-m(k) \big) \frac{k(k-1)}{N(N-1)} \, \SCP{\psi}{P^{(N,k)} \psi} & , \text{for } - \end{array}\right. \nonumber \\
&\leq \left\{\begin{array}{cl} \frac{d(\floor{N^{\gamma}}+d-1)}{N(N-1)} \sum_{k=d}^{\floor{N^{\gamma}}+d} \frac{k}{N^{\gamma}} \, \SCP{\psi}{P^{(N,k)} \psi} &, \text{for } + \\ \frac{d(\floor{N^{\gamma}}-1)}{N(N-1)} \sum_{k=0}^{\floor{N^{\gamma}}} \frac{k}{N^{\gamma}} \, \SCP{\psi}{P^{(N,k)} \psi} & , \text{for } - \end{array}\right. \nonumber \\
&\leq \tilde{C}(d) \, N^{\gamma-2} \alpha_{m^{(\gamma)}}.
\end{align}
Careful consideration of the boundary terms around the $k=\floor{N^{\gamma}}$ summands gives the values of the constants $C(d)$ and $\tilde{C}(d)$.
\end{proof}

\subsection{Diagonalization of $p_2h_{12}p_2$ and Related Lemmas}\label{sec:estimates_projectors}
For handling the terms in the time derivative of $\alpha_f(t)$ from equation \eqref{alpha_derivative}, it is often useful to diagonalize operators $p_2h_{12}p_2$. Let us briefly summarize what we mean by that and introduce some notation on the way. For any $h: \RRR^3 \to [0, \infty)$ and $\varphi_1,\ldots,\varphi_N \in L^2(\RRR^3)$ such that $(h * \rho_N)(x) < \infty$ for all $x\in \RRR^3$ (with $\rho_N(x) = \sum_{i=1}^N |\varphi_i(x)|^2$), the operator $p_2h_{12}p_2$ is a multiplication operator in $x_1$ and a projector onto the $N$-dimensional subspace $\Span(\varphi_1,\ldots,\varphi_N)$ in the second variable. Therefore one can write it as an $x_1$ dependent non-negative self-adjoint $(N \times N)$-matrix acting on the second variable, i.e., (recall our notation from Section~\ref{sec:notation})
\begin{equation}\label{diagonalization}
p_2h_{12}p_2 = \sum_{i,j=1}^N \scp{\varphi_i}{h_{12}\,\varphi_j}_2(x_1) \, \ket{\varphi_i}\bra{\varphi_j}_2 = \sum_{i=1}^N \lambda_i(x_1) \, \ket{\chi_i^{x_1}} \bra{\chi_i^{x_1}}_2 = \sum_{i=1}^N \lambda_i(x_1) \, p_2^{\chi_i^{x_1}},
\end{equation}
where the eigenvectors $\ket{\chi_i^{x_1}}_2$ are orthonormal and can be written as $\ket{\chi_i^{x_1}}_2 = \sum_{k=1}^N U_{ik}(x_1) \, \ket{\varphi_k}_2$, where $U(x_1)$ is some unitary $(N\times N)$-matrix. Note that $\Span(\chi_1^{x_1},\ldots,\chi_N^{x_1}) = \Span(\varphi_1,\ldots,\varphi_N)$ for all $x_1 \in \RRR^3$ and that the projector $p_2$ is independent of the choice of basis, i.e.,
\begin{equation}
p_2 = \sum_{i=1}^N \ket{\varphi_i} \bra{\varphi_i}_2 = \sum_{i=1}^N \ket{\chi_i^{x_1}} \bra{\chi_i^{x_1}}_2.
\end{equation}
The eigenvalues $\lambda_i(x_1)$ have the properties that
\begin{equation}
\lambda_i(x_1) = \scp{\chi_i^{x_1}}{p_2h_{12}p_2\chi_i^{x_1}}_2(x_1) = \scp{\chi_i^{x_1}}{h_{12}\,\chi_i^{x_1}}_2(x_1) < \infty,
\end{equation}
\begin{equation}\label{diagonalization_EV_sum}
\sum_{i=1}^N \lambda_i(x_1) = \sum_{i=1}^N \scp{\chi_i^{x_1}}{h_{12} \, \chi_i^{x_1}}_2(x_1) = \sum_{j=1}^N \scp{\varphi_j}{h_{12} \, \varphi_j}_2(x_1) = (h * \rho_N)(x_1),
\end{equation}
and furthermore, for all $i \neq j$,
\begin{equation}
\scp{\chi_i^{x_1}}{h_{12}\,\chi_j^{x_1}}_2(x_1) = \scp{\chi_i^{x_1}}{p_2h_{12}p_2\,\chi_j^{x_1}}_2(x_1) = \scp{\chi_i^{x_1}}{\lambda_j(x_1)\,\chi_j^{x_1}}_2(x_1) = 0.
\end{equation}
We use this diagonalization in the proof of Lemma~\ref{lem:estimates_terms_alpha_dot_beta} for the operator $p_2v_{12}^{(N)}p_2$ from term $(I)_f$ from the time derivative of $\alpha_f(t)$ in \eqref{alpha_derivative}, and to prove the following lemma which we need in order to bound term $(II)_f$ and $(III)_f$ from \eqref{alpha_derivative}. Note that this lemma is similar to \eqref{tr_HS_op_ineqs}; the additional $N^{-1}$ factors come from Lemma~\ref{lem:projector_norms}.

\begin{lemma}\label{lem:psi_op_psi_diag_tr} 
Let $\varphi_1,\ldots,\varphi_N \in L^2(\RRR^3)$ be orthonormal, $h: \RRR^3 \to [0, \infty)$ and set $\rho_N(x) = \sum_{i=1}^N |\varphi_i(x)|^2$. Then, 
\begin{enumerate}[(a)]
\item for all $\psi \in L^2(\RRR^{3N})$ which are antisymmetric in the variables $x_2,\ldots,x_N$,
\begin{equation}\label{psi_pvp_psi}
\SCP{\psi}{p_2 \, h_{12}\, p_2 \psi} \leq (N-1)^{-1} \, \left(\sup_{y \in \RRR^3} (h * \rho_N)(y) \right) \, \SCP{\psi}{\psi},
\end{equation}
\item for all antisymmetric $\psi \in L^2(\RRR^{3N})$,
\begin{equation}\label{psi_ppvpp_psi}
\SCP{\psi}{p_1p_2 \, h_{12} \, p_1p_2 \psi} \leq (N(N-1))^{-1} \, \left(\int_{\RRR^3} (h * \rho_N)(y)\, \rho_N(y) \, d^3y\right) \, \SCP{\psi}{\psi},
\end{equation}
where this inequality remains true with $(N(N-1))^{-1}$ replaced by $((N-1)(N-2))^{-1}$, when $\psi$ is antisymmetric in all variables except $x_3$.
\end{enumerate}
\end{lemma}

\begin{proof}
Recall our notation for the ``partial scalar product'' from \eqref{partial_scp}. For proving \eqref{psi_pvp_psi} we use the diagonalization \eqref{diagonalization}, Lemma~\ref{lem:projector_norms} and \eqref{diagonalization_EV_sum}. We find
\begin{align}
\bigSCP{\psi}{p_2 \, h_{12}\, p_2 \psi} &= \int d^3x_1 \sum_{i=1}^N \underbrace{\lambda_i(x_1)}_{\geq 0} \underbrace{\bigSCP{\psi}{p_2^{\chi_i^{x_1}} \psi}_{2,\ldots,N}(x_1)}_{\geq 0 ~ \forall \, x_1} \nonumber \\
&\leq \int d^3x_1 \sum_{i=1}^N \lambda_i(x_1) \, (N-1)^{-1} \, \bigSCP{\psi}{\psi}_{2,\ldots,N}(x_1) \nonumber \\
&\leq (N-1)^{-1} \left( \sup_{x_1} \sum_{i=1}^N \lambda_i(x_1)\right) \int d^3x_1 \bigSCP{\psi}{\psi}_{2,\ldots,N}(x_1) \nonumber \\
&= (N-1)^{-1} \left( \sup_{x_1}\, (h*\rho_N)(x_1) \right) \bigSCP{\psi}{\psi}.
\end{align}
For proving \eqref{psi_ppvpp_psi} we diagonalize $p_1(h*\rho_N)(x_1)p_1 = p_1 \sum_{i=1}^N \lambda_i(x_1)p_1$. We call the eigenvalues $\mu_j$ and the eigenvectors $\tilde{\varphi}_j$. With the diagonalization \eqref{diagonalization} and Lemma~\ref{lem:projector_norms} we find
\begin{align}
\bigSCP{\psi}{p_1p_2 \, h_{12} \, p_1p_2 \psi} &= \int d^3x_1 \sum_{i=1}^N \lambda_i(x_1) \bigSCP{p_1\psi}{ p_2^{\chi_i^{x_1}} p_1\psi}_{2,\ldots,N}(x_1) \nonumber \\
&\leq (N-1)^{-1} \bigSCP{\psi}{p_1 \sum_{i=1}^N \lambda_i(x_1) p_1\psi} \nonumber \\
&= (N-1)^{-1} \bigSCP{\psi}{\sum_{j=1}^N \mu_j \, p_1^{\tilde{\varphi}_j} \psi} \nonumber \\
&\leq (N(N-1))^{-1} \left( \sum_{j=1}^N \mu_j \right) ~ \bigSCP{\psi}{\psi} \nonumber \\
&= (N(N-1))^{-1} \left( \int (h*\rho_N)(x)\, \rho_N(x) \,d^3x \right) ~ \bigSCP{\psi}{\psi}.
\end{align}
If $\psi$ is antisymmetric in all variables except $x_3$, then by Lemma~\ref{lem:projector_norms} one can only extract factors $(N-2)^{-1}$ instead of $(N-1)^{-1}$, and $(N-1)^{-1}$ instead of $N^{-1}$ from the antisymmetry of $\psi$.
\end{proof}

\subsection{Bounds on $\partial_t \alpha_f(t)$}\label{sec:alpha_m_dot_rigorous}
We now give the rigorous bounds for the three terms in the time derivative of $\alpha_f(t)$ given by \eqref{alpha_derivative}. Here, we use the weight function $m^{(\gamma)}(k)$ from \eqref{weight_m_gamma}. This also contains the case where $\gamma=1$, i.e., where the weight function is $n(k)$. The estimates are collected in the following lemma, which constitutes the heart of the proof of our main results.

We state this lemma only for positive $v^{(N)}$. If $v^{(N)}$ contains both positive and negative parts, we later decompose $v^{(N)} = v^{(N)}_{+} - v^{(N)}_{-}$, with $v^{(N)}_{+},v^{(N)}_{-} \geq 0$, and then estimate the three terms in \eqref{alpha_derivative} separately for $v^{(N)}_{+}$ and $v^{(N)}_{-}$.

\begin{lemma}\label{lem:estimates_terms_alpha_dot_beta}
Let $\varphi_1,\ldots,\varphi_N \in L^2(\RRR^3)$ be orthonormal and $\psi \in L^2(\RRR^{3N})$ be antisymmetric. Let $v^{(N)}$ be positive and set $\rho_N(x) = \sum_{i=1}^N |\varphi_i(x)|^2$. Let $V_1^{(N)} = (v*\rho_N)_1$ be the direct mean-field interaction. Let $\Omega_N \subseteq \RRR^3$, with $\Omega_N(x_1)=\Omega_N+x_1$ and $\Omega_N^c=\RRR^3 \setminus \Omega_N$ (possibly $\Omega_N=\RRR^3$ or $\Omega_N=\emptyset$). Then, using the weight function $m^{(\gamma)}(k)$ from \eqref{weight_m_gamma}, we find for all $0<\gamma\leq 1$ for the three terms from \eqref{alpha_derivative},
\begin{align}
\label{term_1}\big\lvert (I)_{m^{(\gamma)}} \big\rvert & \leq \sqrt{8} \left( \sup_{x_1\in\RRR^3} \int_{\Omega_N(x_1)}v^{(N)}(x_1-y)^2\rho_N(y)\,d^3y \right)^{1/2} N^{1/2} \, \Big( \alpha_{m^{(\gamma)}} + N^{-\gamma} \Big) & \nonumber \\*
&\quad + \sqrt{12} \left( \sup_{y\in\Omega_N^c} v^{(N)}(y) \right) N^{1/2+\gamma/2} \, \Big( \alpha_{m^{(\gamma)}} + N^{-\gamma} \Big), \\
\label{term_2}\big\lvert (II)_{m^{(\gamma)}} \big\rvert &\leq \sqrt{6} \left( \int \Big(\big(v^{(N)}\big)^2*\rho_N\Big)(y)\,\rho_N(y)\,d^3y \right)^{1/2} \Big( \alpha_{m^{(\gamma)}} + N^{-\gamma} \Big), \\
\label{term_3}\big\lvert (III)_{m^{(\gamma)}} \big\rvert &\leq \sqrt{12} \left( \sup_{y\in\RRR^3} \Big(\big(v^{(N)}\big)^2*\rho_N\Big)(y) \right)^{1/2} \, N^{\gamma/2} \, \alpha_{m^{(\gamma)}}.
\end{align}
\end{lemma}

\begin{proof}
\textbf{Term} $\boldsymbol{(I)_m.}$ In the following, we diagonalize the operator $p_mv^{(N)}_{1m}p_m$ and use the same notation as in \eqref{diagonalization}. We split each eigenvalue into two parts,
\begin{equation}
\lambda_i(x_1) = \underbrace{\int_{\Omega_N(x_1)} v^{(N)}(x_1-y) \left\lvert \chi_i^{x_1}(y) \right\rvert^2 \, d^3y}_{=:\lambda_{i}^{\Omega_N}(x_1)} + \underbrace{\int_{\Omega_N^c(x_1)} v^{(N)}(x_1-y) \left\lvert \chi_i^{x_1}(y) \right\rvert^2 \, d^3y}_{=:\lambda_{i}^{\Omega_N^c}(x_1)}.
\end{equation}
Note that by using Cauchy-Schwarz, $\sum_{i=1}^N \left\lvert \chi_i^{x_1}(y) \right\rvert^2 = \sum_{i=1}^N \left\lvert \varphi_i(y) \right\rvert^2$ and $\norm{\chi_i^{x_1}} = 1$, we find
\begin{align}\label{sup_v2}
\sum_{i=1}^N \lambda_{i}^{\Omega_N}(x_1)^2 &\leq \sum_{i=1}^N  \left[ \left( \int_{\Omega_N(x_1)} \left(v^{(N)}(x_1-y)\right)^2 \left\lvert \chi_i^{x_1}(y) \right\rvert^2 \, d^3y \right)^{1/2} \left( \int_{\Omega_N(x_1)} \left\lvert \chi_i^{x_1}(y) \right\rvert^2 \, d^3y \right)^{1/2}\right]^2 \nonumber \\
&\leq \int_{\Omega_N(x_1)}\left(v^{(N)}(x_1-y)\right)^2 \rho_N(y) \, d^3y,
\end{align}
and
\begin{equation}\label{lambda_omega_bar}
\lambda_{i}^{\Omega_N^c}(x_1) \leq \left( \sup_{y \in \Omega_N^c(x_1)} v^{(N)}(x_1-y) \right) \int_{\RRR^3} \left\lvert \chi_i^{x_1}(y) \right\rvert^2 \, d^3y  = \sup_{y \in \Omega_N^c} v^{(N)}(y).
\end{equation}
For bounding $(I)_m$ it is useful to introduce the projectors
\begin{equation}\label{q_neq_1}
p_{\neq 1}^{\varphi_i} := \sum_{m=2}^N p_m^{\varphi_i}, \quad q_{\neq 1}^{\varphi_i} := \id - p_{\neq 1}^{\varphi_i}
\end{equation}
that act on all but the first variable. From these projectors we only need the properties that for all $\psi_{\as}^1$ that are antisymmetric in all variables except $x_1$, $\left(q_{\neq 1}^{\chi_i^{x_1}}\right)^2 \psi_{\as}^1=q_{\neq 1}^{\chi_i^{x_1}} \psi_{\as}^1$, and that
\begin{equation}\label{p_neq_1}
\bigSCP{\psi_{\as}^1}{\sum_{i=1}^N q_{\neq 1}^{\varphi_i} \psi_{\as}^1} = \bigSCP{\psi_{\as}^1}{\left( N - \sum_{i=1}^N \sum_{m=2}^N p_m^{\varphi_i} \right) \psi_{\as}^1} = (N-1) \bigSCP{\psi_{\as}^1}{q_2 \psi_{\as}^1} + \bigSCP{\psi_{\as}^1}{\psi_{\as}^1},
\end{equation}
which remains true if $\ket{\varphi_i}_m=\ket{\chi_i^{x_1}}_m$ ($m\geq 2$). In the following we abbreviate $\phi = \widehat{m'^{(1)}}^{1/2}\psi$ and $\tilde{\phi} = \widehat{m'^{(-1)}}^{1/2}\psi$. Note that both $\phi$ and $\tilde{\phi}$ are still antisymmetric. Then we find
\begin{align}\label{term_a_estimate2}
\big\lvert (I)_m \big\rvert &=  2 N \, \bigg\lvert \Im\, \bigSCP{\phi}{q_1 \Big( (N-1)p_2v^{(N)}_{12}p_2 - V_1^{(N)} \Big) p_1 \tilde{\phi}} \bigg\rvert \nonumber \\
&= 2 N \, \bigg\lvert \Im\, \bigSCP{\phi}{q_1 \left( \sum_{m=2}^N p_mv^{(N)}_{1m}p_m - V_1^{(N)} \right) p_1 \tilde{\phi}} \bigg\rvert \nonumber \\
\eqexp{by \eqref{diagonalization}, \eqref{diagonalization_EV_sum}}&= 2 N \, \bigg\lvert \Im\, \bigSCP{\phi}{q_1 \Bigg( \sum_{i=1}^N \lambda_i(x_1) \sum_{m=2}^N p_m^{\chi_i^{x_1}} - \sum_{i=1}^N \lambda_i(x_1) \Bigg) p_1 \tilde{\phi}} \bigg\rvert \nonumber \\
&= 2 N \, \bigg\lvert \Im\, \bigSCP{\phi}{q_1 \Bigg( \sum_{i=1}^N \lambda_i(x_1) q_{\neq 1}^{\chi_i^{x_1}} \Bigg) p_1 \tilde{\phi}} \bigg\rvert \nonumber \\
&\leq 2 N \, \bigg\lvert \sum_{i=1}^N \bigSCP{\phi}{q_1 \lambda_{i}^{\Omega_N}(x_1) q_{\neq 1}^{\chi_i^{x_1}} p_1 \tilde{\phi}} \bigg\rvert + 2 N \, \bigg\lvert \sum_{i=1}^N \bigSCP{\phi}{q_1 \lambda_{i}^{\Omega_N^c}(x_1) q_{\neq 1}^{\chi_i^{x_1}} p_1 \tilde{\phi}} \bigg\rvert.
\end{align}
For the first summand in \eqref{term_a_estimate2}, we find by Cauchy-Schwarz,
\begin{align}\label{term_a_estimate_epsilon}
& 2 N \, \bigg\lvert \sum_{i=1}^N \bigSCP{\phi}{q_1 \lambda_{i}^{\Omega_N}(x_1) q_{\neq 1}^{\chi_i^{x_1}} p_1 \tilde{\phi}} \bigg\rvert \nonumber \\
&\qquad\leq 2 N \, \left( \sum_{i=1}^N \bigSCP{\phi}{q_1 \left(\lambda_{i}^{\Omega_N}(x_1)\right)^2 q_1 \phi} \right)^{1/2} \, \left( \sum_{i=1}^N \bigSCP{\tilde{\phi}}{p_1 q_{\neq 1}^{\chi_i^{x_1}} p_1 \tilde{\phi}} \right)^{1/2} \nonumber \\
\eqexpl{by \eqref{p_neq_1}} &\qquad\leq 2 N \, \left[ \left( \sup_{x_1} \sum_{i=1}^N \left(\lambda_{i}^{\Omega_N}(x_1)\right)^2 \right) \bigSCP{\phi}{q_1 \phi} \bigg( (N-1) \bigSCP{\tilde{\phi}}{p_1 q_2 \tilde{\phi}} + \bigSCP{\tilde{\phi}}{p_1 \tilde{\phi}} \bigg) \right]^{1/2} \nonumber \\
\eqexpl{by Lem.~\ref{lem:q_root_f}} &\qquad\leq 2 N \, \left( \sup_{x_1} \sum_{i=1}^N \left(\lambda_{i}^{\Omega_N}(x_1)\right)^2 \right)^{1/2} \bigg( 2 N^{-1} \alpha_{m} \bigg)^{1/2} \, \bigg( \alpha_{m} + N^{-\gamma} \bigg)^{1/2} \nonumber \\
\eqexpl{by \eqref{sup_v2}}&\qquad\leq \sqrt{8} \left( \sup_{x_1} \int_{\Omega_N(x_1)} \left(v^{(N)}(x_1-y)\right)^2 \rho_N(y) \, d^3y \right)^{1/2} N^{1/2} \bigg( \alpha_{m} + N^{-\gamma} \bigg).
\end{align}
For the second summand in \eqref{term_a_estimate2}, we find 
\begin{align}\label{term_a_estimate_epsilon_bar}
&2 N \, \Big\lvert \sum_{i=1}^N \bigSCP{\phi}{q_1 \lambda_{i}^{\Omega_N^c}(x_1) q_{\neq 1}^{\chi_i^{x_1}} p_1 \tilde{\phi}} \Big\rvert \nonumber \\*
&\qquad= 2 N \, \Big\lvert \sum_{i=1}^N \bigSCP{\phi}{q_1 q_{\neq 1}^{\chi_i^{x_1}} \lambda_{i}^{\Omega_N^c}(x_1) q_{\neq 1}^{\chi_i^{x_1}} p_1 \tilde{\phi}} \Big\rvert \nonumber \\
&\qquad\leq 2 N \, \left( \sum_{i=1}^N \bigSCP{\phi}{q_1 \lambda_{i}^{\Omega_N^c}(x_1) q_{\neq 1}^{\chi_i^{x_1}} q_1 \phi} \right)^{1/2} \left( \sum_{i=1}^N \bigSCP{\tilde{\phi}}{p_1 \lambda_{i}^{\Omega_N^c}(x_1) q_{\neq 1}^{\chi_i^{x_1}} p_1 \tilde{\phi}} \right)^{1/2} \nonumber \\
&\qquad\leq 2 N \, \left( \sup_{i,x_1} \lambda_{i}^{\Omega_N^c}(x_1) \right) \left( \sum_{i=1}^N \bigSCP{\phi}{q_1 q_{\neq 1}^{\chi_i^{x_1}} q_1 \phi} \right)^{1/2} \left( \sum_{i=1}^N \bigSCP{\tilde{\phi}}{p_1 q_{\neq 1}^{\chi_i^{x_1}} p_1 \tilde{\phi}} \right)^{1/2} \nonumber \\
\eqexpl{by \eqref{p_neq_1}} &\qquad\leq 2 N \, \left(\sup_{i,x_1} \lambda_{i}^{\Omega_N^c}(x_1)\right) \bigg( (N-1) \bigSCP{\phi}{q_1 q_2 \phi} + \bigSCP{\phi}{q_1 \phi} \bigg)^{1/2} \times \nonumber \\*
&\qquad \quad \times \bigg( (N-1) \bigSCP{\tilde{\phi}}{p_1 q_2 \tilde{\phi}} + \bigSCP{\tilde{\phi}}{p_1 \tilde{\phi}} \bigg)^{1/2} \nonumber \\
\eqexpl{by Lem.~\ref{lem:q_root_f}} &\qquad\leq 2 N \, \left(\sup_{i,x_1} \lambda_{i}^{\Omega_N^c}(x_1)\right) \bigg( 3 N^{\gamma-1} \alpha_{m} + 2 N^{-1} \alpha_{m} \bigg)^{1/2} \bigg( \alpha_{m} + N^{-\gamma} \bigg)^{1/2} \nonumber \\
\eqexpl{by \eqref{lambda_omega_bar}} &\qquad\leq \sqrt{12} \left( \sup_{y \in \Omega_N^c} v^{(N)}(y) \right) N^{\gamma/2+1/2} \bigg( \alpha_{m} + N^{-\gamma} \bigg).
\end{align}

\textbf{Term} $\boldsymbol{(II)_m.}$ We abbreviate $\phi = \widehat{m'^{(2)}}^{1/2}\psi$ and $\tilde{\phi} = \widehat{m'^{(-2)}}^{1/2}\psi$. The idea of the bound for this term is to shift one $q$ to the right side of the scalar product by using the antisymmetry of $\psi$. Using Cauchy-Schwarz and the antisymmetry of $\phi$ and $\tilde{\phi}$, we find
\begin{align}\label{term_b_estimate}
\big\lvert (II)_m \big\rvert &= N \left\lvert \Im\, \bigSCP{\phi}{q_1q_2 (N-1)v^{(N)}_{12} p_1p_2 \tilde{\phi}} \right\rvert \nonumber \\
&= N \bigg\lvert \, \Im\, \bigSCP{\phi}{q_1 \sum_{m=2}^N q_m \, v^{(N)}_{1m} \, p_1p_m \tilde{\phi}} \bigg\rvert \nonumber \\
&\leq  N \norm{q_1\phi} \norm{\sum_{m=2}^N q_m \, v^{(N)}_{1m} \, p_1p_m \tilde{\phi}} \nonumber \\
&= N \norm{q_1\phi} \bigg[ (N-1)(N-2) \bigSCP{\tilde{\phi}}{q_3 p_1 p_2 v^{(N)}_{12} v^{(N)}_{13} p_1 p_3 q_2 \tilde{\phi}} \nonumber \\*
&\quad + (N-1) \bigSCP{\tilde{\phi}}{p_1 p_2 v^{(N)}_{12}q_2v^{(N)}_{12} p_1 p_2 \tilde{\phi}} \bigg]^{1/2} \nonumber \\
\eqexp{by Lem.~\ref{lem:psi_op_psi_diag_tr}} &\leq  N \norm{q_1\phi} \bigg[ \int \Big(\big(v^{(N)}\big)^2*\rho_N\Big)(y)\,\rho_N(y)\,d^3y~ \bigg( \bigSCP{\tilde{\phi}}{q_3\tilde{\phi}} + N^{-1}\bigSCP{\tilde{\phi}}{\tilde{\phi}} \bigg) \bigg]^{1/2} \nonumber \\
\eqexp{by Lem.~\ref{lem:q_root_f}} &\leq  N \bigg[ 3 N^{-1} \alpha_{m^{(\gamma)}} \int \Big(\big(v^{(N)}\big)^2*\rho_N\Big)(y)\,\rho_N(y)\,d^3y~ \bigg( 2 N^{-1} \alpha_{m^{(\gamma)}} + 2 N^{-1}N^{-\gamma} \bigg) \bigg]^{1/2} \nonumber \\
&\leq \sqrt{6} \left( \int \Big(\big(v^{(N)}\big)^2*\rho_N\Big)(y)\,\rho_N(y)\,d^3y \right)^{1/2} \bigg( \alpha_{m} + N^{-\gamma} \bigg).
\end{align}

\textbf{Term} $\boldsymbol{(III)_m.}$ We abbreviate $\phi = \widehat{m'^{(1)}}^{1/2}\psi$ and $\tilde{\phi} = \widehat{m'^{(-1)}}^{1/2}\psi$. By Cauchy-Schwarz we find
\begin{align}
\big\lvert (III)_m \big\rvert &= 2N \, \left\lvert \Im\, \bigSCP{\phi}{q_1q_2 (N-1)v^{(N)}_{12} p_1q_2 \tilde{\phi}} \right\rvert \nonumber \\
&\leq 2N (N-1) \norm{q_1q_2\phi} \norm{v^{(N)}_{12}p_1q_2\tilde{\phi}} \nonumber \\
\eqexp{by Lem.~\ref{lem:psi_op_psi_diag_tr}} &\leq 2N (N-1) \norm{q_1q_2\phi} \bigg[ (N-1)^{-1} \left( \sup_{y} \Big(\big(v^{(N)}\big)^2*\rho_N\Big)(y) \right) \bigSCP{\tilde{\phi}}{q_2 \tilde{\phi}}\bigg]^{1/2} \nonumber \\
\eqexp{by Lem.~\ref{lem:q_root_f}} &\leq 2N (N-1) \bigg[ 3 N^{\gamma-2} \alpha_{m^{(\gamma)}} (N-1)^{-1} \left( \sup_{y} \Big(\big(v^{(N)}\big)^2*\rho_N\Big)(y) \right) N^{-1} \alpha_{m^{(\gamma)}} \bigg]^{1/2} \nonumber \\
&\leq \sqrt{12} \left( \sup_{y} \Big(\big(v^{(N)}\big)^2*\rho_N\Big)(y) \right)^{1/2} \, N^{\gamma/2} \, \alpha_{m^{(\gamma)}}.
\end{align}
\end{proof}

\subsection{Proof of Theorems~\ref{thm:estimates_terms_alpha_dot_beta_n} and \ref{thm:estimates_terms_alpha_dot_beta_general}}\label{sec:proofs_main_theorems_gen}
\begin{proof}[Proof of Theorems \ref{thm:estimates_terms_alpha_dot_beta_n} and \ref{thm:estimates_terms_alpha_dot_beta_general}]
First, we split $v^{(N)} = v^{(N)}_{+} - v^{(N)}_{-}$, with $v^{(N)}_{+}, v^{(N)}_{-} \geq 0$. Accordingly, we have
\begin{equation}\label{alpha_plus_minus}
\partial_t \alpha_{m^{(\gamma)}}(t) = \Term_{+} - \Term_{-} \leq \left\lvert \Term_{+} \right\rvert + \left\lvert \Term_{-} \right\rvert,
\end{equation}
where $\Term_{\pm}$ refers to $(I)_{m^{(\gamma)}}, (II)_{m^{(\gamma)}}, (III)_{m^{(\gamma)}}$ from \eqref{alpha_dot_1}, \eqref{alpha_dot_2} and \eqref{alpha_dot_3} with interaction $v^{(N)}_{\pm}$. We bound $\Term_{\pm}$ separately by using Lemma~\ref{lem:estimates_terms_alpha_dot_beta}, which proves under the stated assumptions the bound
\begin{equation}
\partial_t \alpha_{m^{(\gamma)}}(t) \leq C(t) \Big( \alpha_{m^{(\gamma)}}(t) + N^{-\gamma} \Big).
\end{equation}
Applying the Gronwall Lemma gives the desired bound
\begin{equation}
\alpha_{m^{(\gamma)}}(t) \leq e^{\int_0^t C(s) ds} \, \alpha_{m^{(\gamma)}}(0) + \left( e^{\int_0^t C(s) ds} - 1 \right) N^{-\gamma}.
\end{equation}
The values of the constant $C(t)$ in Theorems \ref{thm:estimates_terms_alpha_dot_beta_n} and \ref{thm:estimates_terms_alpha_dot_beta_general} can be obtained by using the respective assumptions, together with Lemma~\ref{lem:estimates_terms_alpha_dot_beta} and \eqref{alpha_plus_minus}.
\end{proof}

\noindent\textbf{Remarks.}
\begin{enumerate}
\setcounter{enumi}{\theremarks}
\item\label{itm:exch_term_order} Following up on Remark~\ref{itm:exch} after Theorem~\ref{thm:estimates_terms_alpha_dot_beta_n}, let us consider the size of the error we make by neglecting the exchange term. Suppose that the exchange term is of $\bigO(N^{-\delta})$. It then gives an additional term $CN^{-\delta} \sqrt{\alpha_n(t)} \leq C \big( \alpha_n(t) + N^{-2\delta} \big)$ in the time derivative of $\alpha_n(t)$ (where the $\sqrt{\alpha_n(t)}$ comes from the $q_1$ in term $(I)_n$).
\end{enumerate}
\setcounter{remarks}{\theenumi}

\section{Proof of Results for Density $\propto 1$ Regime}\label{sec:mean-field_scalings_general}
In this section, $C$ denotes a constant which can be different from line to line.

\subsection{Kinetic Energy Inequalities}\label{sec:energy inequalities}
Let us first recall two well-known inequalities which we use in Section~\ref{sec:proofs_main_theorems} to show that the conditions of Theorems \ref{thm:estimates_terms_alpha_dot_beta_n} and \ref{thm:estimates_terms_alpha_dot_beta_general} hold if the total kinetic energy is bounded by $AN$. A general version of the Lieb-Thirring or kinetic energy inequality \cite{lieb:1975,lieb:2010,rumin:2011} for orthonormal $\varphi_1,\ldots,\varphi_N \in L^2(\RRR^3)$ is
\begin{equation}\label{Lieb_Thirring}
\int_{\RRR^3} \rho_N(x)^{1+2a/3} \, d^3x \leq C_a \sum_{i=1}^N \norm{\nabla^{a} \varphi_i}^2,
\end{equation}
for any $a>0$, where $\rho_N = \sum_{i=1}^N |\varphi_i|^2$. The Hardy-Littlewood-Sobolev inequality (see, e.g., \cite[Thm.~4.3]{liebloss:2001}) in three dimensions states that for $f \in L^p(\RRR^3)$, $h \in L^r(\RRR^3)$, $p,r > 1$ and $0<\lambda<3$ with $1/p + \lambda/3 + 1/r = 2$, there is a constant $C=C(\lambda,p)$ such that
\begin{equation}\label{Hardy_Littlewood}
\left| \int_{\RRR^3} \int_{\RRR^3} f(x) |x-y|^{-\lambda} h(y) \, d^3x \, d^3y \right| \leq C \norm[p]{f} \norm[r]{h}.
\end{equation}

Let us now show for interactions $|x|^{-s}$, $0<s<6/5$, that the mean-field term $v^{(N)}*\rho_N$ is bounded independent of $N$, if it is scaled with $\beta=1-s/3$ and if the total kinetic energy is bounded by $AN$. We need this statement for the proofs in Section~\ref{sec:proofs_main_theorems}.

\begin{lemma}\label{lem:scaling_x-s}
Let $\varphi_1,\ldots,\varphi_N \in L^2(\RRR^3)$ be orthonormal. We assume that
\begin{equation}\label{ass_kin_energy}
\sum_{i=1}^N \norm{\nabla \varphi_i}^2 \leq AN
\end{equation}
for some $A>0$. Let $v^{(N)}(x) = N^{-\beta} \, |x|^{-s}$ with $\beta = 1 - s/3$, $0<s<6/5$. We set $\rho_N = \sum_{i=1}^N |\varphi_i|^2$. Then there is a constant $0<C\propto A^{s/2}$ (independent of $N$, dependent on $s$) such that
\begin{equation}\label{scaling_x-s_term_1}
\left( v^{(N)}*\rho_N \right)(y) \leq C \quad \forall y \in \RRR^3.
\end{equation}
\end{lemma}

\begin{proof}
First, note that for $0 < s < 6/5$,
\begin{equation}\label{int_2_5}
\left( \int_{B_R(0)} |x|^{-5s/2} \,d^3x \right)^{2/5} = \left( 4\pi \int_{B_R(0)} r^{-5s/2} \,r^2 dr \right)^{2/5} = \left(\frac{4\pi}{3-\frac{5}{2}s}\right)^{\frac{2}{5}} \, R^{6/5-s}.
\end{equation}
Then, using H\"older's inequality, \eqref{Lieb_Thirring}, $\int \rho_N = N$ and \eqref{int_2_5}, we find for any $R>0$,
\begin{align}\label{rho_x-1_calc}
\int_{\RRR^3} \frac{\rho_N(x)}{|x-y|^s} \,d^3x &= \int_{B_{R}(y)} \frac{\rho_N(x)}{|x-y|^s} \,d^3x + \int_{B_{R}(y)^c} \frac{\rho_N(x)}{|x-y|^s} \,d^3x \nonumber \\
&\leq \left( \int_{B_{R}(y)} \rho_N(x)^{5/3} \,d^3x \right)^{3/5} \left( \int_{B_{R}(y)} |x-y|^{-5s/2} \,d^3x \right)^{2/5} \nonumber \\*
&\quad\quad + \left( \int_{B_{R}(y)^c} \rho_N(x) \,d^3x \right) \left( \sup_{x \in B_{R}(y)^c} |x-y|^{-s} \right) \nonumber \\
&\leq C N^{3/5} R^{6/5-s} + N R^{-s}.
\end{align}
Setting $R=N^{1/3}$ (if we set $R=N^{\delta}$ and then optimize \eqref{rho_x-1_calc} with respect to $\delta$ we find $\delta=1/3$) we find
\begin{equation}\label{rho_x-1}
\int_{\RRR^3} \frac{\rho_N(x)}{|x-y|^s} \,d^3x \leq C N^{1-s/3}.
\end{equation}
Using the explicit value $C_1=\frac{5}{9} (2\pi)^{-2/3}$ for the constant from \eqref{Lieb_Thirring} with $a=1$, \eqref{int_2_5}, setting $R_N=rN^{1/3}$, with $N$-independent $r>0$, and minimizing the resulting expression \eqref{rho_x-1_calc} with respect to $r$ gives an explicit value for the constant of \eqref{rho_x-1},
\begin{equation}\label{rho_x-1_constant}
C = \left( \frac{6}{5}-s \right)^{s/2-1} s^{-5s/6} \left(\frac{6}{5}\right) 2^{2s/3} \, 3^{-s} \, 5^{s/6} ~ A^{s/2}.
\end{equation}
\end{proof}

\subsection{Proof of the Results}\label{sec:proofs_main_theorems}
\begin{proof}[Proof of Theorem~\ref{thm:E_kin_only}]
We consider the three different interactions separately. We prove here that the assumptions of Theorems~\ref{thm:estimates_terms_alpha_dot_beta_n} and \ref{thm:estimates_terms_alpha_dot_beta_general} hold in the different cases. Then this result can directly be expressed in term of reduced density matrices by using Lemmas~\ref{lem:density_conv} and \ref{lem:density_conv_alpha_f}.
\begin{enumerate}[(a)]

\item Let $v_s(x) = \pm |x|^{-s}$, with $0<s<3/5$ and $\beta=1-s/3$ and note that $v_s^2 = |v_{2s}|$. We can therefore use Lemma~\ref{lem:scaling_x-s} to show that the condition \eqref{alpha_dot_n_ass_2} from Theorem~\ref{thm:estimates_terms_alpha_dot_beta_n} holds. We find
\begin{equation}\label{s_2s}
\Big(\big(v_s^{(N)}\big)^2*\rho_N^t\Big)(y) = N^{-2\left(1-s/3\right)} \Big(v_s^2*\rho_N^t\Big)(y) \leq N^{-2\left(1-s/3\right)} \, CN^{\left(1-2s/3\right)} = C N^{-1}.
\end{equation}
If we use that the constant in \eqref{scaling_x-s_term_1} is proportional to $A^{s/2}$, we find that the constant in \eqref{s_2s} is proportional to $A^{s}$ and thus the $C$ appearing in the $\alpha_n$-estimate \eqref{main_alpha_ineq_n_applied1} is proportional to $A^{s/2}$.

\item For interactions $v = \pm v_{s,\delta}$ with
\begin{align}
0 \leq v_{s,\delta}(x) \left\{\begin{array}{cl} \leq DN^{\delta s} &, \, \text{for } |x|\leq N^{-\delta} \\ =|x|^{-s} &, \, \text{for } |x|>N^{-\delta} , \end{array}\right.
\end{align}
with $D>0$, $0<s<6/5$, $\beta = 1 - s/3$ and $\delta < (3-2s)/(6s)$ we use Theorem~\ref{thm:estimates_terms_alpha_dot_beta_general} with $\Omega_N = \emptyset$. We thus have to verify Assumption~\ref{ass:for_main_thm}. By using Lemma~\ref{lem:scaling_x-s} we find
\begin{equation}
\Big(\big(v_{s,\varepsilon}^{(N)}\big)^2*\rho_N^t\Big)(y) \leq N^{-\left(1-s/3\right)} \Big( \sup_y v_{s,\varepsilon}(y) \Big) \Big(v_{s,\varepsilon}^{(N)}*\rho_N^t\Big)(y) \leq C N^{-\left(1-s/3\right)+\delta s},
\end{equation}
i.e., \eqref{alpha_dot_m_ass_2} holds for all $\gamma \leq \beta - \delta s$. In order to show that \eqref{alpha_dot_m_ass_3} holds, we use the Hardy-Littlewood-Sobolev inequality \eqref{Hardy_Littlewood}. Note that from $\int \rho_N^t = N$ and $\int (\rho_N^t)^{5/3} \leq C N$ it follows that $\int (\rho_N^t)^p \leq C N$ for all $1 \leq p \leq 5/3$. For $\lambda=2s$ we have $p=(1-s/3)^{-1}$ and, since $0<s<6/5$, we find $1 < p < 5/3$, so that
\begin{equation}\label{HLS_applied}
\int_{\RRR^3} \int_{\RRR^3} \frac{\rho_N^t(x)\rho_N^t(y)}{|x-y|^{2s}} \,d^3x \,d^3y \leq C \norm[p]{\rho_N^t}^2 = C \left( \int (\rho_N^t)^p \right)^{2/p} \leq C N^{2/p} = C N^{2\left(1-s/3\right)},
\end{equation}
i.e., since $v_{s,\varepsilon}(x) \leq D|x|^{-s}$, \eqref{alpha_dot_m_ass_3} is satisfied. Furthermore,
\begin{equation}
\sup_{y \in \RRR^3} v_{s,\varepsilon}^{(N)}(y) \leq C \, N^{-\left(1-s/3\right)+\delta s},
\end{equation}
i.e., \eqref{alpha_dot_m_ass_5} holds if $\beta - \delta s \geq 1/2 + \gamma/2$. Therefore, the desired bound \eqref{main_alpha_ineq_m_applied1} holds for all $\gamma \leq 2\beta-1-2\delta s=1-2s/3 - 2\delta s$.

\item For $v(x) = \pm |x|^{-1}$, we use Theorem~\ref{thm:estimates_terms_alpha_dot_beta_n}, i.e., we have to verify \eqref{alpha_dot_n_ass_2}. By H\"older's inequality we find for any $R>0$,
\begin{align}
\left( \Big(v^{(N)}\Big)^2*\rho_N^t \right)(y) &= N^{-4/3} \left( \int_{B_R(y)} |x-y|^{-2} \rho_N^t(x)\,d^3x + \int_{B_R(y)^c} |x-y|^{-2} \rho_N^t(x)\,d^3x \right) \nonumber \\
&\leq N^{-4/3} \left( \int_{B_R(0)} |x|^{-2p} \,d^3x \right)^{1/p} \left( \int_{\RRR^3} \left(\rho_N^t\right)^q \right)^{1/q} + \left( \sup_{x \in B_R(0)^c} |x|^{-2} \right) \norm[1]{\rho_N^t}.
\end{align}
Now $\int_{B_R(0)} |x|^{-2p} \,d^3x \leq C R^{3-2p}$ for $p < 3/2$, i.e., $q>3$. For $q < \infty$ we use the kinetic energy inequality \eqref{Lieb_Thirring} with $a=3(q-1)/2>3$ and for $q=\infty$ we use $\norm[\infty]{\rho_N^t} \leq C$. Then we find for $R=N^{1/3}$ (recall $1/p + 1/q = 1$),
\begin{align}
\left( \Big(v^{(N)}\Big)^2*\rho_N^t \right)(y) \leq C N^{-4/3} \left( R^{3/p-2} N^{1/q} + R^{-2} N \right) \leq C N^{-1}.
\end{align}
\end{enumerate}
\end{proof}

\begin{proof}[Proof of Proposition~\ref{pro:coulomb_free_N_epsilon}]
\begin{enumerate}[(a)]
\item We start from the expression \eqref{alpha_derivative_n_remark} which yields, by using $p_m+q_m = 1$, $V_1^{(N)} = 0$ and $v_{12}^{(N)} = N^{-\delta}v_{12}$,
\begin{equation}\label{proof_prop_E_kin}
\partial_t \alpha_n(t) = 2(N-1)N^{-\delta} \, \Im\, \bigSCP{\psi^t}{v_{12} p_1^t \psi^t}.
\end{equation}
By energy conservation we have $\SCP{\psi^t}{(-\Delta_1)\psi^t} \leq EN^{-1}$. Using the many-particle Hardy inequality for fermions from \cite{frank:fermions,hoffmann-ostenhof:2008}, we find
\begin{equation}\label{calc_proof_prop}
\partial_t \alpha_n(t) \leq 2 N^{1-\delta} \SCP{\psi^t}{v_{12}^2 \psi^t}^{1/2} \leq C N^{1-\delta} \left( N^{-2/3}EN^{-1} \right)^{1/2} \leq CN^{1/6-\delta} E^{1/2}.
\end{equation}
Integration and application of Lemma~\ref{lem:density_conv} yields the desired bound \eqref{proposition_result}.

\item Taking the time derivative of $\alpha_n(\psi^t,\varphi_1,\ldots,\varphi_N)$ and then using Cauchy-Schwarz, the same steps as in \eqref{calc_proof_prop}, and \eqref{proof_prop_E_kin}, we find
\begin{align}
\partial_t \alpha_n(\psi^t,\varphi_1,\ldots,\varphi_N) &= 2 \Im \bigSCP{\psi^t}{(-\Delta_1)\psi^t} + 2 (N-1)N^{-\delta} \Im \bigSCP{\psi^t}{v_{12}p_1\psi^t} \nonumber \\
&\leq  2\left( \bigSCP{\psi^t}{(-\Delta_1)\psi^t}\bigSCP{\psi^t}{p_1(-\Delta_1)p_1\psi^t}\right)^{1/2} + CN^{1/6-\delta} E^{1/2} \nonumber \\
&\leq EN^{-1} + E_{\kin}^{\mf}N^{-1} + CN^{1/6-\delta} E^{1/2}.
\end{align}
Again, integration and application of Lemma~\ref{lem:density_conv} yields the desired bound \eqref{proposition_result2}.
\end{enumerate} 
\end{proof}

\section{Proof of Results for Density $\propto N$ Regime}\label{sec:proof_sc_scaling}
Let us first state a result about the propagation of the semiclassical initial data that was obtained in \cite{benedikter:2013}. We state this result in a slightly less general form than \cite[Propostition~$3.4$]{benedikter:2013}. Note that \cite[Propostition~$3.4$]{benedikter:2013} holds also without exchange term, i.e., for the fermionic Hartree equations. Recall that $\hat{v}$ denotes the Fourier transform of $v$ and $p^t = \sum_{j=1}^N \ketbr{\varphi_j^t}$.

\begin{lemma}\label{lem:sc_prop_sc}
Let $v \in L^1(\RRR^3)$ be such that
\begin{equation}
\int d^3k \, (1+|k|^2)\, |\hat{v}(k)| < \infty.
\end{equation}
Let $p^0$ be such that
\begin{equation}
\sup_{k\in\RRR^3} (1+|k|)^{-1} \, \Big\lvert\Big\lvert \big[ p^0, e^{ik\cdot x} \big] \Big\rvert\Big\rvert_{\tr} \leq C N^{2/3},
\end{equation}
\begin{equation}
\Big\lvert\Big\lvert\big[ p^0, \nabla \big] \Big\rvert\Big\rvert_{\tr} \leq C N.
\end{equation}
Let $\varphi_1^t,\ldots,\varphi_N^t$ be solutions to the Hartree-Fock equations \eqref{outline_hartree_fock_scaled_sc_app} or the Hartree equations \eqref{hartree_scaled_sc} with initial data $\varphi_1^0,\ldots,\varphi_N^0$. Then, there exist constants $c_1,c_2>0$, only depending on $v$, such that
\begin{equation}
\sup_{k\in\RRR^3} (1+|k|)^{-1} \, \Big\lvert\Big\lvert \big[ p^t, e^{ik\cdot x} \big] \Big\rvert\Big\rvert_{\tr} \leq c_1N^{2/3} \exp(c_2|t|),
\end{equation}
\begin{equation}
\Big\lvert\Big\lvert\big[ p^t, \nabla \big] \Big\rvert\Big\rvert_{\tr} \leq c_1 N \exp(c_2|t|),
\end{equation}
for all $t \in \RRR$.
\end{lemma}

From this lemma it follows that, using $p_1q_1=0$ and \eqref{tr_HS_op_ineqs},
\begin{equation}\label{sc_prop}
\norm[\tr]{q_1^t e^{ikx} p_1^t} = \norm[\tr]{q_1^t \left[ p_1^t, e^{ikx} \right]} \leq \norm[\tr]{\left[ p_1^t, e^{ikx} \right]} \leq N^{2/3} \, Ce^{Ct} \, (1 + \lvert k \rvert). 
\end{equation}

In the following proof, we often make use of the singular value decomposition for compact operators (see, e.g., \cite[Thm.\ VI.17]{reedsimon1:1980}). We state this decomposition for later reference in a separate lemma.

\begin{lemma}[Singular value decomposition]\label{lem:sing_value}
Let $A$ be a compact operator on a Hilbert space $\Hilbert$. Then there exist (not necessarily complete) orthonormal sets $\{ \phi_{\ell}\}_{\ell \in \NNN}$ and $\{ \tilde{\phi}_{\ell}\}_{\ell \in \NNN}$ and positive real numbers $\mu_{\ell}$ such that
\begin{equation}
A = \sum_{\ell} \mu_{\ell} \ketbra{\phi_{\ell}}{\tilde{\phi}_{\ell}}.
\end{equation}
The singular values $\mu_{\ell}$ are the eigenvalues of $|A|$, such that in particular $\norm[\tr]{A} = \sum_{\ell} \mu_{\ell}$.
\end{lemma}

\begin{proof}[Proof of Theorem~\ref{thm:sc_main_thm}]
The strategy of the proof is again to bound
\begin{equation}\label{gronwall_sc}
\left| \partial_t \alpha_n(t) \right| \leq Ce^{Ct} \big(\alpha_n(t) + N^{-1}\big)
\end{equation}
and use the Gronwall Lemma and then Lemmas~\ref{lem:density_conv} and \ref{lem:density_conv_alpha_f} to conclude the desired bound \eqref{main_alpha_ineq_sc}. Recall that we use the weight function $n(k)=k/N$ here, i.e., we can use the form \eqref{alpha_derivative_n_remark} for the time derivative of $\alpha_n(t)$. Using the scaling $v^{(N)}=N^{-1}v$ and noting the additional $N^{-1/3}$ in front of the time derivatives in the Schr\"odinger and mean-field equations, we find that $\partial_t \alpha_n(t)$ is given by the sum of the three terms
\begin{align}
(I)_n &= 2 N^{-2/3} \, \Im\, \bigSCP{\psi^t}{q_1^t\Big( (N-1)p_2^tv_{12}p_2^t - V_1 \Big) p_1^t \psi^t}, \nonumber \\
(II)_n &= 2 N^{-2/3} \, \Im\, \bigSCP{\psi^t}{q_1^tq_2^t\, (N-1)v_{12} \,p_1^tp_2^t \psi^t}, \nonumber \\
(III)_n &= 2 N^{-2/3} \, \Im\, \bigSCP{\psi^t}{q_1^tq_2^t\, (N-1)v_{12} \,p_1^tq_2^t \psi^t},
\end{align}
with $V_1 = V_1^{\dir}$ in the case of the fermionic Hartree equations, and $V_1 = V_1^{\dir} + V_1^{\exch}$ in the case of the Hartree-Fock equations. For ease of notation we write $\psi^t = \psi$, $\varphi_i^t = \varphi_i$ in the following. The double exponential in \eqref{main_alpha_ineq_sc} comes from one exponential in Lemma~\ref{lem:sc_prop_sc} and another exponential from the Gronwall Lemma applied to \eqref{gronwall_sc}. In the following estimates, we decompose $v(x) = \int d^3k \, \hat{v}(k) e^{ikx}$. Note that the assumption $\int d^3k \, (1+|k|^2)\, |\hat{v}(k)| < \infty$ in particular implies that $\int d^3k \, |\hat{v}(k)| < \infty$ and $\int d^3k \, |k| \, |\hat{v}(k)| < \infty$.

\textbf{Term} $\boldsymbol{(I)_n.}$ Let us first bound the contribution from the exchange term. Using the Fourier decomposition of $v$ and Cauchy-Schwarz we find
\begin{align}\label{V_1_exch}
N^{-2/3} \, \Big\lvert \bigSCP{\psi}{q_1 V^{\exch}_1 p_1 \psi} \Big\rvert &= N^{-2/3} \, \Big\lvert \bigSCP{\psi}{q_1 \sum_{j,\ell=1}^N \big(v_{12} * (\varphi_{\ell}^*\varphi_j)\big)(x_1) \ketbra{\varphi_{\ell}}{\varphi_j}_1 \psi} \Big\rvert \nonumber \\
&= N^{-2/3} \Big\lvert \int d^3k \, \hat{v}(k) \, \bigSCP{\psi}{q_1 e^{ikx_1} \sum_{j,\ell=1}^N \scp{\varphi_{\ell}}{e^{-ikx} \varphi_j} \ketbra{\varphi_{\ell}}{\varphi_j}_1 \psi} \Big\rvert \nonumber \\
&= N^{-2/3} \Big\lvert \int d^3k \, \hat{v}(k) \, \bigSCP{\psi}{q_1 e^{ikx_1} p_1 e^{-ikx_1} p_1 \psi} \Big\rvert \nonumber \\
&\leq N^{-2/3} \int d^3k \, \lvert\hat{v}(k)\rvert  \, \norm{q_1\psi} \nonumber \\
&\leq C N^{-2/3} \sqrt{\alpha_n} \nonumber \\
&\leq C \Big( \alpha_n + N^{-4/3} \Big).
\end{align}
Here we see explicitly that the contribution from the exchange term is of lower order in $N$. Let us now bound $(I)_n$ only with direct interaction. Using the Fourier decomposition of $v$ we find
\begin{align}\label{sc_term1_1}
&N^{-2/3} \bigSCP{\psi}{q_1\Big( (N-1)p_2v_{12}p_2 - V^{\dir}_1 \Big) p_1 \psi} \nonumber \\*
&\quad = N^{-2/3} \int d^3k \, \hat{v}(k) \bigSCP{\psi}{\Big( (N-1) p_2 e^{-ikx_2} p_2 - \sum_{j=1}^N \scp{\varphi_j}{e^{-ikx}\varphi_j} \Big) q_1 e^{ikx_1} p_1 \psi}.
\end{align}
Similar to Lemma~\ref{lem:estimates_terms_alpha_dot_beta} we would like to diagonalize the operator $p_2e^{-ikx_2}p_2$. However, since it is not self-adjoint, we decompose $e^{-ikx} = \cos(kx) - i \sin(kx)$ and diagonalize the self-adjoint operators
\begin{equation}\label{sc_diag}
p_2 \cos(kx_2) p_2 = \sum_{j=1}^N \lambda_j p_2^{\chi_j}, \quad\quad p_2 \sin(kx_2) p_2 = \sum_{j=1}^N \tilde{\lambda}_j p_2^{\tilde{\chi}_j},
\end{equation}
where the real eigenvalues $\lambda_j, \tilde{\lambda}_j$ and orthonormal eigenvectors $\chi_j, \tilde{\chi}_j$ depend on $k$. Note that $\lambda_j = \scp{\chi_j}{\cos(kx)\chi_j}$, so $|\lambda_j| \leq \norm{\chi_j}^2 = 1$ and $\sum_{j=1}^N \lambda_j  = \sum_{j=1}^N \scp{\chi_j}{\cos(kx) \chi_j} = \sum_{j=1}^N \scp{\varphi_j}{\cos(kx) \varphi_j}$, and analogous for $\tilde{\lambda}_j$. In the following, we use the projector $q^{\chi_j}_{\neq 1} = 1 - \sum_{m=2}^N p_m^{\chi_j}$ introduced in \eqref{q_neq_1} and the singular value decomposition $q_1e^{ikx_1}p_1 = \sum_{\ell} \mu_{\ell}\ketbra{\phi_{\ell}}{\tilde{\phi}_{\ell}}_1$ (see Lemma~\ref{lem:sing_value}), with  $\sum_{\ell} \mu_{\ell} = \norm[\tr]{q_1e^{ikx_1}p_1}$. Let us now decompose term $(I)_n$ by using $e^{-ikx} = \cos(kx) - i \sin(kx)$. For the $\cos$-term we find, using the antisymmetry of $\psi$ and Cauchy-Schwarz,
\begin{align}\label{sc_term1_part1}
&N^{-2/3} \, \bigg\lvert \int d^3k \, \hat{v}(k) \bigSCP{\psi}{\Big( (N-1) p_2 \cos(kx_2) p_2 - \sum_{j=1}^N \scp{\varphi_j}{\cos(kx) \varphi_j} \Big) q_1 e^{ikx_1} p_1 \psi} \bigg\rvert \nonumber \\
&\qquad = N^{-2/3} \, \bigg\lvert \int d^3k \, \hat{v}(k) \sum_{j=1}^N \lambda_j \bigSCP{\psi}{\Big( (N-1) p_2^{\chi_j} - 1 \Big) q_1 e^{ikx_1} p_1 \psi} \bigg\rvert \nonumber \\
&\qquad = N^{-2/3} \, \bigg\lvert \int d^3k \, \hat{v}(k) \sum_{j=1}^N \lambda_j \sum_{\ell} \mu_{\ell} \bigSCP{\psi}{ q^{\chi_j}_{\neq 1} \ketbra{\phi_{\ell}}{\tilde{\phi}_{\ell}}_1 q^{\chi_j}_{\neq 1} \psi} \bigg\rvert \nonumber \\
&\qquad \leq N^{-2/3} \, \int d^3k \, |\hat{v}(k)| \sum_{j=1}^N |\lambda_j| \sum_{\ell} \mu_{\ell} \norm{\bra{\phi_{\ell}}_1q^{\chi_j}_{\neq 1}\psi} \norm{\bra{\tilde{\phi}_{\ell}}_1q^{\chi_j}_{\neq 1}\psi} \nonumber \\
\eqexpl{by \eqref{p_neq_1}} &\qquad\leq N^{-2/3} \, \int d^3k \, |\hat{v}(k)| \sum_{\ell} \mu_{\ell} \sqrt{\bigSCP{\psi}{\ketbra{\phi_{\ell}}{\phi_{\ell}}_1 (Nq_2+p_2) \psi}} \,\times \nonumber \\* 
&\qquad \quad \quad \times \sqrt{\bigSCP{\psi}{ \ketbra{\tilde{\phi}_{\ell}}{\tilde{\phi}_{\ell}}_1 (Nq_2+p_2)  \psi}} \nonumber \\
\eqexpl{by Lem.~\ref{lem:projector_norms}} &\qquad \leq N^{-2/3} \, \int d^3k \, |\hat{v}(k)| \norm[\tr]{q_1e^{ikx_1}p_1} \bigSCP{\psi}{(q_2+N^{-1}p_2) \psi} \nonumber \\ 
\eqexpl{by \eqref{sc_prop}} &\qquad \leq Ce^{Ct} \int d^3k \, |\hat{v}(k)| (1+|k|) \big( \alpha_n + N^{-1} \big) \nonumber \\
&\qquad \leq Ce^{Ct} \big( \alpha_n + N^{-1} \big).
\end{align}
The same bound holds for the $\sin$-term.

\textbf{Term} $\boldsymbol{(II)_n.}$ Similarly to Lemma~\ref{lem:estimates_terms_alpha_dot_beta}, we use the antisymmetry of $\psi$ to shift one $q$ to the right side of the scalar product, i.e.,
\begin{align}\label{sc_term_II}
\left\lvert (II)_n \right\rvert &= 2 N^{-2/3} \, \Big\lvert \Im\, \bigSCP{\psi}{q_1q_2\, (N-1)v_{12} \,p_1p_2 \psi} \Big\rvert \nonumber \\
&= 2 N^{-2/3} \, \Big\lvert \Im\, \bigSCP{\psi}{q_1 \sum_{m=2}^N q_m v_{1m} p_1p_m \psi} \Big\rvert \nonumber \\
&\leq 2 N^{-2/3} \norm{q_1 \psi} \norm{q_1 \sum_{m=2}^N q_m v_{1m} p_1p_m \psi} \nonumber \\
&\leq 2 N^{-2/3} \norm{q_1 \psi} \bigg[ N^2 \bigSCP{\psi}{q_3p_1p_2v_{12}q_1v_{13}p_1p_3q_2\psi} + N \bigSCP{\psi}{p_1p_2v_{12}q_1q_2v_{12}p_1p_2\psi} \bigg]^{1/2}.
\end{align}
For the following estimate, note that for all trace class operators $A_1$, $B_2$, we find by the singular value decomposition and Lemma~\ref{lem:projector_norms},
\begin{align}\label{tr_tr_aux}
\bigSCP{\psi}{q_3 \, A_1 B_2 \, q_3 \psi} &= \sum_{j,\ell} \mu_j \mu'_{\ell} \, \bigSCP{\psi}{q_3 \, \ketbra{\phi_j}{\tilde{\phi}_j}_1 \, \ketbra{\phi'_{\ell}}{\tilde{\phi}'_{\ell}}_2 \, q_3 \psi} \nonumber \\
&\leq C \norm[\tr]{A} \norm[\tr]{B} N^{-2} \, \bigSCP{\psi}{q_3 \psi}.
\end{align}
Using first Cauchy-Schwarz, then the Fourier decomposition of $v$ and $\big|\big| q_1e^{ikx_1}p_1 \big|\big|_{\op} \leq 1$, we find
\begin{align}
&\bigSCP{\psi}{q_3p_1p_2v_{12}q_1v_{13}p_1p_3q_2\psi} \nonumber \\
&\qquad\leq \bigSCP{\psi}{q_3p_1p_2v_{12}q_1q_2v_{12}p_1p_2q_3\psi} \nonumber \\
&\qquad= \int d^3k d^3k' \, \hat{v}(k) \hat{v}(k') \bigSCP{\psi}{q_3 \left( p_1e^{ikx_1}q_1e^{ik'x_1}p_1 \right) \left( p_2e^{-ikx_2}q_2e^{-ik'x_2}p_2 \right) q_3 \psi} \nonumber \\
\eqexpl{by \eqref{tr_tr_aux}}&\qquad\leq C N^{-2} \int d^3k d^3k' \, |\hat{v}(k)| |\hat{v}(k')| \norm[\tr]{p_1e^{ikx_1}q_1e^{ik'x_1}p_1} \norm[\tr]{p_2e^{-ikx_2}q_2e^{-ik'x_2}p_2} \bigSCP{\psi}{q_3 \psi} \nonumber \\
\eqexpl{by \eqref{tr_HS_op_ineqs}}&\qquad\leq C N^{-2} \int d^3k d^3k' \, |\hat{v}(k)| |\hat{v}(k')| \norm[\tr]{p_1e^{ikx_1}q_1} \norm[\tr]{q_2e^{-ik'x_2}p_2} \alpha_n \nonumber \\
\eqexpl{by \eqref{sc_prop}}&\qquad\leq Ce^{Ct} N^{-2/3} \alpha_n,
\end{align}
and, by doing the same calculation with $q_3$ replaced by $1$,
\begin{equation}
\bigSCP{\psi}{p_1p_2v_{12}q_1q_2v_{12}p_1p_2\psi} \leq Ce^{Ct} N^{-2/3}.
\end{equation}
Thus, continuing from \eqref{sc_term_II}, we find
\begin{equation}
\left\lvert (II)_n \right\rvert \leq Ce^{Ct} N^{-2/3} \sqrt{\alpha_n} \left[ N^2 N^{-2/3} \alpha_n + N N^{-2/3} \right]^{1/2} \leq Ce^{Ct} \Big( \alpha_n + N^{-1} \Big).
\end{equation}

\textbf{Term} $\boldsymbol{(III)_n.}$ Using the Fourier decomposition of $v$, the singular value decomposition of $q_1e^{ikx_1}p_1$ and Cauchy-Schwarz we find
\begin{align}
\Big\lvert (III)_n \Big\rvert &\leq N^{1/3} \, \Big\lvert \int d^3k \, \hat{v}(k) \bigSCP{\psi}{q_1e^{ikx_1}p_1q_2 e^{-ikx_2}q_2 \psi} \Big\rvert \nonumber \\
&= N^{1/3} \Big\lvert \int d^3k \, \hat{v}(k) \sum_{\ell} \mu_{\ell} \bigSCP{\psi}{q_2\ketbra{\phi_{\ell}}{\tilde{\phi}_{\ell}}_1 e^{-ikx_2}q_2 \psi} \Big\rvert \nonumber \\
&\leq N^{1/3} \int d^3k \, |\hat{v}(k)| \sum_{\ell} \mu_{\ell} \, \Big|\Big|\bra{\phi_{\ell}}_1 q_2\psi\Big|\Big| \norm{\bra{\tilde{\phi}_{\ell}}_1 q_2 \psi} \nonumber \\
\eqexp{by Lem.~\ref{lem:projector_norms}} &\leq N^{1/3} \int d^3k \, |\hat{v}(k)| \norm[\tr]{q_1e^{ikx_1}p_1} N^{-1} \norm{q_2\psi}^2 \nonumber \\
\eqexp{by \eqref{sc_prop}} &\leq Ce^{Ct} \alpha_n.
\end{align}
\end{proof}

\bigskip

\noindent{\it Acknowledgments.} We thank Detlef D\"urr, L\'{a}szl\'{o} Erd\H{o}s, Maximilian Jeblick, David Mitrouskas and Robert Seiringer for many helpful discussions. We thank the referee for many comments that helped to improve the presentation of the article. S.P.\ gratefully acknowledges support from Cusanuswerk. S.~P.'s research has received funding from the People Programme (Marie Curie Actions) of the European Union's Seventh Framework Programme (FP7/2007-2013) under REA grant agreement n\textdegree~291734. The article is based on one of the author's (S.P.) PhD thesis \cite{petrat:2014_phd}.

\bibliographystyle{plain}
\bibliography{references}

\end{document}